\documentclass[10pt]{article} 

\usepackage{framed,multirow}
\usepackage[margin=0.8in]{geometry}

\usepackage{amssymb}
\usepackage{latexsym}

\usepackage{url}
\usepackage{xcolor}
\definecolor{newcolor}{rgb}{.8,.349,.1}
\usepackage[toc]{appendix}
\usepackage{graphicx}

\usepackage{cite}
\usepackage{sectsty}

\usepackage{authblk}

\sectionfont{\fontsize{12}{15}\selectfont}
\usepackage[integrals]{wasysym}
\usepackage{amsmath}
\usepackage{upgreek}
\usepackage{bm}
\usepackage{amsthm}
\usepackage{enumitem}
\usepackage{ulem}
\newtheorem{theorem}{Theorem}[section]
\newtheorem{corollary}{Corollary}[theorem]
\newtheorem{lemma}[theorem]{Lemma}
\theoremstyle{definition}
\newtheorem{definition}{Definition}[section]
\newcommand{\dt}{\Delta t}
\newcommand{\dx}{\Delta x}
\newcommand{\dy}{\Delta y}
\newcommand{\dz}{\Delta z}

\newcommand{\diagon}[1]{\mathbf \Lambda_{#1}}

\newcommand{\volSec}{\diagon{V}''}
\newcommand{\areaSec}{\diagon{S}''}
\newcommand{\lenPrim}{\diagon{l}'}
\newcommand{\Lbot}{\mathbf L}
\newcommand{\ndot}{\diagon{(\hat n\cdot)}}
\newcommand{\areaSecBndry}{\diagon{S,b}''}
\newcommand{\HH}{\mathbf H}
\newcommand{\PSI}{\bm {\uppsi}}
\newcommand{\D}{\mathbf D}
\newcommand{\I}{\mathbf I}
\newcommand{\W}{\mathbf W}

\newcommand{\J}{\mathbf J}
\newcommand{\Pbig}{\bm{\mathcal P}}

\newcommand{\PSIbig}{\bm{\psi}}

\begin{document}

\title{Conservation properties of a leapfrog finite-difference time-domain method for the Schr\"odinger equation}

\author[1]{Fadime Bekmambetova}
\author[1]{Piero Triverio\footnote{E-mail: piero.triverio@utoronto.ca}}
%
\affil[1]{The Edward S. Rogers Sr. Department
	of Electrical \& Computer Engineering, University of Toronto}

\maketitle

\begin{abstract}
We study the probability and energy conservation properties of a leap-frog finite-difference time-domain (FDTD) method for solving the Schr\"odinger equation. We propose expressions for the total numerical probability and energy contained in a region, and for the flux of probability current and power through its boundary. We show that the proposed expressions satisfy the conservation of probability and energy under suitable conditions. We demonstrate their connection to the Courant-Friedrichs-Lewy condition for stability. We argue that these findings can be used for developing a modular framework for stability analysis in advanced algorithms based on FDTD for solving the Schr\"odinger equation.

\end{abstract}

{
	
	\vspace{0.5cm}

}

\section{Introduction} 

The finite-difference time-domain (FDTD) algorithm is a popular numerical method for solving Maxwell's equations~\cite{taflove-computational-3rd-2005}. The leap-frog FDTD approach has also been proposed for solving the Schr\"odinger equation~\cite{visscher-1991-cp, sullivan-2000-book-em-sim-fdtd-for-other-types-sim,soriano-analysis-2004}. 
This scheme, which we will refer to as quantum FDTD (FDTD-Q), has since been used in various applications. For example, in \cite{sullivan-qd-2001} it was applied to model two electrons in a quantum dot,
in \cite{sullivan-2005-jap-efncs-nanoscale} it formed the core of the method for determining the eigenstates of arbitrary nanoscale structures, in \cite{zhao-2017-sr-antireflect} it was used for studying anti-reflective coating models,
and in \cite{castellanos-jaramillo-2018-ejp-diffraction} for simulating an electron diffraction through a double slit. Different properties of the FDTD-Q method have been studied, such as accuracy and stability~\cite{visscher-1991-cp,soriano-analysis-2004,dai-stability-2005,nagel-2009-aces}. Other time-domain techniques for the Schr\"odinger equation that are based on finite differences include non-leap-frog implicit~\cite{visscher-1991-cp, subasi-2002-nmpde} and explicit~\cite{zhidong-2009-js-abc,ryu-fdtd-quantum-paper-2016} approaches, as well as higher order methods~\cite{decleer-2021-jcam,shen-2013-cpc-ho-sfdtd}.

In the continuous domain, solutions of the Schr\"odinger equation respect the principle of conservation for both probability and energy. When one considers the entire space, the total probability of finding a particle must be constant and, with proper normalization, equal to one~\cite{miller-2008}. Similarly, the amount of energy (or more precisely the expectation value of energy) associated with the particle in a time-invariant potential must stay constant~\cite{miller-2008,chaus-1992-umj-energy-flux-quantum,hong-2006-anm}. The probability and energy would also stay constant when one considers a region that is isolated from the surrounding space, for example an infinite potential well~\cite{miller-2008,chaus-1992-umj-energy-flux-quantum}.
In general, when the Schr\"odinger equation is discretized, the conservation properties are not guaranteed to be preserved. Proving that a discretization method conserves probability and energy can be done by finding discrete counterparts of these quantities and demonstrating that they remain unchanged from one time step to the next~\cite{
	kosloff-1983-jcp,
	visscher-1991-cp, 
	fei-1995-amc-conservative,
	huang-1999-jcm,
	hong-2006-anm,
	zhu-2011-anm-symplectic,
	moxley-2013-cpc,
	wang-2015-jcp,
	ertug-2016-amcs,
	feng-2019-ccp-me-conserv,
	feng-2021-siam-na-high-order}.
This task is not trivial for the case of FDTD-Q due to the staggered sampling of the real and imaginary parts of the wavefunction.
In relation to the conservation of probability, in \cite{visscher-1991-cp} two approximations were proposed for the probability density in one-dimensional FDTD-Q, which  were argued, though without detailed proof, to conserve the principle of probability conservation. Regarding energy conservation, we are not aware of any work that studies this property specifically in the context of leap-frog FDTD-Q. However, many works have investigated the conservation of energy for other time-domain methods for the linear~\cite{kosloff-1983-jcp,leforestier-1991-jcp,hong-2006-anm,zhu-2011-anm-symplectic} and nonlinear~\cite{fei-1995-amc-conservative, huang-1999-jcm,hong-2006-anm,ertug-2016-amcs,barletti-2018-amc-energy-cons-nl,feng-2019-ccp-me-conserv,feng-2021-siam-na-high-order,zhu-2011-anm-symplectic} Schrödinger equation. 
Approaches based on symplectic integrators~\cite{yoshida-1993-cmda-progress-symplectic,kong-2007-amc-symplectic} have been proposed for solving the Schr\"odinger equation~\cite{gray-1996-jcp-symplectic-schrodinger,blanes-2006-jcp-symplectic,shen-2013-cpc-ho-sfdtd,huang-2015-cma-schrodinger-symplectic,zhu-2011-anm-symplectic}.
Symplectic algorithms can be constructed to conserve energy~\cite{yoshida-1993-cmda-progress-symplectic,zhu-2011-anm-symplectic} by preserving the symplectic structure of the continuous equations. Symplectic integrators give rise to a wide class of methods with different temporal discretization, including the leap-frog approach~\cite{yoshida-1993-cmda-progress-symplectic,gray-1996-jcp-symplectic-schrodinger,blanes-2006-jcp-symplectic}. However, to the best of our knowledge, they have not been applied to analyze the conservation properties of the FDTD-Q scheme considered in this paper.

The works in the literature on conservation properties of numerical methods for the Schr\"odinger equation typically assume either zero~\cite{visscher-1991-cp,wang-2015-jcp,feng-2019-ccp-me-conserv,feng-2021-siam-na-high-order} or periodic boundary conditions~\cite{kosloff-1983-jcp,visscher-1991-cp,huang-1999-jcm,zhu-2011-anm-symplectic,ertug-2016-amcs,barletti-2018-amc-energy-cons-nl,feng-2019-ccp-me-conserv}. 
Such boundary conditions imply that the energy and probability contained in the region stay constant with time. However, there is a motivation for studying the conservation properties for the general case where probability and energy can flow from the region into the surrounding space and vice versa. For example, some simulation scenarios~\cite{sullivan-2005-jap-efncs-nanoscale,nagel-2009-aces,zhidong-2009-js-abc,decleer-2021-jcam} involve absorbing boundary conditions meant to model unbounded domains. One may also be interested in quantifying the energy and probability in a sub-region of a larger system. Examples include studies of tunneling phenomena in hydrogen transfer reactions~\cite{cabonell-kostin-1973-ijqc-tunneling} or in quantum dot potential wells~\cite{holovatsky-2016-phys-e}. Consistency with physical laws is an important accuracy criterion for numerical methods. Hence, expressions approximating the probability and energy in a sub-region should, ideally, obey the discrete counterparts of the principle of probability and energy conservation, in addition to being close in values to the analytical solutions. 

A mechanism for the numerical probability and energy to leave or enter the region introduces new challenges in the analysis of the conservation properties. In particular, one needs to (i)~quantify the rate at which the exchange of probability and energy occurs with the surrounding space and (ii)~show that this rate balances with the rate of change of probability and energy stored in the region. Moreover, one needs to (iii)~ensure that the region is unable to provide indefinite amounts of energy and probability to the surrounding space. In this study, these three challenges are systematically addressed. The work by Fei, Pérez-García, and Vázquez~\cite{fei-1995-amc-conservative} should be mentioned, as it provides an investigation of the form of the nonlinear Schrödinger equation that allows the total charge to vary with time. In~\cite{fei-1995-amc-conservative}, questions similar to (i) and (ii) are addressed for a scheme that involves a variation of the Crank-Nicholson discretization in time and centered finite-difference discretization in space.

Moreover, in view of scenarios where a region is part of a larger setup, one may wish to study the conservation properties of a numerical method without knowing a priori what that region is connected to. This approach has proven useful in facilitating stability analysis and enforcement in FDTD for the Maxwell equations~\cite{jnl-2018-fdtd-3d-dissipative,jnl-tap-2018-fdtd-mor}. The basis for using conservation arguments for stability analysis lies in the fact that even a small violation of energy or probability conservation provides the growth mechanism that can lead to numerical solutions that grow indefinitely, which constitutes instability. This connection between the conservation properties and stability has been recognized for the case of FDTD for Maxwell's equations~\cite{kung-2003-pier,edelvik2004general,jnl-2018-fdtd-3d-dissipative} but not for the case of the leap-frog FDTD-Q scheme for the Schrödinger equation. By ensuring that an FDTD region respects the conservation properties, one guarantees that this region would not contribute to instability when integrated in a larger setup. An advantage of the conservation approach is that this guarantee can be made without having knowledge of the surrounding space, which could involve a grid of different resolution~\cite{salehi-2020-jce-polar,okoniewski-3d-subgridding-1997,jnl-2018-fdtd-3d-dissipative}, a reduced order model~\cite{kulas-macromodels-2004,jnl-tap-2018-fdtd-mor}, a representation of an open boundary~\cite{sullivan-2005-jap-efncs-nanoscale,nagel-2009-aces,zhidong-2009-js-abc,decleer-2021-jcam}, or another model.

Derivations of stability conditions for FDTD-Q have been performed using approaches related to the von Neumann analysis, which involve the investigation of temporal growth of plane wave solutions~\cite{visscher-1991-cp,soriano-analysis-2004,nagel-2009-aces}. These methods are meant for simple scenarios involving constant potential and uniform discretization, where the plane wave functions are valid solutions to the discretized equations. In \cite{dai-stability-2005}, stability conditions were derived by studying the time evolution of the norm of the error between two solutions. Both non-uniform and time-varying potentials were considered, making the proofs more general than those obtained using von Neumann-type analyses. However, the derivations in \cite{dai-stability-2005} are not applicable to schemes where an FDTD-Q region is finite and is coupled to models other than a restricted set of boundary conditions. Another method involves analyzing the eigenvalues of the so-called iteration matrix or system amplification matrix. The iteration matrix method has been used in~\cite{decleer-2021-jcam} for deriving stability limits of schemes closely related to FDTD-Q. The approach could be used to analyze scenarios where an FDTD-Q region is part of a larger setup consisting of multiple parts. However, in general, the eigenvalues would need to be studied for the matrix corresponding to the entire coupled scheme~\cite{wang-2007-aces}, which can make the analysis challenging. The conservation approach to stability analysis can circumvent this issue. 
Lastly, it should be noted that the conservation argument for stability has appeared in the literature on other time-domain numerical methods for the linear and nonlinear Schrödinger equations~\cite{fei-1995-amc-conservative,hong-2006-anm,wang-2015-jcp}.

This work presents a systematic study of probability and energy conservation in FDTD-Q for an open region, extending the energy conservation and dissipativity approaches developed previously for electromagnetics~\cite{kung-2003-pier,kung-2005-mtt, edelvik2004general,jnl-2018-fdtd-3d-dissipative,lett-2022-fdtd-potentials-awpl}. The concepts in this work take root in the theory of dissipative systems~\cite{willems1972dissipative,byrnes1994losslessness}. We formulate the FDTD-Q equations for a region, introducing unknowns on the boundary~\cite{venkatarayalu2007stable,jnl-2017-fdtd-dissipative} that allow quantifying the energy and probability exchange with the space outside the boundary. We propose expressions for discrete probability and energy, as well as particle current and supplied power. 
Using these expressions, we derive the conditions for the conservation of probability and energy and reveal that they are related to the Courant-Friedrichs-Lewy (CFL) condition that is traditionally understood as a condition for ensuring stability of an isolated system~\cite{soriano-analysis-2004}. For the case of the basic FDTD-Q scheme in an isolated region, our approach can serve as an alternative derivation of the CFL limit, which we illustrate in the paper. Moreover, in contrast to the traditional approaches of stability analysis~\cite{visscher-1991-cp,soriano-analysis-2004,nagel-2009-aces,dai-stability-2005, decleer-2021-jcam}, the conservation approach allows making conclusions on whether the region is capable of destabilizing a simulation, prior to having any knowledge of how the region is terminated or what model is used to describe the space outside the region's boundary~\cite{jnl-2018-fdtd-3d-dissipative}. 
Lastly, we verify that the discrete expressions serve as an accurate approximation of their continuous counterparts, with the conservation properties being an obvious advantage over other possible expressions. 

This paper is organized as follows. Section~\ref{sec:background} provides background information on the leap-frog FDTD-Q method in the literature. Section~\ref{sec:equations-rgn} describes the equations for the region, with modifications on the boundary to allow the probability and energy to travel through the boundary. Sections~\ref{sec:probability-conserv} and \ref{sec:energy-conserv}, respectively, analyze the discrete conservation of probability and energy for the region. Section~\ref{sec:stability} discusses how the proposed theory could be used for stability analysis and enforcement. Section~\ref{sec:ne} provides numerical examples and Section~\ref{sec:conslusions} concludes the paper.

\section{Background}
\label{sec:background}

This section describes the FDTD-Q method~\cite{visscher-1991-cp,sullivan-2000-book-em-sim-fdtd-for-other-types-sim,soriano-analysis-2004}, which is taken as the starting point in this work.
The method solves the Schr\"odinger equation, which reads
\begin{subequations}
	\begin{equation}
		\label{eq:cont-a}
		\hbar \frac{\partial \psi_R}{\partial t} = -\frac{\hbar^2}{2m} \nabla^2 \psi_I 
		+ U\psi_I \,, 
	\end{equation}
	\begin{equation}
		\label{eq:cont-b}
		\hbar \frac{\partial \psi_I}{\partial t} = \frac{\hbar^2}{2m} \nabla^2 \psi_R
		- U \psi_R \,,
	\end{equation}
\end{subequations}
where $\psi_R(x,y,z,t)$ and $\psi_I(x,y,z,t)$ are the real and imaginary parts of the wavefunction, respectively, $m$ is the mass of the particle, $\hbar$ is the reduced Planck's constant, and $U(x,y,z)$ is the potential energy profile. 

A rectangular region divided into $n_x\times n_y \times n_z$ primary cells\footnote{The concept of primary and secondary grids comes from FDTD in electromagnetics~\cite{gedney-2011}.} with dimensions $\dx \times \dy \times \dz$, as shown in Fig.~\ref{fig:primary-secondary}. The edges of the primary cells are called primary edges, which are oriented in the $+x$, $+y$, and $+z$ directions. The nodes at the corners of the primary cells are referred to as primary nodes. The primary nodes are indexed as $(i,j,k)$ from $(1,1,1)$ to $(n_x+1, n_y+1, n_z+1)$, where $(i,j,k)$ corresponds to coordinates $x=(i-1)\Delta x$, $y=(j-1)\Delta y$, $z=(k-1) \Delta z$. The time is divided into $n_t$ uniform time steps of size $\Delta t$, with temporal index $n$ denoting $t = n\Delta t$.
Both $\psi_R$ and $\psi_I$ are sampled at the primary nodes. The real part of the wavefunction $\psi_R$ is sampled at the integer time steps $n$ in $\{0, 1, \hdots, n_t\}$ and the imaginary part $\psi_I$ is sampled at the time instances shifted by half a time step, namely $\{-0.5, 0.5, \hdots, n_t-0.5\}$. Using the centered differences to discretize the time derivatives and Laplacian operators in \eqref{eq:cont-a}--\eqref{eq:cont-b}, one obtains~\cite{visscher-1991-cp,sullivan-2000-book-em-sim-fdtd-for-other-types-sim,soriano-analysis-2004}
\begin{subequations}
	\begin{multline}
		\label{eq:fdtd-a}
		\hbar \dfrac{\psi_R|_{i,j,k}^{n+1} - \psi_R|_{i,j,k}^{n}}{\dt} 
		= -\frac{\hbar^2}{2m} 
		\Bigg[
		\dfrac{ 
			\psi_I|_{i+1,j,k}^{n+\frac12} 
			- 2 \psi_I|_{i,j,k}^{n+\frac12} 
			+ \psi_I|_{i-1,j,k}^{n+\frac12}
		}{(\Delta x)^2}\\
		+ 
		\dfrac{ 
			\psi_I|_{i,j+1,k}^{n+\frac12} 
			- 2 \psi_I|_{i,j,k}^{n+\frac12} 
			+ \psi_I|_{i,j-1,k}^{n+\frac12}
		}{(\Delta y)^2}
		+ 
		\dfrac{ 
			\psi_I|_{i,j,k+1}^{n+\frac12} 
			- 2 \psi_I|_{i,j,k}^{n+\frac12} 
			+ \psi_I|_{i,j,k-1}^{n+\frac12}
		}{(\Delta z)^2}
		\Bigg]
		+ U|_{i,j,k} \psi_I|_{i,j,k}^{n+\frac12}
	\end{multline}
	\begin{multline}
		\label{eq:fdtd-b}
		\hbar \dfrac{\psi_I|_{i,j,k}^{n+\frac12}-\psi_I|_{i,j,k}^{n-\frac12}}{\dt} = \frac{\hbar^2}{2m} 
		\Bigg[
		\dfrac{
			\psi_R|_{i+1,j,k}^n 
			- 2\psi_R|_{i,j,k}^n 
			+ \psi_R|_{i-1,j,k}^n
		}{(\dx)^2}\\
		+ 
		\dfrac{
			\psi_R|_{i,j+1,k}^n 
			- 2\psi_R|_{i,j,k}^n 
			+ \psi_R|_{i,j-1,k}^n
		}{(\dy)^2}
		+ 
		\dfrac{
			\psi_R|_{i,j,k+1}^n 
			- 2\psi_R|_{i,j,k}^n 
			+ \psi_R|_{i,j,k-1}^n
		}{(\dz)^2}
		\Bigg]\nabla^2 \psi_R
		- U|_{i,j,k} \psi_R|_{i,j,k}^n
	\end{multline}
\end{subequations}
for the nodes strictly inside the region. The superscripts in \eqref{eq:fdtd-a}--\eqref{eq:fdtd-b} represent the time instances when the quantities are sampled and the subscripts represent the indices of the primary nodes. The samples involved in~\eqref{eq:fdtd-a} are shown in Fig.~\ref{fig:hanging-vars}(a). 
The numerical solution is obtained starting from initial conditions $\psi_I|^{-0.5}$ and $\psi_R|^0$ and computing $\psi_I|^{n+0.5}$ and $\psi_R|^{n+1}$ from \eqref{eq:fdtd-b} and \eqref{eq:fdtd-a} in a leap-frog manner. In order to ensure that the scheme is stable, the time step needs to be taken below the CFL limit~\cite{dai-stability-2005}
\begin{equation}
	\label{eq:cfl}
	\dt < \dt_{\text{CFL}} = \dfrac{2}{\dfrac{2\hbar}{m}\left(
		\dfrac{1}{(\dx)^2} + \dfrac{1}{(\dy)^2} + \dfrac{1}{(\dz)^2}
		\right) + \dfrac{\max_{i,j,k}\big|U|_{i,j,k}\big|}{\hbar}} \,.
\end{equation}

The update procedure for a wavefunction at a particular node requires knowing the previous time step values corresponding to the six surrounding nodes, which are only available when performing the updates at the nodes strictly inside the region. Hence, boundary conditions need to be assumed, such as zero Dirichlet or periodic boundary conditions. The six faces of the boundary are referred to as west~(W), east~(E), south~(S), north~(N), bottom~(B), and top (T), corresponding to $i=1$, $i=n_x+1$, $j=1$, $j=n_y+1$, $k=1$, and $k=n_z+1$, respectively.

\begin{figure}
	\centering
	\includegraphics[scale=0.5]{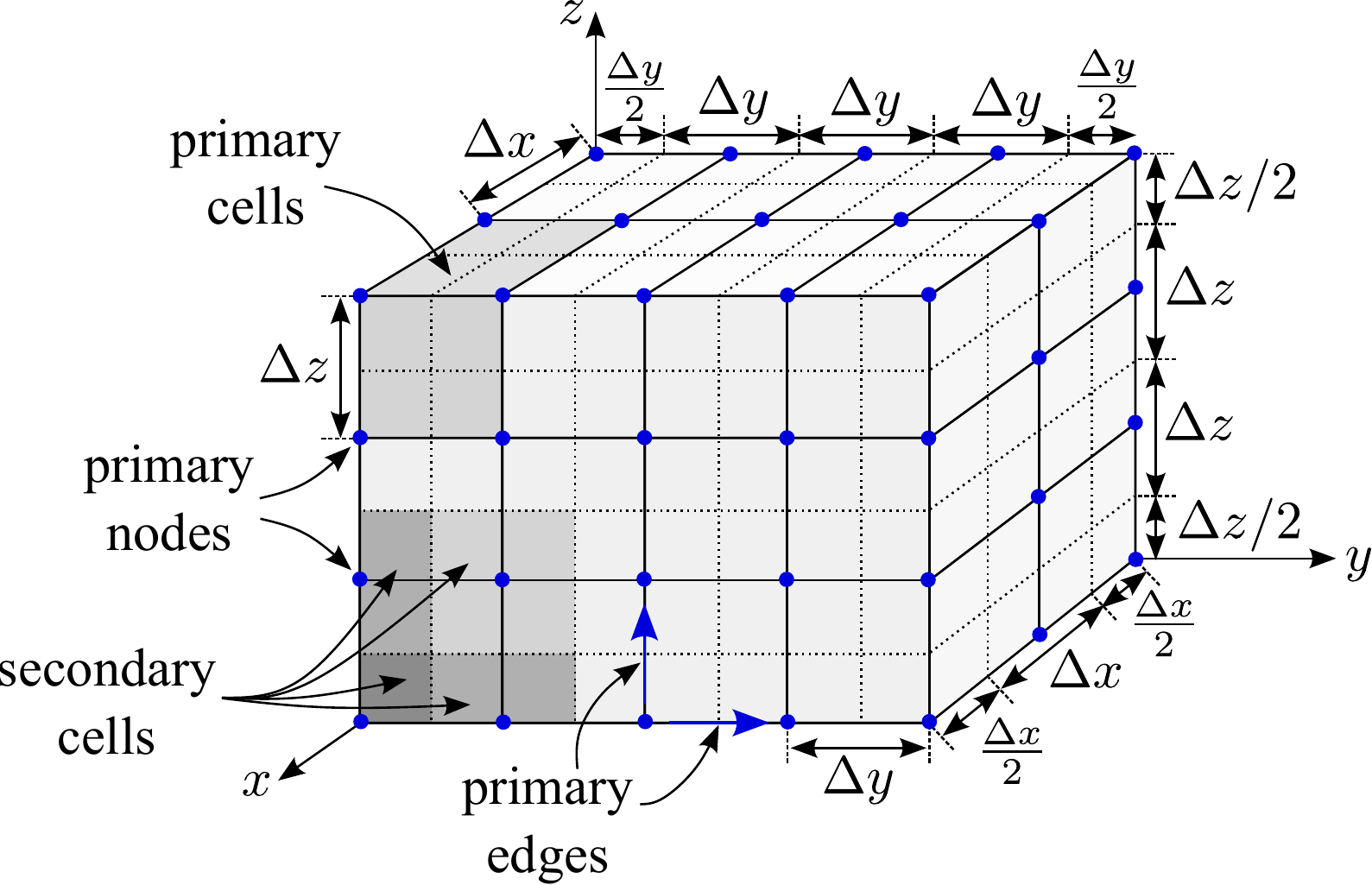}	
	\caption{Illustration of the geometrical quantities associated with the discretized region. In this example ${n_x = 2}$, ${n_y = 4}$, ${n_z = 3}$. All primary cells have dimensions ${\dx \times \dy \times \dz}$. The secondary cells strictly inside the region have dimensions ${\dx \times \dy \times \dz}$. The secondary cells adjacent to one face of the boundary are halved in the dimension normal to the face of the boundary. Similarly, the secondary cells adjacent to two or three faces of the boundary are halved in two or three dimensions, respectively.}
	\label{fig:primary-secondary}
\end{figure}
\section{Equations for the region}
\label{sec:equations-rgn}

This section presents the proposed discretization of \eqref{eq:cont-a}--\eqref{eq:cont-b}, which facilitates the investigation of probability and energy conservation in Sections \ref{sec:probability-conserv} and \ref{sec:energy-conserv}. As a starting point, we take the FDTD-Q method outlined in Section~\ref{sec:background}. The proposed equations consider a general scenario where the FDTD-Q region could either be terminated with boundary conditions or constitute a portion of a larger domain. In this scenario, the staggered nature of spatial sampling in FDTD-Q makes it difficult to precisely define the region's boundary. As a result, quantification of the probability and energy pertaining to the region becomes non-trivial. In this section, we propose equations on the boundary involving the so-called hanging variables~\cite{venkatarayalu2007stable,jnl-2018-fdtd-3d-dissipative}. These equations allow for an unambiguous separation between the region and the space outside the region's boundary. The concept of hanging variables is related to the mortar methods~\cite{maday-1989-siam}. 

\subsection{Equations at each node}
\label{sec:scalar-equations}

The discussion below shows a detailed treatment of the discrete equations corresponding to \eqref{eq:cont-a}. Equation \eqref{eq:cont-b} is treated analogously.
For the samples of $\psi_R$ strictly inside the region, we take \eqref{eq:fdtd-a}, multiplied on both sides by the factor $\dx \dy \dz$:
\begin{multline}
	\label{eq:scalar-r-internal}
	\hbar \ \dx\dy \dz \dfrac{\psi_R|_{i,j,k}^{n+1} - \psi_R|_{i,j,k}^{n}}{\dt} 
	= -\frac{\hbar^2}{2m} 
	\Bigg[
	\dy \dz\dfrac{ 
		\psi_I|_{i+1,j,k}^{n+\frac12} 
		- \psi_I|_{i,j,k}^{n+\frac12}
	}{\Delta x}
	- 
	\dy \dz\dfrac{ 
		\psi_I|_{i,j,k}^{n+\frac12} 
		- \psi_I|_{i-1,j,k}^{n+\frac12}
	}{\Delta x}		\\
	+ 
	\dx \dz\dfrac{ 
		\psi_I|_{i,j+1,k}^{n+\frac12} 
		- \psi_I|_{i,j,k}^{n+\frac12} 
	}{\Delta y}
	- 
	\dx \dz\dfrac{ 
		\psi_I|_{i,j,k}^{n+\frac12} 
		- \psi_I|_{i,j-1,k}^{n+\frac12}
	}{\Delta y}
	+ 
	\dx\dy\dfrac{ 
		\psi_I|_{i,j,k+1}^{n+\frac12} 
		- \psi_I|_{i,j,k}^{n+\frac12} 
	}{\Delta z}
	-
	\dx\dy\dfrac{ 
		\psi_I|_{i,j,k}^{n+\frac12} 
		- \psi_I|_{i,j,k-1}^{n+\frac12}
	}{\Delta z}
	\Bigg]\\
	+ \dx\dy \dz \ U|_{i,j,k} \psi_I|_{i,j,k}^{n+\frac12} \,,
\end{multline}
where the samples involved in the equation are shown in Fig.~\ref{fig:hanging-vars}(a).
The factor of $\dx \dy \dz$ is introduced for convenience, as will become clear in the subsequent derivations. This factor corresponds to the volume of the $\dx \times \dy \times \dz$ cell centered on the node $(i,j,k)$. Such a cell is referred to as a ``secondary cell''. The subdivision of the region into the secondary cells is illustrated in Fig.~\ref{fig:primary-secondary} using the dotted lines. Equation~\eqref{eq:scalar-r-internal} can be interpreted as a discretization of the integral form of~\eqref{eq:cont-a}, which reads\footnote{The concept of discretizing partial differential equations via their integral form is related to finite volume methods~\cite{moukalled-2016-book-fin-vol}.}
\begin{equation}
	\label{eq:cont-a-integral}
	\hbar \iiint_{\Delta V''} \dfrac{\partial \psi_R}{\partial t} dV
	= 
	- \dfrac{\hbar^2}{2m} \oiint_{\partial \Delta V''} \nabla \psi_I \cdot \vec{dS}
	+ 
	\iiint_{\Delta V''} U \psi_I\ dV \,,
\end{equation}
where the integrals are taken over the volume of the corresponding secondary cell $\Delta V''$ and over its boundary $\partial \Delta V''$. The first term on the right hand side of \eqref{eq:cont-a-integral} involves the flux of $\nabla \psi_I$ through the boundary of the secondary cell, and is the continuous counterpart of the term in the square brackets in \eqref{eq:scalar-r-internal}. 

\begin{figure}
	\centering
	\includegraphics[scale=1]{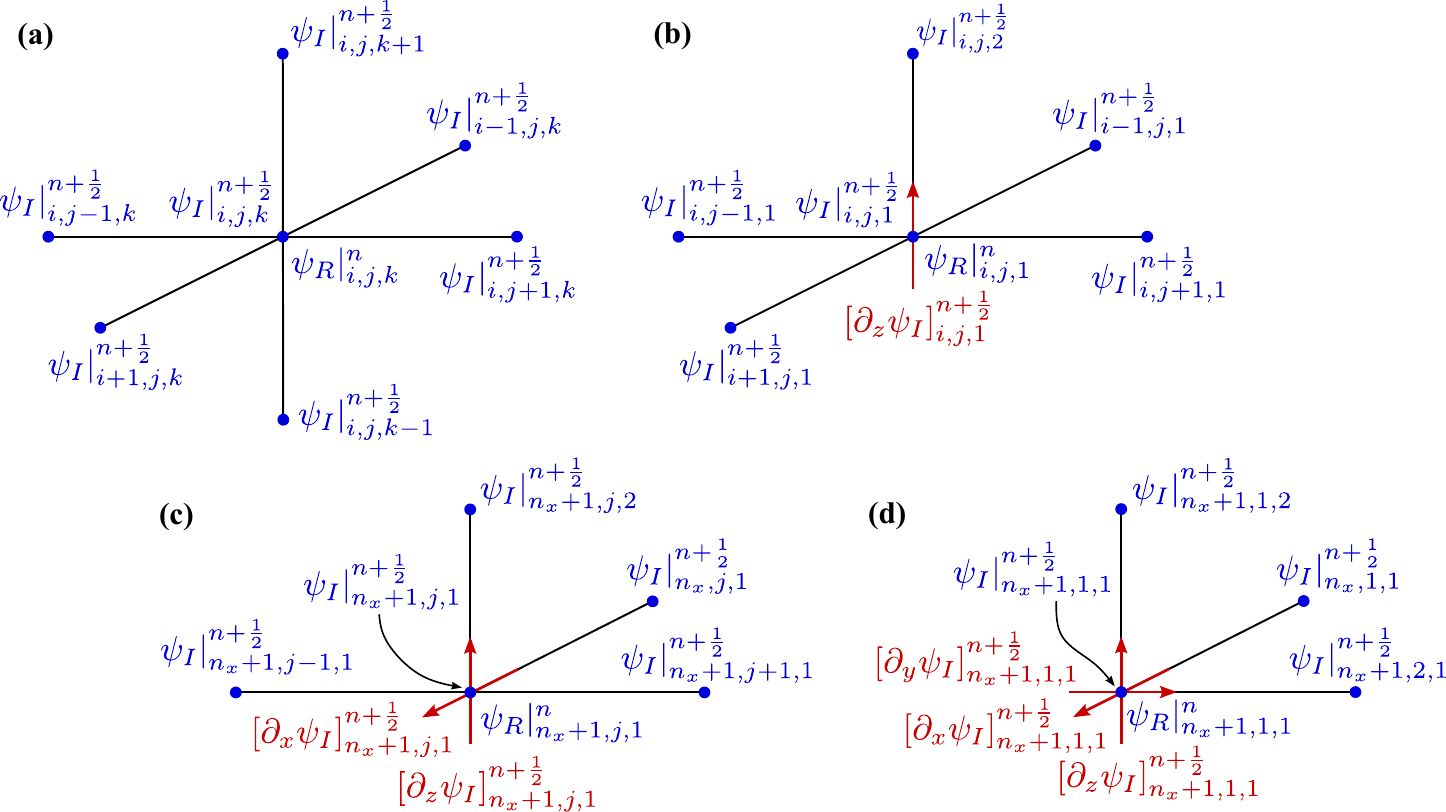}
	\caption{Samples involved in the discrete equations updating the real part of the wavefunction. The hanging variables are shown in red. (a)~Samples in \eqref{eq:scalar-r-internal} for an internal node. (b)~Samples in \eqref{eq:scalar-r-bottom} for a node on the bottom face of the boundary. (c)~Samples in \eqref{eq:scalar-r-bottom-east} for a node on the edge shared between the bottom and east faces of the boundary. (d)~Samples in \eqref{eq:scalar-r-bottom-south-east} for the node on the corner formed by the bottom, south, and east faces of the boundary.}
	\label{fig:hanging-vars}
\end{figure} 

For nodes $(i,j,1)$ on the bottom boundary of the region, equation~\eqref{eq:scalar-r-internal} would involve samples of the wavefunction that are outside the boundary, namely $\psi_I|_{i,j,k-1}$. The involvement of such samples is undesirable for two reasons. First, it would make it difficult to distinguish the energy and probability corresponding to the region from the energy and probability that should be attributed to the space outside the region. Second, the samples $\psi_I|_{i,j,k-1}$, in general, may not be available. For example, the space outside the boundary of the region may involve a grid of different resolution or an entirely different model that does not involve FDTD-Q samples. Hence, we place a hanging variable~\cite{venkatarayalu2007stable} representing the $z$-component of the gradient of $\psi_I$ on the boundary of the region. The hanging variable is shown in red in Fig.~\ref{fig:hanging-vars}(b). Using this variable, we write the discretization of \eqref{eq:cont-a-integral} over the corresponding secondary cell as
\begin{multline}
	\label{eq:scalar-r-bottom}
	\hbar \ \dx\dy \dfrac{\dz}2 \dfrac{\psi_R|_{i,j,1}^{n+1} - \psi_R|_{i,j,1}^{n}}{\dt} 
	= -\frac{\hbar^2}{2m} 
	\Bigg[
	\dy \dfrac{\dz}2\dfrac{ 
		\psi_I|_{i+1,j,1}^{n+\frac12} 
		- \psi_I|_{i,j,1}^{n+\frac12}
	}{\Delta x}
	- 
	\dy \dfrac{\dz}2\dfrac{ 	
		\psi_I|_{i,j,1}^{n+\frac12} 
		- \psi_I|_{i-1,j,1}^{n+\frac12}
	}{\Delta x}		\\
	+ 
	\dx \dfrac{\dz}2\dfrac{ 
		\psi_I|_{i,j+1,1}^{n+\frac12} 
		- \psi_I|_{i,j,1}^{n+\frac12} 
	}{\Delta y}
	- 
	\dx \dfrac{\dz}2\dfrac{ 
		\psi_I|_{i,j,1}^{n+\frac12} 
		- \psi_I|_{i,j-1,1}^{n+\frac12}
	}{\Delta y}
	+ 
	\dx\dy\dfrac{ 
		\psi_I|_{i,j,2}^{n+\frac12} 
		- \psi_I|_{i,j,1}^{n+\frac12} 
	}{\Delta z}
	-
	\dx\dy[\partial_z \psi_I]_{i,j,1}^{n+\frac12}
	\Bigg]\\
	+ \dx\dy \dfrac{\dz}2 \ U|_{i,j,1} \psi_I|_{i,j,1}^{n+\frac12} \,.
\end{multline}
The differences with \eqref{eq:scalar-r-internal} are the dimensions of the secondary cell involved (which are $\dx \times \dy \times \dz/2$ for the cells adjacent to the bottom boundary), and the introduction of the hanging variable $[\partial_z \psi_I]_{i,j,1}$. The notation $\partial_z$ indicates that the variable represents a partial derivative with respect to $z$ and the square brackets are used to distinguish the hanging variables from the finite-difference approximation of the gradient of $\psi_I$. The hanging variables can be used to couple the region with the model describing the space beyond the boundary of the region~\cite{jnl-2018-fdtd-3d-dissipative,venkatarayalu2007stable}. As will become clear in the subsequent discussion, the hanging variables will help quantify the rate at which the region exchanges the probability and energy with the surrounding space~\cite{jnl-2018-fdtd-3d-dissipative}.

For nodes on the edges of the boundary, the equations involve two hanging variables and are written over secondary cells of volume $\dx \dy \dz/4$. For example, for the nodes on the bottom-east edge of the boundary, such as the node shown in Fig.~\ref{fig:hanging-vars}(c), the proposed equation involves hanging variables $[\partial_x \psi_I]$ and $[\partial_z \psi_I]$
\begin{multline}
	\label{eq:scalar-r-bottom-east}
	\hbar \ \dfrac{\dx}2 \dy \dfrac{\dz}2 \dfrac{\psi_R|_{n_x+1,j,1}^{n+1} - \psi_R|_{n_x+1,j,1}^{n}}{\dt} 
	= -\frac{\hbar^2}{2m} 
	\Bigg[
	\dy \dfrac{\dz}2 [\partial_x \psi_I]_{n_x+1,j,1}^{n+\frac12}
	- 
	\dy \dfrac{\dz}2\dfrac{ 
		\psi_I|_{n_x+1,j,1}^{n+\frac12} 
		- \psi_I|_{n_x,j,1}^{n+\frac12}
	}{\Delta x}	\\
	+ 
	\dfrac{\dx}2\dfrac{\dz}2\dfrac{ 
		\psi_I|_{n_x+1,j+1,1}^{n+\frac12} 
		- \psi_I|_{n_x+1,j,1}^{n+\frac12} 
	}{\Delta y}
	- 
	\dfrac{\dx}2\dfrac{\dz}2\dfrac{ 
		\psi_I|_{n_x+1,j,1}^{n+\frac12} 
		- \psi_I|_{n_x+1,j-1,1}^{n+\frac12}
	}{\Delta y}\\
	+ 
	\dfrac{\dx}2\dy\dfrac{ 
		\psi_I|_{n_x+1,j,2}^{n+\frac12} 
		- \psi_I|_{n_x+1,j,1}^{n+\frac12} 
	}{\Delta z}
	-
	\dfrac{\dx}2\dy[\partial_z \psi_I]_{n_x+1,j,1}^{n+\frac12}
	\Bigg]
	+ \dfrac{\dx}2\dy \dfrac{\dz}2 \ U|_{n_x+1,j,1} \psi_I|_{n_x+1,j,1}^{n+\frac12} \,.
\end{multline}

For nodes on the region's corners, the proposed equations involve three hanging variables ($[\partial_x \psi_I]$, $[\partial_y \psi_I]$, and $[\partial_z \psi_I]$) and are written over secondary cells with dimensions $\dx/2 \times \dy/2 \times \dz / 2$. For example, on the bottom-south-east corner illustrated in Fig.~\ref{fig:hanging-vars}(d), the equation reads
\begin{multline}
	\label{eq:scalar-r-bottom-south-east}
	\hbar \ \dfrac{\dx}2 \dfrac{\dy}2 \dfrac{\dz}2 \dfrac{\psi_R|_{n_x+1,1,1}^{n+1} - \psi_R|_{n_x+1,1,1}^{n}}{\dt} 
	= -\frac{\hbar^2}{2m} 
	\Bigg[
	\dfrac{\dy}2 \dfrac{\dz}2 [\partial_x \psi_I]_{n_x+1,1,1}^{n+\frac12}
	- 
	\dfrac{\dy}2 \dfrac{\dz}2\dfrac{ 
		\psi_I|_{n_x+1,1,1}^{n+\frac12} 
		- \psi_I|_{n_x,1,1}^{n+\frac12}
	}{\Delta x}	\\
	+ 
	\dfrac{\dx}2\dfrac{\dz}2\dfrac{ 
		\psi_I|_{n_x+1,2,1}^{n+\frac12} 
		- \psi_I|_{n_x+1,1,1}^{n+\frac12} 
	}{\Delta y}
	- 
	\dfrac{\dx}2\dfrac{\dz}2
	[\partial_y \psi_I]_{n_x+1,1,1}^{n+\frac12}\\
	+ 
	\dfrac{\dx}2\dfrac{\dy}2\dfrac{ 
		\psi_I|_{n_x+1,1,2}^{n+\frac12} 
		- \psi_I|_{n_x+1,1,1}^{n+\frac12} 
	}{\Delta z}
	-
	\dfrac{\dx}2\dfrac{\dy}2[\partial_z \psi_I]_{n_x+1,1,1}^{n+\frac12}
	\Bigg]
	+ \dfrac{\dx}2\dfrac{\dy}2 \dfrac{\dz}2 \ U|_{n_x+1,1,1} \psi_I|_{n_x+1,1,1}^{n+\frac12}\,.
\end{multline}
The discretization of \eqref{eq:cont-b} is performed analogously, resulting in equations similar to \eqref{eq:scalar-r-internal}, \eqref{eq:scalar-r-bottom}, \eqref{eq:scalar-r-bottom-east}, and \eqref{eq:scalar-r-bottom-south-east}.

\subsection{Compact matrix form}
\label{sec:cmf}

In order to facilitate the subsequent derivations, we write the equations described in Section~\ref{sec:scalar-equations} in a compact matrix form. Equations corresponding to \eqref{eq:cont-a}, such as \eqref{eq:scalar-r-internal}, \eqref{eq:scalar-r-bottom}, \eqref{eq:scalar-r-bottom-east}, and \eqref{eq:scalar-r-bottom-south-east}, can be written as
\begin{equation}
	\label{eq:cmf-pre-a}
	\hbar \volSec \dfrac{\PSI_R^{n+1}- \PSI_R^n}{\dt} 
	= 
	-\dfrac{\hbar^2}{2m} \left(-\D \areaSec (\lenPrim)^{-1} \D^T\PSI_I^{n+\frac12}
	+ \Lbot \ndot \areaSecBndry [\nabla \PSI_I]_{\bot}^{n+\frac12}\right)
	+ \volSec \diagon{U}\PSI_I^{n+\frac12} 
\end{equation}
and an analogous matrix form can be written for the discrete equations corresponding to \eqref{eq:cont-b}.
In \eqref{eq:cmf-pre-a}, vectors $\PSI_R$ and $\PSI_I$ contain samples of $\psi_R$ and $\psi_I$ at the primary nodes and vector $[\nabla \PSI_I]_{\bot}$ collects the hanging variables on the boundary of the region. Matrix $\volSec$ is a diagonal matrix containing the volumes of secondary cells depicted in Fig.~\ref{fig:primary-secondary}. Diagonal matrix $\diagon{U}$ contains the values of the potential $U$ at primary nodes.
The two terms in the brackets on the right hand side of \eqref{eq:cmf-pre-a} correspond to the discrete outward flux of $\nabla \psi_I$ through the boundary of each of the secondary cells. The first term is the contribution due to the finite-difference approximation of $\nabla \psi_I$ on the primary edges. The second term is the contribution due to the hanging variables. 
The rows of matrix $\D$ correspond to the primary nodes\footnote{Equivalently, we can say that the rows of $\D$ correspond to the secondary cells.} and its columns correspond to the primary edges. For each primary edge, the respective column of $\D$ contains a $+1$ in the row corresponding to the primary node at the tail of the primary edge and a $-1$ in the row corresponding to the head of the primary edge.
Diagonal matrices $\lenPrim$ and $\areaSec$ contain, respectively, the length of the primary edges and the area of the secondary cell faces pierced by these edges. With these definitions, 
\begin{equation}\label{eq:gradPsiI}
	\nabla \PSI_I^{n+\frac12} = -(\lenPrim)^{-1}\D^T\PSI_I^{n+\frac12}
\end{equation}
is a vector containing the finite difference approximations of $\nabla \psi_I$ on each of the primary edges, and the left multiplication of this vector by $\D \areaSec$ in \eqref{eq:cmf-pre-a} computes their contribution to the outward flux values for each secondary cell. The columns of the matrix $\Lbot$ correspond to the additional boundary edges where the hanging variables are sampled. For each such column, $\Lbot$ contains a $+1$ in the row corresponding to the node collocated with the additional boundary edge. Diagonal matrix $\ndot$ has the diagonal elements equal to $-1$ for the hanging variables on the west, south, and bottom boundaries and to $+1$ for the variables on the east, north, and top boundaries. Matrix $\areaSecBndry$ contains on the diagonal the areas of the secondary cell faces pierced by the edges where the hanging variables are sampled. Similarly to $\nabla \PSI_I^{n+0.5}$ in \eqref{eq:gradPsiI}, we also define a compact notation for the vector containing the finite difference approximations of $\nabla \psi_R$, which is given by
\begin{equation}\label{eq:gradPsiR}
	\nabla \PSI_R^n = -(\lenPrim)^{-1}\D^T\PSI_R^n
\end{equation} 
and will be useful in the subsequent discussion.
The indexing convention and the corresponding matrix expressions are detailed in Appendix~\ref{appendix:indexing-and-cmf-expressions}.

Equation \eqref{eq:cmf-pre-a} and its counterpart for the imaginary part of the wavefunction can be written more compactly as
\begin{subequations}
	\begin{equation}
		\label{eq:cmf-a}
		\hbar \volSec \dfrac{\PSI_R^{n+1}- \PSI_R^n}{\dt} 
		= 
		\HH \PSI_I^{n+\frac12}
		- \HH_{\bot} [\nabla \PSI_I]_{\bot}^{n+\frac12}\,, 
		\quad \forall n = 0, 1, \hdots, n_t-1 \,,
	\end{equation}
	\begin{equation}
		\label{eq:cmf-b}
		\hbar \volSec \dfrac{\PSI_I^{n+\frac12}-\PSI_I^{n-\frac12}}{\dt} 
		= - \HH \PSI_R^n
		+ \HH_{\bot} [\nabla \PSI_R]_{\bot}^n
		\,, 
		\quad \forall n = 0, 1, \hdots, n_t-1 \,,
	\end{equation}
\end{subequations}
where
\begin{equation}\label{eq:hh}
	\HH = \dfrac{\hbar^2}{2m} \D \areaSec (\lenPrim)^{-1} \D^T+ \volSec \diagon{U} \,,
\end{equation}
\begin{equation}
	\label{eq:mtx-Hbot}
	\HH_{\bot} = 
	\frac{\hbar^2}{2m} \Lbot \ndot \areaSecBndry \,.
\end{equation}
Equations \eqref{eq:cmf-a}--\eqref{eq:cmf-b} can also be written as a single matrix equation,
\begin{equation}
	\label{eq:cmf-single}
	\hbar \Pbig
	\dfrac{\PSIbig^{n+1}- \PSIbig^n}{\dt} = 
	\left(\J_1 \otimes \HH\right)
	\dfrac{\PSIbig^{n+1} + \PSIbig^n}2 
	- (\J_1 \otimes \HH_{\bot})
	[\nabla \PSIbig]_{\bot}^{n+\frac12}\,,
	\quad
	n = 0, 1, \hdots, n_t-1 \,,
\end{equation}
where 
\begin{equation}
	\label{eq:mtx-P}
	\bm{\mathcal P} = 
	\I_2 \otimes \volSec
	- \dfrac{\dt}{2\hbar} |\J_1|\otimes \mathbf H
	= 
	\begin{bmatrix}
		\volSec & -\frac{\dt}{2\hbar} \mathbf H\\
		-\frac{\dt}{2\hbar} \mathbf H & \volSec
	\end{bmatrix} \,,
\end{equation}
\begin{equation}
	\label{eq:mtx-psibig}
	\PSIbig^n = \begin{bmatrix}
		\PSI_R^n\\ \PSI_I^{n-\frac12}
	\end{bmatrix}\,,
	\quad 
	n = 0, 1, \hdots, n_t\,,
\end{equation}
\begin{equation}
	\label{eq:mtx-nabla-psi-big}
	[\nabla \PSIbig]_{\bot}^{n+\frac12}
	= 
	\begin{bmatrix}
		[\nabla \PSI_R]_{\bot}^{n} \\ 
		[\nabla \PSI_I]_{\bot}^{n+\frac12}
	\end{bmatrix}\,,
	\quad n = 0, 1, \hdots, n_t-1\,,
\end{equation}
where $\J_m$ is a $2m \times 2m$ matrix of the form
\begin{equation}
	\label{eq:mtx-J}
	\J_m = \begin{bmatrix}
		\mathbf 0 & \I_m \\
		-\I_m & \mathbf 0
	\end{bmatrix}\,,
\end{equation}
and matrix $\I_m$ is an $m\times m$ identity matrix. 
Brackets ``$|\cdot|$'' denote an element-wise absolute value operation and ``$\otimes$'' is the Kronecker product~\cite{golub-matrix-comput-4ed}. Matrix $\Pbig$ in \eqref{eq:mtx-P} will lead to an expression for the probability of finding the particle in the region. Equation \eqref{eq:cmf-single} can be also seen as a discrete-time dynamical system~\cite{haddad-2011-stability-and-ctrl} that approximates the solution of the Schrödinger equation describing the evolution in time of the real and imaginary parts of the wavefunction. The evolution of the system depends on the values of the hanging variables on the boundary $[\nabla \PSIbig]_{\bot}^{n+0.5}$, which act as an excitation to \eqref{eq:cmf-single}. This excitation is referred to as the input of the dynamical system~\cite{haddad-2011-stability-and-ctrl}.
\section{Probability conservation}
\label{sec:probability-conserv}
In this section we propose expressions for the total probability in the region and for the probability current leaving the region through the boundary. We show that these expressions satisfy probability conservation under a condition on $\dt$, which is recognized to be a generalized CFL limit. Furthermore, we prove that the conventional CFL limit~\eqref{eq:cfl} is a sufficient condition for the probability conservation.

\subsection{Total probability and probability current}
In the continuous domain, the probability of finding a particle in a volume $V$ is given by~\cite{miller-2008}
\begin{equation}
	\label{eq:probability-continuous}
	\mathcal P(t) = \iiint_V |\psi|^2\ dV = \iiint_V (\psi_R^2 + \psi_I^2)\ dV \,.
\end{equation}
The probability can leave the region through the boundary $\partial V$ at the rate dictated by the outward probability current
\begin{equation}
	\label{eq:current-continuous}
	\mathcal I_P(t)
	\approx
	\oiint_{\partial V} \vec J_P(x,y,z,t) \cdot \hat n \ dS	\,,
\end{equation}
where $\hat n$ is the outward normal vector and $\vec J_P(x,y,z,t)$ is the probability current density, also known in the literature as the particle current density~\cite{miller-2008}
\begin{equation}
	\label{eq:current-density-approximation}
	\vec J_P 
	=
	\dfrac{i\hbar}{2m} (\psi \nabla \psi^* - \psi^*\nabla \psi)
	=
	\frac{\hbar}{m} 
	\left(	 	 
	\psi_R \nabla \psi_I
	-
	\psi_I	\nabla \psi_R
	\right)\,.
\end{equation}
In order to analyze the probability conservation properties of FDTD-Q, we define discrete expressions that approximate \eqref{eq:probability-continuous} and \eqref{eq:current-continuous}.

The most obvious approach to defining the probability associated with $\PSI^n$ in \eqref{eq:mtx-psibig} would be to directly discretize the integral in \eqref{eq:probability-continuous} with the use of $\PSI_R^n$ and $\PSI_I^{n-0.5}$, obtaining
\begin{equation}
	\label{eq:Pnaive}
	\mathcal P_{\text{simple}}^n 
	= 
	\sum_{i=1}^{n_x+1} 
	\sum_{j=1}^{n_y+1} 
	\sum_{k=1}^{n_z+1} 
	\Delta V''|_{i,j,k} \left(
	(\psi_R|_{i,j,k}^{n})^2
	+ (\psi_I|_{i,j,k}^{n-\frac12} )^2
	\right)
	=
	 (\PSI_R^{n})^T \volSec \PSI_R^n 
	+ (\PSI_I^{n-\frac12})^T \volSec \PSI_I^{n-\frac12}	\,,
\end{equation}
where $\Delta V''|_{i,j,k}$ is the volume of the secondary cell associated with the node $(i,j,k)$. 
However, as we will demonstrate in Section~\ref{sec:ne}, this expression does not respect the principle of probability conservation even for an isolated region. Instead, we propose an expression for the total probability that involves the matrix $\Pbig$, which appears in equation \eqref{eq:cmf-single} describing the time evolution of the wavefunction.

\begin{definition}[Total probability]
	\label{def:probability}
	The total probability of finding the particle in a region described by FDTD-Q equations \eqref{eq:cmf-a}--\eqref{eq:cmf-b} is given by
	\begin{equation}
		\label{eq:probability}
		\mathcal P^n = (\PSIbig^n)^T \Pbig \PSIbig^n
		\,, \quad 
		n = 0, 1, \hdots, n_t \,.
	\end{equation}
\end{definition}

To see the connection between \eqref{eq:probability} and \eqref{eq:probability-continuous}, we expand \eqref{eq:probability} using the definitions of $\Pbig$ and $\PSIbig^n$ in \eqref{eq:mtx-P} and \eqref{eq:mtx-psibig}, respectively
\begin{multline}
	\mathcal P^n 
	= (\PSI_R^{n})^T \volSec \PSI_R^n 
	+ (\PSI_I^{n-\frac12})^T \volSec \PSI_I^{n-\frac12}
	- \dfrac{\dt}{\hbar} (\PSI_I^{n-\frac12})^T \HH \PSI_R^{n}\\
	= (\PSI_R^{n})^T \volSec \PSI_R^n 
	+(\PSI_I^{n-\frac12})^T \left( \volSec \PSI_I^{n-\frac12} 
	- \dfrac{\dt}{\hbar} \HH \PSI_R^{n}\right) 
	\,, 
	\quad 
	n = 0, 1, \hdots, n_t \,,
\end{multline}
and with the use of \eqref{eq:cmf-b} obtain
\begin{equation}\label{eq:probability-rewritten-for-n}
	\mathcal P^n 
	= (\PSI_R^{n})^T \volSec \PSI_R^n 
	+ (\PSI_I^{n-\frac12})^T \volSec \PSI_I^{n+\frac12}
	- \dfrac{\dt}{\hbar} (\PSI_I^{n-\frac12})^T \HH_{\bot} [\nabla\PSI_R]_{\bot}^{n}\,, \quad 
	\forall n = 0, 1, \hdots, n_t -1\,.
\end{equation}
When $\dt$ is small, the last term in \eqref{eq:probability-rewritten-for-n} can be neglected and $\mathcal P^n$ can be approximated as
\begin{equation}
	\label{eq:probability-approximate-visscher-n}
	\mathcal P^n 
	\approx (\PSI_R^{n})^T \volSec \PSI_R^n 
	+ (\PSI_I^{n-\frac12})^T \volSec \PSI_I^{n+\frac12} \,,
\end{equation}
which involves the samples of the real part of the wavefunction at time $t = n\dt$ and the product of the samples of the imaginary part at $(n-0.5)\dt$ and $(n+0.5)\dt$.
Hence, $\mathcal P^n$ can be seen as an approximation of $\mathcal P(t)$ in \eqref{eq:probability-continuous} at\footnote{Alternatively, $\mathcal P^n$ can be interpreted as an approximation of $\mathcal P(t)$ performed at $t = (n-0.5)\Delta t$. To see this, we can use \eqref{eq:cmf-a} instead of \eqref{eq:cmf-b} to write $\mathcal P^n$ as 
\begin{equation}\label{eq:probability-rewritten-for-n-0.5}
	\mathcal P^n = 
	(\PSI_R^{n})^T \volSec \PSI_R^{n-1} 
	+ (\PSI_I^{n-\frac12})^T \volSec \PSI_I^{n-\frac12}
	- \dfrac{\dt}{\hbar} (\PSI_R^n)^T \HH_{\bot} [\nabla\PSI_I]_{\bot}^{n-\frac12}
	\,, 
	\quad 
	\forall n = 1, 2, \hdots, n_t \,,
\end{equation}
to obtain
\begin{equation}
	\label{eq:probability-approximate-visscher-n-0.5}
	\mathcal P^n 
	\approx 
	(\PSI_R^{n})^T \volSec \PSI_R^{n-1} 
	+ (\PSI_I^{n-\frac12})^T \volSec \PSI_I^{n-\frac12} \,.
\end{equation}} $t = n \Delta t$.

The combination of staggered samples in \eqref{eq:probability-approximate-visscher-n} (and \eqref{eq:probability-approximate-visscher-n-0.5} in the footnote) has been used by Visscher~\cite{visscher-1991-cp} to propose expressions for the probability density in 1D leap-frog FDTD-Q, where periodic or zero Dirichlet boundary conditions were assumed. Visscher argued that those expressions ensure that the total probability stays constant with time. This combination of samples has also been shown effective in defining energy in FDTD for electromagnetics~\cite{edelvik2004general}.
For the case of zero Dirichlet or zero Neumann boundary conditions, the last term in \eqref{eq:probability-rewritten-for-n} becomes zero and 
$\mathcal P^n$ reduces to \eqref{eq:probability-approximate-visscher-n}, becoming analogous to the expressions proposed in~\cite{visscher-1991-cp}. Hence, Definition~\ref{def:probability} can be thought of as a generalization of \eqref{eq:probability-approximate-visscher-n} for a 3D region that could be either isolated or open to the flow of probability through the boundary. The proposed Definition~\ref{def:probability} has an advantage of being a quadratic form, which will be important for the analysis of the conservation properties~\cite{edelvik2004general,jnl-2018-fdtd-3d-dissipative}.

Next, we propose an expression approximating \eqref{eq:current-continuous} in FDTD-Q. This expression quantifies the probability current through the boundary of the region and, as will be shown in Section~\ref{sec:subsec-prob-conserv}, respects the principle of probability conservation.
\begin{definition}[Probability current]
	The probability current flowing out of a region described by FDTD-Q equations \eqref{eq:cmf-a}--\eqref{eq:cmf-b} is given by
	\begin{equation}
		\label{eq:current}
		\mathcal I_P^{n+\frac12} 
		= 
		\dfrac2{\hbar} \left(\dfrac{\PSIbig^{n+1}+\PSIbig^n}{2}\right)^T 
		(\J_1 \otimes \HH_{\bot}) [\nabla \PSIbig]_{\bot}^{n+\frac12}\,,
		\quad
		\forall n = 0, 1, \hdots, n_t-1\,.
	\end{equation}
\end{definition}

In order to recognize that \eqref{eq:current} approximates \eqref{eq:current-continuous}, we expand \eqref{eq:current} using \eqref{eq:mtx-Hbot}, \eqref{eq:mtx-psibig}, \eqref{eq:mtx-nabla-psi-big}, and \eqref{eq:mtx-J} to obtain
\begin{multline}
	\label{eq:current-expanded-0}
	\mathcal I_P^{n+\frac12} 
	= 
	\dfrac2{\hbar} \left(\dfrac{\PSI_R^{n+1}+\PSI_R^n}{2}\right)^T 
	\HH_{\bot} [\nabla \PSI_I]_{\bot}^{n+\frac12}
	-
	\dfrac2{\hbar} \left(\dfrac{\PSI_I^{n+\frac12}+\PSI_I^{n-\frac12}}{2}\right)^T \HH_{\bot} [\nabla \PSI_R]_{\bot}^{n}
	\\ 
	= 
	\frac{\hbar}{m} \left(\Lbot^T \dfrac{\PSI_R^{n+1}+\PSI_R^n}{2}\right)^T 
	\ndot \areaSecBndry [\nabla \PSI_I]_{\bot}^{n+\frac12}
	-
	\frac{\hbar}{m} \left(\Lbot^T \dfrac{\PSI_I^{n+\frac12}+\PSI_I^{n-\frac12}}{2}\right)^T 
	\ndot \areaSecBndry [\nabla \PSI_R]_{\bot}^{n}\,.
\end{multline} 
In \eqref{eq:current-expanded-0}, each row of $\Lbot^T$ selects from $\PSI_R$ or $\PSI_I$ the sample on the boundary node collocated with the corresponding hanging variable $[\nabla \psi_I]_{\bot}$ or $[\nabla \psi_R]_{\bot}$. This way, expressions of the form $\left(\Lbot^T \PSI_R\right)^T 
\ndot \areaSecBndry [\nabla \PSI_I]_{\bot}$ or $\left(\Lbot^T \PSI_I\right)^T \times
\ndot \areaSecBndry [\nabla \PSI_R]_{\bot}$ are summations of terms representing the contribution of each secondary cell face on the boundary to the flux of $\psi_R \nabla \psi_I$ or $\psi_I \nabla \psi_R$, respectively. With that, we can explicitly write the different terms in $\mathcal I_P^{n+0.5}$ by first splitting \eqref{eq:current-expanded-0} into the contributions from the six faces of the region's boundary
\begin{equation}\label{eq:current-expanded-sides}
	\mathcal I_P^{n+\frac12} 
	= 
	\mathcal I_{P,\text{W}}^{n+\frac12}
	+ \mathcal I_{P,\text{E}}^{n+\frac12}
	+ \mathcal I_{P,\text{S}}^{n+\frac12}
	+ \mathcal I_{P,\text{N}}^{n+\frac12}
	+ \mathcal I_{P,\text{B}}^{n+\frac12}
	+ \mathcal I_{P,\text{T}}^{n+\frac12}\,,
\end{equation}
The contribution from the west face is 
\begin{equation}\label{eq:current-expanded-w}
	\mathcal I_{P,\text{W}}^{n+\frac12}
	= 
	\sum_{j=1}^{n_y+1} \sum_{k=1}^{n_z+1} 
	\frac{\hbar}{m} [\hat n_{\text{W}} \cdot \hat x]
	\Delta S_x''|_{1,j,k} 
	\left(	 	 
	\dfrac{\psi_R|_{1,j,k}^{n+1}+\psi_R|_{1,j,k}^n}{2} [\partial_x \psi_I]_{1,j,k}^{n+\frac12}
	-
	\dfrac{\psi_I|_{1,j,k}^{n+\frac12}+\psi_I|_{1,j,k}^{n-\frac12}}{2}
	[\partial_x \psi_R]_{1,j,k}^{n}
	\right)	\,,
\end{equation}
where $[\hat n_{\text{W}} \cdot \hat x] = -1$, with $\hat n_\text{W}$ representing the outward normal vector on the west face of the boundary. The contributions from the other five faces have analogous expressions. From \eqref{eq:current-expanded-sides} and \eqref{eq:current-expanded-w}, one can see that the first term in \eqref{eq:current-expanded-0} approximates the flux of $(\hbar/m) \psi_R \nabla \psi_I$ at $t = (n+0.5)\dt$ and the second term approximates the flux of $-(\hbar/m) \psi_I \nabla \psi_R$ at $t = n\dt$, which are the fluxes involved in \eqref{eq:current-continuous}.

\subsection{Probability conservation}
\label{sec:subsec-prob-conserv}
In order to ensure that expressions in \eqref{eq:probability} and \eqref{eq:current} respect the principle of probability conservation in the discrete domain, we need to satisfy the following conditions:
\begin{enumerate}
	\item The total probability \eqref{eq:probability} should be non-negative.
	\item The rate of change of the total probability should equal the rate at which the probability is absorbed through the boundary via the probability current \eqref{eq:current}.
\end{enumerate}
These conditions are analogous to the conditions for a lossless system in the context of dissipative systems theory~\cite{byrnes1994losslessness,willems1972dissipative}.
The importance of a condition such as Condition~1 for defining lossless systems has been discussed in \cite{willems1972dissipative}. This condition has the same significance for studying the probability conservation. If Condition~2 holds without imposing Condition~1, the probability contained in the region is allowed to become infinitely negative. If that occurs, the region will supply an infinite amount of probability to the surrounding space, akin to a bottomless well. This would clearly indicate a violation in the principle of probability conservation. The following theorem provides a restriction on the time step in FDTD-Q which ensures that the proposed expressions for the total probability~\eqref{eq:probability} and probability current~\eqref{eq:current} satisfy the two conditions above.

\begin{theorem}
	\label{thm:probability-conserv}
	
	Consider a region described by FDTD-Q equations \eqref{eq:cmf-a}--\eqref{eq:cmf-b} with the time step taken below the following generalized CFL limit
	\begin{equation}
		\label{eq:cfl-gen}
		\dt < \dt_{\text{CFL,gen}} = \dfrac{2}{\rho\left(
			\dfrac{1}{\hbar} (\volSec)^{-\frac12} \HH (\volSec)^{-\frac12}
			\right)}\,,
	\end{equation}
	where $\rho(.)$ is the spectral radius of a matrix and $(.)^{-\frac12}$ denotes the inverse of the principal square root\footnote{
	For $\volSec$, $(\volSec)^{-\frac12}$ is simply a diagonal matrix containing the reciprocals of square roots of the diagonal elements of $\volSec$ \cite{golub-matrix-comput-4ed}.}~\cite{golub-matrix-comput-4ed}.
	For this region, the total probability \eqref{eq:probability} is bounded below by zero
	\begin{equation}
		\label{eq:bound-on-P}
		\mathcal P^n \ge 0, \quad \forall n = 0, \hdots, n_t \,.
	\end{equation}
	Moreover, the total probability \eqref{eq:probability} and the probability current \eqref{eq:current} satisfy the following relation:
	\begin{equation}
		\label{eq:probability-balance}
		\dfrac{
			\mathcal P^{n+1} - \mathcal P^n
		}{\dt} = -\mathcal I_P^{n+\frac12}
		\,, \quad 
		\forall n = 0, 1, \hdots, n_t-1\,.
	\end{equation}
\end{theorem}
Prior to showing the proof of the theorem, we elaborate on its meaning. 
When the generalized CFL condition~\eqref{eq:cfl-gen} holds, the largest amount of probability that the region can supply to the surrounding space via $\mathcal I_P^{n+0.5}$ over the course of the simulation is equal to the probability stored in the region at the beginning of the simulation~\cite{willems1972dissipative}. In contrast, when the time step exceeds the generalized CFL limit, there is no bound on how much probability can leave the region. Hence, violation of the generalized CFL condition allows the region to provide an infinite amount of spurious probability to the surrounding space. This behavior would be unphysical and would distort calculations of any quantity one wishes to obtain from the simulation.

The matrices involved in the expression for the generalized CFL limit in \eqref{eq:cfl-gen} depend only on the cell dimensions and on the potential profile, similarly to the conventional CFL limit~\eqref{eq:cfl}. Moreover, the following relation holds between the conventional and generalized CFL limits \eqref{eq:cfl} and \eqref{eq:cfl-gen}. 
\begin{theorem}
	\label{thm:cfl<=cflgen}
	Consider a region described by \eqref{eq:cmf-a}--\eqref{eq:cmf-b}. Let $\dt_{\text{CFL}}$ be the CFL limit in \eqref{eq:cfl} and let $\dt_{\text{CFL,gen}}$ be the generalized CFL limit in \eqref{eq:cfl-gen}. Then,
	\begin{equation}
		\dt_{\text{CFL}} \le \dt_{\text{CFL,gen}}\,.
	\end{equation}
\end{theorem}
\begin{proof}	
	See Appendix~\ref{appendix:proof-CFL<=CFLgen}.
\end{proof}
Hence, the CFL limit \eqref{eq:cfl} can be used in place of \eqref{eq:cfl-gen} in Theorem~\ref{thm:probability-conserv} as a sufficient condition to ensure probability conservation. 
The proof of Theorem~\ref{thm:cfl<=cflgen} is based on showing that the total probability can be written as the sum of probabilities associated with each cell and proving the statement of Theorem~\ref{thm:cfl<=cflgen} for each single-cell region. A similar approach has been used in \cite{edelvik2004general} in the context of electromagnetic energy. In essence, the condition $\dt < \dt_{\text{CFL}}$ ensures that probability is conserved in each primary cell and, consequently, in any region composed by multiple primary cells.
Theorems~\ref{thm:probability-conserv} and \ref{thm:cfl<=cflgen} give a new meaning to the CFL condition for stability~\eqref{eq:cfl}. Specifically, we recognize that the same condition can be used to also ensure the conservation of probability of a general open region in FDTD-Q.
Next, we provide a proof of Theorem~\ref{thm:probability-conserv}, starting with the following lemma.

\begin{lemma}
	\label{lemma:Ppd<=>dt<CFLgen}
	Matrix $\Pbig$ in \eqref{eq:mtx-P} has the following property:
	\begin{equation}
		\label{eq:Ppd<=>dt<CFLgen}
		\Pbig \succ 0 \quad \iff \quad \dt < \dt_{\text{CFL,gen}} \,,
	\end{equation}
	where ``$\succ 0$'' denotes a positive definite matrix.
\end{lemma}
\begin{proof}
	Consider matrix $\Pbig$ defined by \eqref{eq:mtx-P}. Using the properties of the Schur complement~\cite{boyd-convex-2004}, the condition $\Pbig \succ 0$ holds if and only if
	\begin{equation}
		\label{eq:P-Schur}
		\begin{cases}
		\volSec \succ 0 \\
		\volSec - \left(\dfrac{\dt}{2\hbar}\right)^2 \HH (\volSec)^{-1} \HH \succ 0
		\end{cases}. 
	\end{equation}
	The first condition in \eqref{eq:P-Schur} holds for any time step. The second condition can be simplified by writing the equivalent~\cite{golub-matrix-comput-4ed} condition
	\begin{equation}
		\label{eq:proofP>0eqCFLgenIntermid}
		(\volSec)^{-\frac12}\left(\volSec - \left(\dfrac{\dt}{2\hbar}\right)^2 \HH (\volSec)^{-1} \HH \right) (\volSec)^{-\frac12} 
		\succ 0\,,
	\end{equation}
	which reduces to
	\begin{equation}
		\I - \left(\frac{\dt}{2}\right)^2\mathbf{\Sigma}^2 \succ 0\,,
	\end{equation}
	where
	\begin{equation}
		\mathbf{\Sigma} = \dfrac{1}{\hbar} (\volSec)^{-\frac12}
		\HH (\volSec)^{-\frac12}\,.
	\end{equation}
	Let
	\begin{equation}
		\label{eq:P>0<=>dt<CFLgen-mtx-decompose}
		\mathbf{\Sigma}
		=
		\mathbf Q \mathbf{\Lambda} \mathbf Q^H
	\end{equation}
	be a Schur decomposition of the symmetric real (hence normal) matrix $\mathbf{\Sigma}$~\cite{golub-matrix-comput-4ed}, where $\mathbf Q$ is a square unitary matrix, $(.)^H$ denotes a conjugate transpose, and $\mathbf{\Lambda}$ is a diagonal matrix containing the real eigenvalues of $\mathbf{\Sigma}$. Then~\cite{golub-matrix-comput-4ed}, 
	\begin{equation}
		\I - \left(\frac{\dt}{2}\right)^2\mathbf{\Sigma}^2 \succ 0 
		\quad \iff \quad 
		\I - \left(\frac{\dt}{2}\right)^2 \mathbf{\Lambda}^2\succ 0  
		\quad \iff \quad
		\dt < \dfrac{2}{\rho(\mathbf{\Sigma)}}\,,
	\end{equation}
	proving~\eqref{eq:Ppd<=>dt<CFLgen}.
\end{proof}
\begin{proof}[Proof of Theorem~\ref{thm:probability-conserv}]
	Assume the time step is taken below the generalized CFL limit~\eqref{eq:cfl-gen}. From Lemma~\ref{lemma:Ppd<=>dt<CFLgen}, this implies that $\Pbig$ is positive definite. With this, \eqref{eq:bound-on-P} follows directly from Definition~\ref{def:probability}.
	
	The relation \eqref{eq:probability-balance} can be shown by expanding the left hand side using Definition~\ref{def:probability}
	\begin{equation}
		\dfrac{
			\mathcal P^{n+1} - \mathcal P^n
		}{\dt}
		=
		\dfrac{
			(\PSIbig^{n+1})^T \Pbig \PSIbig^n - (\PSIbig^{n})^T \Pbig \PSIbig^n
		}{\dt}
		=
		2 \dfrac{(\PSIbig^{n+1}+\PSIbig^n)^T}{2} \Pbig \dfrac{\PSIbig^{n+1}-\PSIbig^n}{\dt}\,.
	\end{equation}
	Using \eqref{eq:cmf-single},
	\begin{equation}
		\dfrac{
			\mathcal P^{n+1} - \mathcal P^n
		}{\dt}
		=		
		\dfrac{2}{\hbar} \dfrac{(\PSIbig^{n+1}+\PSIbig^n)^T}{2} \left(\J_1 \otimes \HH\right)
		\dfrac{\PSIbig^{n+1} + \PSIbig^n}2 
		- \dfrac{2}{\hbar} \dfrac{(\PSIbig^{n+1}+\PSIbig^n)^T}{2} (\J_1 \otimes \HH_{\bot})
		[\nabla \PSIbig]_{\bot}^{n+\frac12}\,,
	\end{equation}
	and using the fact that $\J_1 \otimes \HH$ is a skew-symmetric matrix,
	\begin{equation}
		\dfrac{
			\mathcal P^{n+1} - \mathcal P^n
		}{\dt}
		=	
		- \dfrac{2}{\hbar} \dfrac{(\PSIbig^{n+1}+\PSIbig^n)^T}{2} (\J_1 \otimes \HH_{\bot})
		[\nabla \PSIbig]_{\bot}^{n+\frac12} 
		= 
		-\mathcal I_P^{n+\frac12}\,,
	\end{equation}
	which proves \eqref{eq:probability-balance}.
\end{proof}

\section{Energy conservation}
\label{sec:energy-conserv}

In this section, we propose expressions for the total energy in the region and the power supplied through its boundary and study conditions under which these expressions satisfy the principle of energy conservation. We find that energy conservation can be demonstrated under the generalized CFL limit \eqref{eq:cfl-gen} if the total probability~\eqref{eq:probability} is bounded from above. We further argue that the existence of the upper bound on probability is guaranteed as long as the model of the space outside the region conserves probability as well.

\subsection{Total energy and supplied power}

The following expression can be used to describe the energy associated with the region
\begin{equation}\label{eq:hcont}
	\mathcal H(t) = \iiint_V W(x,y,z,t) \ dV\,,
\end{equation}
where $W$ is the energy density\footnote{
	Different expressions can be chosen to represent the kinetic energy contribution $W_{\text{kin}}$ to the energy density \eqref{eq:energy-den-cont}. In this work we choose $W_{\text{kin}}^{(1)} = \frac{\hbar^2}{2m} \nabla\psi^{*}\cdot \nabla \psi$, which appears in \cite{cohen-1979-jcp,chaus-1992-umj-energy-flux-quantum,hong-2006-anm}. Expression $W_{\text{kin}}^{(2)} = - \frac{\hbar^2}{2m} \psi^{*}\nabla^2 \psi$ from~\cite{miller-2008, cohen-1979-jcp} is one possible alternative. 
	When considering the entire space in \eqref{eq:hcont}, the two expressions $W_{\text{kin}}^{(1)}$ and $W_{\text{kin}}^{(2)}$ can be shown to give the same value of $\mathcal H(t)$~\cite{cohen-1979-jcp}. However, this is not the case when a finite region $V$ is considered. 
	Various expressions for $W_{\text{kin}}$, including $W_{\text{kin}}^{(1)}$ and $W_{\text{kin}}^{(2)}$, have been studied in \cite{cohen-1979-jcp}.
}~\cite{chaus-1992-umj-energy-flux-quantum}
\begin{equation}\label{eq:energy-den-cont}
	W(x,y,z,t) = \frac{\hbar^2}{2m} \nabla \psi^* \cdot \nabla \psi 
	+ U \psi^* \psi
	=\frac{\hbar^2}{2m} \nabla \psi_R \cdot \nabla \psi_R 
	+ \frac{\hbar^2}{2m} \nabla \psi_I \cdot \nabla \psi_I 
	+ U \psi_R^2 + U \psi_I^2   \,.
\end{equation}
The corresponding power entering the region through the boundary is given by
\begin{equation}
	\label{eq:power-cont}
	s(t)
	=
	\oiint_S \vec S(x,y,z,t) \cdot (-\hat n) \ dS	\,,
\end{equation}
where $\vec S$ is the energy flux density given by~\cite{chaus-1992-umj-energy-flux-quantum}
\begin{equation}
	\vec S(x,y,z,t) =
	\dfrac{i \hbar}{2m}\left( 
	\left(-\frac{\hbar^2}{2m}\nabla^2 \psi + U \psi\right)\nabla \psi^* - \left(-\frac{\hbar^2}{2m}\nabla^2 \psi + U \psi\right)^* \nabla \psi
	\right)
	=
	- \frac{\hbar^2}{m} \dfrac{\partial \psi_R}{\partial t}	\nabla \psi_R
	- \frac{\hbar^2}{m} \dfrac{\partial \psi_I}{\partial t} \nabla \psi_I\,.
\end{equation} 
The proposed definitions of the total energy and energy flux in an FDTD-Q region serve as discrete counterparts of \eqref{eq:hcont} and \eqref{eq:power-cont}.

Similarly to the case of probability, one could define the total energy in FDTD-Q by directly discretizing the volume integration and the gradient of the wavefunction in \eqref{eq:hcont}, arriving at
\begin{equation}
	\label{eq:H-naive-expanded}
	\mathcal H_{\text{simple}}^n 
	= \dfrac{\hbar^2}{2m} (\nabla \PSI_R^{n})^T\areaSec\lenPrim (\nabla \PSI_R^n) 
	+ \dfrac{\hbar^2}{2m} (\nabla\PSI_I^{n-\frac12})^T \areaSec \lenPrim (\nabla\PSI_I^{n-\frac12})
	+ (\PSI_R^{n})^T\volSec \diagon{U}\PSI_R^n
	+ (\PSI_I^{n-\frac12})^T\volSec \diagon{U}\PSI_I^{n-\frac12}\,,
\end{equation}
where $\nabla \PSI_R^n$ and $\nabla\PSI_I^{n-0.5}$ are defined in \eqref{eq:gradPsiR} and \eqref{eq:gradPsiI}, respectively. Using \eqref{eq:hh}, \eqref{eq:H-naive-expanded} can be written more compactly as
\begin{equation}
	\label{eq:H-naive}
	\mathcal H_{\text{simple}}^n 
	= (\PSI_R^n)^T \HH \PSI_R^n + (\PSI_I^{n-\frac12})^T \HH \PSI_I^{n-\frac12}\,.
\end{equation}
Similarly to $\mathcal P_{\text{simple}}^n$, $\mathcal H_{\text{simple}}^n$ does not respect the energy conservation principle, as demonstrated in Section~\ref{sec:ne}. Instead, we propose expressions for the total energy and the corresponding supplied power for which the energy conservation can be shown.

\begin{definition}[Total energy]
	The total energy stored in a region described by FDTD-Q equations \eqref{eq:cmf-a}--\eqref{eq:cmf-b} is given by
	\begin{equation}
		\label{eq:energy}
		\mathcal H^{n} = (\PSI_R^n)^T \HH \PSI_R^n + (\PSI_I^{n-\frac12})^T \HH \PSI_I^{n-\frac12} + \dt \dfrac{(\PSI_R^n-\PSI_R^{n-1})^T}{\dt} (\hbar \volSec) \dfrac{\PSI_I^{n+\frac12}-\PSI_I^{n-\frac12}}{\dt} \,, 
		\quad n = 1, 2, \hdots, n_t-1\,.
	\end{equation}
\end{definition}
The expression for $\mathcal H^n$ consists of $\mathcal H_{\text{simple}}^n$ and a term proportional to $\dt$, which would vanish when the time step approaches zero. In that respect, the expression \eqref{eq:energy} somewhat resembles the expressions for energy in FDTD for Maxwell's equations~\cite{edelvik2004general}. The last term in \eqref{eq:energy} is needed to ensure that the principle of energy conservation is respected, as will be shown in the subsequent discussion.

In order to see why $\mathcal H^n$ approximates \eqref{eq:hcont}, we rewrite \eqref{eq:energy} using \eqref{eq:cmf-a} as
\begin{equation}
	\mathcal H^{n} = (\PSI_R^n)^T \HH \PSI_R^n + (\PSI_I^{n-\frac12})^T \HH \PSI_I^{n-\frac12} + \dt \left(\HH \PSI_I^{n-\frac12}
	- \HH_{\bot} [\nabla \PSI_I]_{\bot}^{n-\frac12}\right)^T \dfrac{\PSI_I^{n+\frac12}-\PSI_I^{n-\frac12}}{\dt}\,, 
	\quad n = 1, 2, \hdots, n_t-1\,,
\end{equation}
which simplifies to
\begin{equation}
	\mathcal H^{n} = (\PSI_R^n)^T \HH \PSI_R^n+ (\PSI_I^{n-\frac12})^T \HH \PSI_I^{n+\frac12}	
	- \dt 
	\left(\HH_{\bot}[\nabla \PSI_I]_{\bot}^{n-\frac12}\right)^T
	 \dfrac{\PSI_I^{n+\frac12}-\PSI_I^{n-\frac12}}{\dt}
	\,, 
	\quad n = 1, 2, \hdots, n_t-1\,.
\end{equation}
Assuming $\dt$ is small 
\begin{equation}\label{eq:happrox-n}
	\mathcal H^n 
	\approx (\PSI_R^n)^T \HH \PSI_R^n+ (\PSI_I^{n-\frac12})^T \HH \PSI_I^{n+\frac12}
\end{equation}
and using the definition of matrix $\HH$ in \eqref{eq:hh}, 
\begin{equation}\label{eq:happrox-n-expanded}
	\mathcal H^n 
	\approx 
	\dfrac{\hbar^2}{2m} (\nabla \PSI_R^n)^T \areaSec \lenPrim (\nabla \PSI_R^n)
	+ 
	\dfrac{\hbar^2}{2m} (\nabla\PSI_I^{n-\frac12})^T \areaSec \lenPrim (\nabla\PSI_I^{n+\frac12})
	+
	(\PSI_R^n)^T \volSec \diagon{U} \PSI_R^n
	+
	(\PSI_I^{n-\frac12})^T \volSec \diagon{U} \PSI_I^{n+\frac12}	\,.
\end{equation}
From \eqref{eq:happrox-n-expanded}, $\mathcal H^n$ can be seen as an approximation of~\eqref{eq:hcont} at\footnote{Energy $\mathcal H^n$ can also be interpreted as an approximation of \eqref{eq:hcont} at $t = (n-0.5)\Delta t$ by rewriting \eqref{eq:energy} using \eqref{eq:cmf-b} as 
	\begin{equation}\label{eq:h-before-approx-n}
		\mathcal H^{n} = 
		(\PSI_R^{n-1})^T\HH \PSI_R^{n}
		+(\PSI_I^{n-\frac12})^T \HH \PSI_I^{n-\frac12} 
		+ \dt \dfrac{(\PSI_R^n - \PSI_R^{n-1} )^T}{\dt}\HH_{\bot} [\nabla \PSI_R]_{\bot}^n,
		\quad n = 1, 2, \hdots, n_t-1\,
	\end{equation}
	and neglecting the last term. After substituting the definition of $\HH$ in \eqref{eq:hh}, one would obtain
	\begin{equation}\label{eq:happrox-n-0.5}
		\mathcal H^n \approx 
		\dfrac{\hbar^2}{2m} (\nabla \PSI_R^{n-1})^T \areaSec \lenPrim (\nabla \PSI_R^{n})
		+ 
		\dfrac{\hbar^2}{2m} (\nabla\PSI_I^{n-\frac12})^T \areaSec \lenPrim (\nabla\PSI_I^{n-\frac12})
		+
		(\PSI_R^{n-1})^T \volSec \diagon{U} \PSI_R^n
		+
		(\PSI_I^{n-\frac12})^T \volSec \diagon{U} \PSI_I^{n-\frac12} \,.
	\end{equation}}
 $t = n\dt$.

\begin{definition}[Supplied power]
	The power supplied through the boundary to a region described by \eqref{eq:cmf-a}--\eqref{eq:cmf-b} is given by
	\begin{multline}
		\label{eq:supply}
		s^{n+\frac12} = 
		2 \dfrac{(\PSI_R^{n+1}-\PSI_R^n)^T}{\dt} 
		\HH_{\bot} 
		\dfrac{[\nabla \PSI_R]_{\bot}^{n+1} + [\nabla \PSI_R]_{\bot}^{n}}2
		+ 
		2 \dfrac{(\PSI_I^{n+\frac12}-\PSI_I^{n-\frac12})^T}{\dt} 
		\HH_{\bot} 
		\dfrac{[\nabla \PSI_I]_{\bot}^{n+\frac12} + [\nabla \PSI_I]_{\bot}^{n-\frac12}}{2}\,,\\
		 n = 1, 2, \hdots, n_t-2\,.
	\end{multline}
\end{definition}
In order to reveal the similarity between \eqref{eq:supply} and \eqref{eq:power-cont}, we expand \eqref{eq:supply} using the definition of $\HH_{\bot}$ in \eqref{eq:mtx-Hbot} to obtain
\begin{multline}
	\label{eq:supply-expanded}
	s^{n+\frac12} = 
	\frac{\hbar^2}{m} \left(\Lbot^T \dfrac{\PSI_R^{n+1}-\PSI_R^n}{\dt} \right)^T  \ndot \areaSecBndry
	\dfrac{[\nabla \PSI_R]_{\bot}^{n+1} + [\nabla \PSI_R]_{\bot}^{n}}2\\
	+ 
	\frac{\hbar^2}{m} \left(\Lbot^T \dfrac{\PSI_I^{n+\frac12}-\PSI_I^{n-\frac12}}{\dt}\right)^T \ndot \areaSecBndry 
	\dfrac{[\nabla \PSI_I]_{\bot}^{n+\frac12} + [\nabla \PSI_I]_{\bot}^{n-\frac12}}{2}\,.
\end{multline}
The first term in \eqref{eq:supply-expanded} is an approximation of the part of \eqref{eq:power-cont} associated with $\psi_R$ at $t = (n+0.5)\dt$. The second term is an approximation of the part of \eqref{eq:power-cont} associated with $\psi_I$ at $t = n\dt$.

\subsection{Energy conservation}

The energy conservation properties of \eqref{eq:energy} and \eqref{eq:supply} are verified in a similar way to the probability conservation. In particular, using concepts from dissipative systems theory \cite{willems1972dissipative,byrnes1994losslessness}, we investigate the conditions for which the energy is bounded from below and show that the rate of change of energy equals the power supplied to the region through the boundary.
\begin{theorem}
	\label{thm:energy-conserv}
	
	Consider a region described by FDTD-Q equations \eqref{eq:cmf-a}--\eqref{eq:cmf-b}. Assume that $\dt < \dt_{\text{CFL,gen}}$ in \eqref{eq:cfl-gen} and that the total probability \eqref{eq:probability} is bounded from above by some finite value $\mathcal P_{\max}$
	\begin{equation}\label{eq:en-conserv-P-bound-condition}
		\mathcal P^n \le \mathcal P_{\max} \,, \quad n = 0,1,\hdots, n_t \,.
	\end{equation}
	For this region, total energy \eqref{eq:energy} is bounded from below as follows
	\begin{equation}
		\label{eq:H-lower-bound}
		\mathcal H^n \ge 
		\dx \dy \dz
		\dfrac{\mathcal P_{\max}}{\lambda_{\min}(\Pbig)}
		\left(
		\min
		\left(\min_{i,j,k} U|_{i,j,k}, 0\right) -\dfrac{4\hbar}{\dt} 
		\right)  \,, 
		\quad \forall n = 1, 2, \hdots, n_t-1\,,
	\end{equation}
	where $\lambda_{\min}$ denotes the smallest eigenvalue of a symmetric real matrix.
	Moreover, the total energy \eqref{eq:energy} and the supplied power \eqref{eq:supply} satisfy the following relation:
	\begin{equation}
		\label{eq:energy-balance}
		\dfrac{
			\mathcal H^{n+1} - \mathcal H^n
		}{\dt} = s^{n+\frac12}
		\,, \quad 
		\forall n = 1, 2, \hdots, n_t-2\,.
	\end{equation}
\end{theorem}
Before proving Theorem~\ref{thm:energy-conserv}, we argue that the condition \eqref{eq:en-conserv-P-bound-condition} can be assumed if the model of the space outside the region obeys the principle of probability conservation.
\begin{lemma}\label{lemma:Pmax-condition-validity}
	Consider the region in Fig.~\ref{fig:primary-secondary} described by \eqref{eq:cmf-a}--\eqref{eq:cmf-b}. Let the time step be taken below the generalized CFL limit~\eqref{eq:cfl-gen}. Let 
	\begin{equation}\label{eq:Poutside-nonnegative}
	\mathcal P^n_{\dagger} \ge 0	\,, \quad \forall n = 0, 1, \hdots, n_t
	\end{equation}
	be the probability associated with the space outside the region and let $\mathcal I_{P \dagger}^{n+0.5}$, the probability current leaving that space, satisfy
	\begin{equation}\label{eq:probability-balance-outside}
		\dfrac{\mathcal P_{\dagger}^{n+1} - \mathcal P_{\dagger}^n}{\dt} = -\mathcal I_{P\dagger}^{n+\frac12} \,, \quad \forall n = 0, 1, \hdots, n_t-1 \,.
	\end{equation}
	Furthermore, assume that the probability current leaving the region equals the current entering the space surrounding the region:
	\begin{equation}\label{eq:bc-region-and-outside}
		\mathcal I_{P}^{n+\frac12} = - \mathcal I_{P\dagger}^{n+\frac12} \,, \quad \forall n = 0, 1, \hdots, n_t-1\,.
	\end{equation}
	Then the probability associated with the region has a finite a priori upper bound
	\begin{equation}\label{eq:Pmax-in-Lemma-Pmax}
		\mathcal P^n \le \mathcal P_{\max} \,, \quad \forall n = 0, 1, \hdots n_t\,,
	\end{equation} 
	where $\mathcal P_{\max} = \mathcal P^0 + \mathcal P_{\dagger}^{0}$.
\end{lemma}
\begin{proof}
	From Theorem~\ref{thm:probability-conserv}, $\mathcal P^n$ and $\mathcal I_P^{n+0.5}$ satisfy \eqref{eq:probability-balance}. Adding \eqref{eq:probability-balance} and \eqref{eq:probability-balance-outside} and using \eqref{eq:bc-region-and-outside},
	\begin{equation}\label{eq:change-tot-prob-lemma-Pmax-validity}
		\dfrac{(\mathcal P^{n+1}+\mathcal P_{\dagger}^{n+1}) - (\mathcal P^n+\mathcal P_{\dagger}^n)}{\dt} = -\mathcal I_{P}^{n+\frac12} - \mathcal I_{P\dagger}^{n+\frac12} = 0 \,, \quad \forall n = 0, 1, \hdots, n_t-1 \,.
	\end{equation}
	From \eqref{eq:change-tot-prob-lemma-Pmax-validity}, 
	\begin{equation}
		\mathcal P^n + \mathcal P_{\dagger}^n = \mathcal P^0 + \mathcal P_{\dagger}^0  \,, \quad \forall n = 0, 1, \hdots, n_t\,.
	\end{equation}
	Hence, using \eqref{eq:Poutside-nonnegative} we conclude
	\begin{equation}
		\mathcal P^n \le\mathcal P^0 + \mathcal P_{\dagger}^0  \,, \quad \forall n = 0, 1, \hdots, n_t\,.
	\end{equation}
	proving \eqref{eq:Pmax-in-Lemma-Pmax}.
\end{proof}
Typically, the simulation setup would be such that $\mathcal P_{\max}$ in Lemma~\ref{lemma:Pmax-condition-validity} equals to one. A region terminated in zero Dirichlet or zero Neumann boundary conditions would constitute a trivial case of the lemma, with $\mathcal P_{\dagger}^n = 0$ and $\mathcal I_{P\dagger}^{n+0.5} = 0$. Section~\ref{sec:stability-joint} considers in more detail an example of a setup in Lemma~\ref{lemma:Pmax-condition-validity} consisting of two connected regions, as well as the boundary conditions at their interface that ensure \eqref{eq:bc-region-and-outside}. Next, we prove Theorem~\ref{thm:energy-conserv}. 

\begin{proof}[Proof of Theorem~\ref{thm:energy-conserv}.]
	
	First, let us derive the bound \eqref{eq:H-lower-bound}. Under the generalized CFL limit \eqref{eq:cfl-gen}, all eigenvalues of $\Pbig$ are strictly positive (Lemma~\ref{lemma:Ppd<=>dt<CFLgen}). This allows us to derive an upper bound on $||\PSIbig^n||_2$, which will then be used to derive \eqref{eq:H-lower-bound}:
	\begin{equation}
		\lambda_{\min}(\Pbig)||\PSIbig^n||_2^2 \le (\PSIbig^n) \Pbig \PSIbig^n \le \mathcal P_{\max} \,, \quad n = 0, 1, \hdots, n_t\,,
	\end{equation}
	\begin{equation}
		\label{eq:bound-on-norm-psibig}
		||\PSIbig^n||_2 \le \sqrt{\dfrac{\mathcal P_{\max}}{\lambda_{\min}(\Pbig)}}\,, 
		\quad n = 0, 1, \hdots, n_t\,.
	\end{equation}
	
	The first two terms in \eqref{eq:energy} can be written using \eqref{eq:hh} and \eqref{eq:mtx-psibig} as
	\begin{equation}
		\label{eq:Hbound-proof-step}
		(\PSI_R^n)^T \HH \PSI_R^n + (\PSI_I^{n-\frac12})^T \HH \PSI_I^{n-\frac12}
		= 
		(\PSIbig^n)^T \left(\I_2 \otimes \dfrac{\hbar^2}{2m} \D \areaSec (\lenPrim)^{-1} \D^T\right) \PSIbig^n
		+
		(\PSIbig^n)^T (\I_2 \otimes \volSec \diagon{U}) \PSIbig^n\,.
	\end{equation}
	The first term on the right hand side of \eqref{eq:Hbound-proof-step} is nonnegative and the second term is bounded from below as follows
	\begin{equation}
		\begin{cases}
			(\PSIbig^n)^T (\I_2 \otimes \volSec \diagon{U}) \PSIbig^n 
			\ge 0\,, & \text{if~} \min_{i,j,k} U|_{i,j,k} \ge 0\\
			(\PSIbig^n)^T (\I_2 \otimes \volSec \diagon{U}) \PSIbig^n 
			\ge \dx \dy \dz \ \min_{i,j,k} U|_{i,j,k}\ ||\PSIbig^n||_2^2\,, & \text{if~} \min_{i,j,k} U|_{i,j,k} < 0
		\end{cases}\,.
	\end{equation}
	Thus, with the use of \eqref{eq:bound-on-norm-psibig}, the first two terms in \eqref{eq:energy} are bounded from below as
	\begin{multline}
		\label{eq:HboundT1T2}
		(\PSI_R^n)^T \HH \PSI_R^n + (\PSI_I^{n-\frac12})^T \HH \PSI_I^{n-\frac12}
		\ge \dx \dy \dz \ \min\left(\min_{i,j,k} U|_{i,j,k},0\right) ||\PSIbig^n||_2^2 \\
		\ge \dx \dy \dz \ \min\left(\min_{i,j,k} U|_{i,j,k},0\right)\!\dfrac{\mathcal P_{\max}}{\lambda_{\min}(\Pbig)}\,.
	\end{multline}
	The last term in \eqref{eq:energy} has a following lower bound:
	\begin{multline}
		\dt \dfrac{(\PSI_R^n-\PSI_R^{n-1})^T}{\dt} (\hbar \volSec) \dfrac{\PSI_I^{n+\frac12}-\PSI_I^{n-\frac12}}{\dt}
		\ge
		- \dfrac{\hbar}{\dt} \left|(\PSI_R^n-\PSI_R^{n-1})^T \volSec (\PSI_I^{n+\frac12}-\PSI_I^{n-\frac12})\right|\\
		\ge - \dfrac{\hbar}{\dt} ||\PSI_R^n-\PSI_R^{n-1}||_2 ||\volSec||_2 ||\PSI_I^{n+\frac12}-\PSI_I^{n-\frac12}||_2\\
		\ge - \dfrac{\hbar}{\dt} \dx \dy \dz (||\PSI_R^n||_2+||\PSI_R^{n-1}||_2) (||\PSI_I^{n+\frac12}||+||\PSI_I^{n-\frac12}||_2)\,.
	\end{multline}
	Using \eqref{eq:bound-on-norm-psibig},
	\begin{equation}
		\label{eq:HboundT3}
		\dt \dfrac{(\PSI_R^n-\PSI_R^{n-1})^T}{\dt} (\hbar \volSec) \dfrac{\PSI_I^{n+\frac12}-\PSI_I^{n-\frac12}}{\dt}\ge - \dx \dy \dz \dfrac{4\hbar}{\dt} \dfrac{\mathcal P_{\max}}{\lambda_{\min}(\Pbig)}\,.
	\end{equation}
	Finally, combining \eqref{eq:HboundT1T2} and \eqref{eq:HboundT3},
	\begin{equation}
		\mathcal H^{n} 
		\ge
		\dx \dy \dz \ \min\left(\min_{i,j,k} U|_{i,j,k},0\right)\dfrac{\mathcal P_{\max}}{\lambda_{\min}(\Pbig)}
		-\dx \dy \dz \dfrac{4 \hbar}{\dt} \dfrac{\mathcal P_{\max}}{\lambda_{\min}(\Pbig)}\,,
	\end{equation}
	which proves \eqref{eq:H-lower-bound}.
	
	To show \eqref{eq:energy-balance}, we use \eqref{eq:energy} to expand the left hand side of \eqref{eq:energy-balance} as	
	\begin{multline}\label{eq:energy-balance-proof-step-1}
		\dfrac{
			\mathcal H^{n+1} - \mathcal H^n
		}{\dt}
		=
		\dfrac{
			(\PSI_R^{n+1})^T \HH \PSI_R^{n+1} - (\PSI_R^{n})^T \HH \PSI_R^{n}
		}{\dt} 
		+ \dfrac{
			(\PSI_I^{n+\frac12})^T \HH \PSI_I^{n+\frac12} - (\PSI_I^{n-\frac12})^T \HH \PSI_I^{n-\frac12}
		}{\dt} \\
		+ \dfrac{(\PSI_R^{n+1}-\PSI_R^{n})^T}{\dt} (\hbar \volSec) \dfrac{\PSI_I^{n+\frac32}-\PSI_I^{n+\frac12}}{\dt}
		-  \dfrac{(\PSI_I^{n+\frac12}-\PSI_I^{n-\frac12})^T}{\dt} (\hbar \volSec)\dfrac{\PSI_R^n-\PSI_R^{n-1}}{\dt}\,,
	\end{multline}
	where the last term, being a scalar, has been written as its own transpose.
	The first two terms on the right hand side of \eqref{eq:energy-balance-proof-step-1} can be written as 
	\begin{equation}
		\dfrac{
			(\PSI_R^{n+1})^T \HH \PSI_R^{n+1} - (\PSI_R^{n})^T \HH \PSI_R^{n}
		}{\dt} = \dfrac{(\PSI_R^{n+1}-\PSI_R^n)^T}{\dt} (\HH\PSI_R^{n+1}+\HH\PSI_R^{n}) \,,
	\end{equation}
	\begin{equation}
		\dfrac{
			(\PSI_I^{n+\frac12})^T \HH \PSI_I^{n+\frac12} - (\PSI_I^{n-\frac12})^T \HH \PSI_I^{n-\frac12}
		}{\dt}
		=
		\dfrac{(\PSI_I^{n+\frac12}-\PSI_I^{n-\frac12})^T}{\dt} (\HH \PSI_I^{n+\frac12} + \HH \PSI_I^{n-\frac12})\,.
	\end{equation}
	With that, we can write \eqref{eq:energy-balance-proof-step-1} as
	\begin{multline}
		\dfrac{
			\mathcal H^{n+1} - \mathcal H^n
		}{\dt}
		=
		\dfrac{(\PSI_R^{n+1}-\PSI_R^n)^T}{\dt} \left(
		\HH \PSI_R^{n}+\HH \PSI_R^{n+1}
		+  \hbar \volSec \dfrac{\PSI_I^{n+\frac32}-\PSI_I^{n+\frac12}}{\dt}
		\right)
		\\	
		+ 
		\dfrac{(\PSI_I^{n+\frac12} - \PSI_I^{n-\frac12})^T}{\dt}
		\left(
		\HH \PSI_I^{n+\frac12} + \HH \PSI_I^{n-\frac12}
		- \hbar \volSec \dfrac{\PSI_R^n-\PSI_R^{n-1}}{\dt} 
		\right)\,.
	\end{multline}
	Using~\eqref{eq:cmf-a} and \eqref{eq:cmf-b}, we can make the following substitutions
	\begin{equation}
		\HH \PSI_R^n
		= - \hbar \volSec \dfrac{\PSI_I^{n+\frac12}-\PSI_I^{n-\frac12}}{\dt}
		+ \HH_{\bot} [\nabla \PSI_R]_{\bot}^n
		\,, 
		\quad \forall n = 0, 1, \hdots, n_t-1\,,
	\end{equation}
	\begin{equation}
		\HH \PSI_R^{n+1}
		+ \hbar \volSec \dfrac{\PSI_I^{n+\frac32}-\PSI_I^{n+\frac12}}{\dt}  
		= 
		\HH_{\bot} [\nabla \PSI_R]_{\bot}^{n+1}
		\,, 
		\quad \forall n = -1, 1, \hdots, n_t-2\,,
	\end{equation}
	\begin{equation}
		\HH \PSI_I^{n+\frac12}
		= 
		\hbar \volSec \dfrac{\PSI_R^{n+1}- \PSI_R^n}{\dt} 
		+ \HH_{\bot} [\nabla \PSI_I]_{\bot}^{n+\frac12}
		\,, 
		\quad \forall n = 0, 1, \hdots, n_t-1\,,
	\end{equation} 
	\begin{equation}
		\HH \PSI_I^{n-\frac12}
		-\hbar \volSec \dfrac{\PSI_R^{n}- \PSI_R^{n-1}}{\dt} 
		= 
		\HH_{\bot} [\nabla \PSI_I]_{\bot}^{n-\frac12}\,, 
		\quad \forall n = 1, \hdots, n_t\,,
	\end{equation}
	and write the left hand side of \eqref{eq:energy-balance} as
	\begin{multline}
		\dfrac{
			\mathcal H^{n+1} - \mathcal H^n
		}{\dt}
		=
		\dfrac{(\PSI_R^{n+1}-\PSI_R^n)^T}{\dt} \left(
		- \hbar \volSec \dfrac{\PSI_I^{n+\frac12}-\PSI_I^{n-\frac12}}{\dt} 
		+ \HH_{\bot} [\nabla \PSI_R]_{\bot}^n
		+ \HH_{\bot} [\nabla \PSI_R]_{\bot}^{n+1}
		\right)\\
		+ 
		\dfrac{(\PSI_I^{n+\frac12} - \PSI_I^{n-\frac12})^T}{\dt}
		\left(
		\hbar \volSec \dfrac{\PSI_R^{n+1}- \PSI_R^n}{\dt} 
		+ \HH_{\bot} [\nabla \PSI_I]_{\bot}^{n+\frac12}
		+ \HH_{\bot} [\nabla \PSI_I]_{\bot}^{n-\frac12}
		\right)\\
		=
		\dfrac{(\PSI_R^{n+1}-\PSI_R^n)^T}{\dt} \left(
		\HH_{\bot} [\nabla \PSI_R]_{\bot}^n
		+ \HH_{\bot} [\nabla \PSI_R]_{\bot}^{n+1}
		\right)
		+ 
		\dfrac{(\PSI_I^{n+\frac12} - \PSI_I^{n-\frac12})^T}{\dt}
		\left(
		\HH_{\bot} [\nabla \PSI_I]_{\bot}^{n+\frac12}
		+ \HH_{\bot} [\nabla \PSI_I]_{\bot}^{n-\frac12}
		\right)\,,
	\end{multline}
	which is equal to the right hand side of \eqref{eq:energy-balance}.
\end{proof}

\section{A new framework to create FDTD-Q Schemes with guaranteed stability}
\label{sec:stability}

In many applications it is desirable to create FDTD schemes for the Schrödinger equation where different models are coupled together. Improper coupling is a notorious cause for instabilities in FDTD-type schemes~\cite{kulas-reciprocity}. Hence, careful analysis is required to ensure that a coupled scheme is stable, which is challenging with existing methods. The theory proposed in this work provides a rigorous way to construct new stable schemes in a modular and constructive fashion. The approach is based on ensuring that the models and the coupling between them respect the principle of probability conservation.

Conservation-based approaches have been previously used for analyzing stability of FDTD schemes in electromagnetics~\cite{edelvik2004general,jnl-2018-fdtd-3d-dissipative}. The effectiveness of this type of analysis was demonstrated in FDTD for Maxwell's equations by creating a subgridding scheme~\cite{jnl-2018-fdtd-3d-dissipative} and embedding reduced models with extended CFL limit~\cite{jnl-tap-2018-fdtd-mor}.
Here, as a proof of concept, we show how the conservation approach can be used for analyzing stability of FDTD schemes for the Schrödinger equation using two examples: a region isolated to the flow of probability current and two regions that are coupled via boundary conditions. Application to more advanced scenarios is left for the future work.

\subsection{Region with no probability current through the boundary}

\begin{lemma}
	\label{lemma:stability}
	Consider a region described by FDTD-Q equations \eqref{eq:cmf-a}--\eqref{eq:cmf-b}. Assume that $\dt < \dt_{\text{CFL,gen}}$ in \eqref{eq:cfl-gen} and that 
	\begin{equation}
		\label{eq:stability-IP=0}
		\mathcal I_{P}^{n+\frac12} = 0 \,, \quad \forall n = 0, 2, \hdots, n_t-1\,.
	\end{equation}
	Then
	\begin{equation}
		||\PSIbig^{n}||_2 \le \sqrt{\kappa(\Pbig)} ||\PSIbig^{0}||_2, \quad \forall n = 0, 1, \hdots, n_t\,,
	\end{equation}
	where $\kappa(\Pbig)$ is the condition number of $\Pbig$.
\end{lemma}
Examples of boundary conditions where no net probability current exists include zero Dirichlet ($\psi=0$) and zero Neumann $[\nabla \psi]_{\bot}=0$ boundary conditions.
\begin{proof}[Proof of Lemma~\ref{lemma:stability}]
	Assume $\dt < \dt_{\text{CFL,gen}}$. From Theorem~\ref{thm:probability-conserv} and from \eqref{eq:stability-IP=0}, $\mathcal P^n$ stays constant over the course of the simulation
	\begin{equation}
		(\PSIbig^n)^T \Pbig \PSIbig^n = (\PSIbig^0)^T \Pbig \PSIbig^0 \quad \forall n = 0, 1, \hdots , n_t\,.
	\end{equation}
	Moreover, 
	\begin{equation}
		\lambda_{\min}(\Pbig) ||\PSIbig^n||_2^2\le (\PSIbig^n)^T \Pbig \PSIbig^n
	\end{equation}
	and
	\begin{equation}
		(\PSIbig^0)^T \Pbig \PSIbig^0 \le \lambda_{\max}(\Pbig) ||\PSIbig^0||_2^2\,,
	\end{equation}
	giving
	\begin{equation}
		\lambda_{\min}(\Pbig) ||\PSIbig^n||_2^2\le \lambda_{\max}(\Pbig) ||\PSIbig^0||_2^2\,.
	\end{equation}
	Since, from Lemma~\ref{lemma:Ppd<=>dt<CFLgen}, $\lambda_{\max}(\Pbig)$ is strictly positive,
	\begin{equation}
		||\PSIbig^n||_2^2
		\le
		\dfrac{\lambda_{\max}(\Pbig)}{\lambda_{\min}(\Pbig)} ||\PSIbig^0||_2^2\,.
	\end{equation}
	The ratio of eigenvalues for a symmetric positive definite matrix is the condition number~\cite{golub-matrix-comput-4ed}, which proves the statement of the lemma.
\end{proof}
Hence, the 2-norm of the vector $\PSIbig^n$, which contains the samples of the real and imaginary parts of the wavefunction, has a bound that is known prior to running the simulation. The existence of such bound guarantees stability.
From Theorem~\ref{thm:cfl<=cflgen}, the conventional CFL limit $\dt_{\text{CFL}}$ in \eqref{eq:cfl} can be used in place of $\dt_{\text{CFL, gen}}$ in Lemma~\ref{lemma:stability}. Hence, Lemma~\ref{lemma:stability} in conjunction with Theorem~\ref{thm:cfl<=cflgen} provide an alternative proof for the previously known result~\cite{soriano-analysis-2004,dai-stability-2005} on FDTD-Q stability under the CFL limit \eqref{eq:cfl}. 
Generalized stability limits that could be relaxed to simpler conditions such as \eqref{eq:cfl} have been derived in the past using the iteration matrix approach for a method similar to FDTD-Q~\cite{decleer-2021-jcam} and for FDTD for the Maxwell equations~\cite{remis-2000-jcp,gedney-2011}. 

\subsection{Connection of FDTD-Q regions}
\label{sec:stability-joint}

\newcommand{\psiRA}{\psi_{R}^{A}}
\newcommand{\psiIA}{\psi_{I}^{A}}
\newcommand{\psiRB}{\psi_{R}^{B}}
\newcommand{\psiIB}{\psi_{I}^{B}}
\newcommand{\nyA}{n_y^{A}}

In this section, we discuss how the proposed theory can be used to couple FDTD-Q models in a stable manner. We consider a simple but representative scenario in Fig.~\ref{fig:joint-regions}, involving two adjacent regions discretized using FDTD-Q with the same grid size and time step. The two regions need to be appropriately coupled at the interface to maintain probability conservation and hence stability of the scheme. The north face of the first region (Region~A) is adjacent to the south face of the second region (Region~B). To couple the two regions, we equate the samples of the wavefunction and the hanging variables at the interface between Region~A and Region~B
\begin{subequations}
	\begin{align}
		\psiRA|_{i,\nyA+1,k}^n = \psiRB|_{i,1,k}^n, &
		\quad 2 \le i \le n_x,\ 2 \le k \le n_z,\ 0 \le n \le n_t\,,
		\label{eq:connection-ab-first}
		\\
		\psiIA\big|_{i,\nyA+1,k}^{n-\frac12} = \psiIB\big|_{i,1,k}^{n-\frac12}, &
		\quad 2 \le i \le n_x,\ 2 \le k \le n_z,\ 0 \le n \le n_t\,,
		\label{eq:connection-ab-second}
		\\
		[\partial_y \psiRA]_{i,\nyA+1,k}^n = [\partial_y \psiRB]_{i,1,k}^n,&
		\quad 2 \le i \le n_x,\ 2 \le k \le n_z,\ 0 \le n \le n_t-1\,,
		\label{eq:connection-ab-third}
		\\
		[\partial_y\psiIA]_{i,\nyA+1,k}^{n+\frac12} = [\partial_y \psiIB]_{i,1,k}^{n+\frac12},&
		\quad 2 \le i \le n_x,\ 2 \le k \le n_z ,\ 0 \le n \le n_t-1\,,
		\label{eq:connection-ab-last}
	\end{align}
\end{subequations}
where ``A'' and ``B'' denote the region a quantity corresponds to. At all other nodes on the boundary of the two regions, zero Dirichlet boundary condition is imposed. 
Next, we derive the update equations resulting from \eqref{eq:connection-ab-first}--\eqref{eq:connection-ab-last} and show that the coupled scheme is stable if the generalized CFL limit is satisfied in each of the regions.

\begin{figure}
	\centering
	\includegraphics[scale=0.5]{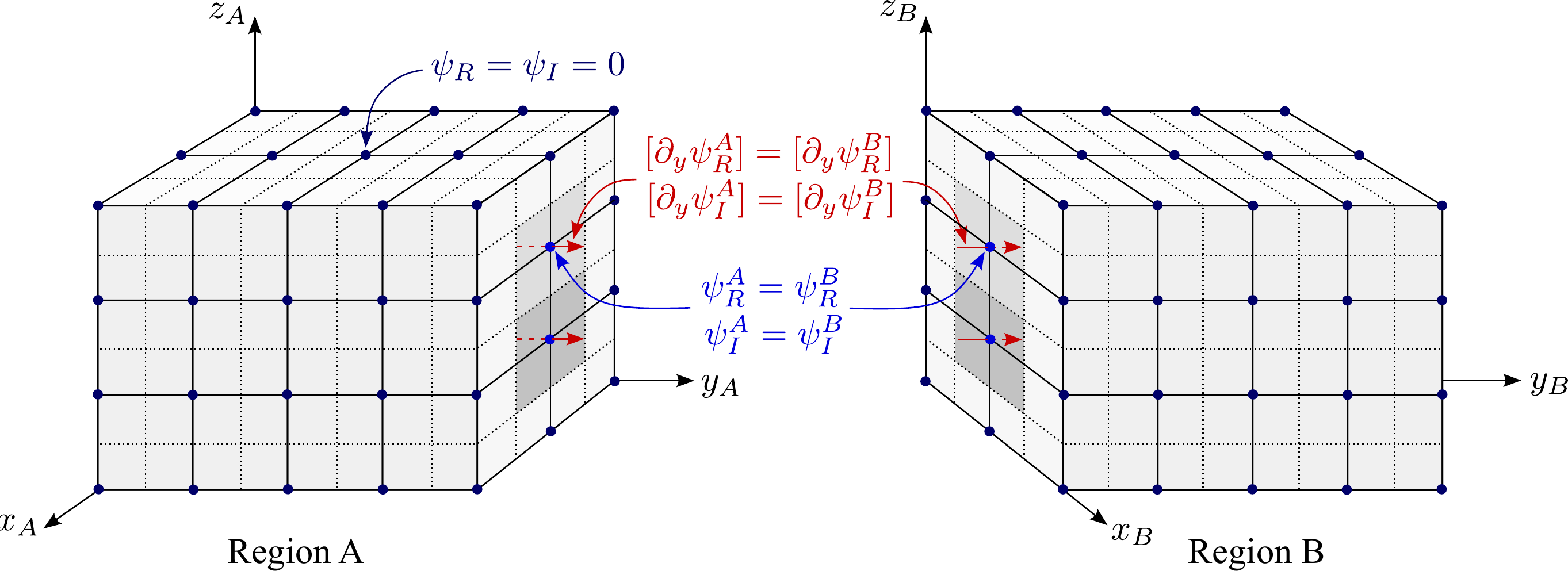}	
	\caption{Two FDTD-Q regions joined together by equating the quantities on the adjacent boundary as described in \eqref{eq:connection-ab-first}--\eqref{eq:connection-ab-last}.}
	\label{fig:joint-regions}
\end{figure}

\subsubsection{Update equations at the interface}
From Section~\ref{sec:equations-rgn}, the equation discretizing \eqref{eq:cont-a} for the nodes on the north face of Region~A is given by
\begin{multline}
	\label{eq:A-north}
	\hbar \ \dx \dfrac{\dy}2 \dz \dfrac{\psi_R^A|_{i,n_y^A+1,k}^{n+1} - \psi_R^A|_{i,n_y^A+1,k}^{n}}{\dt} 
	= -\frac{\hbar^2}{2m} 
	\Bigg[
	\dfrac{\dy}2 \dz\dfrac{ 
		\psi_I^A|_{i+1,n_y^A+1,k}^{n+\frac12} 
		- \psi_I^A|_{i,n_y^A+1,k}^{n+\frac12}
	}{\Delta x}\\	
	- 
	\dfrac{\dy}2 \dz\dfrac{ 
		\psi_I^A|_{i,n_y^A+1,k}^{n+\frac12} 
		- \psi_I^A|_{i-1,n_y^A+1,k}^{n+\frac12}
	}{\Delta x}	
	+ 
	\dx \dz [\partial_y\psi_I^A]_{i,n_y^A+1,k}^{n+\frac12}
	- 
	\dx \dz\dfrac{ 
		\psi_I^A|_{i,n_y^A+1,k}^{n+\frac12} 
		- \psi_I^A|_{i,n_y^A,k}^{n+\frac12}
	}{\Delta y}\\
	+ 
	\dx\dfrac{\dy}2\dfrac{ 
		\psi_I^A|_{i,n_y^A+1,k+1}^{n+\frac12} 
		- \psi_I^A|_{i,n_y^A+1,k}^{n+\frac12} 
	}{\Delta z}
	-
	\dx\dfrac{\dy}2\dfrac{ 
		\psi_I^A|_{i,n_y^A+1,k}^{n+\frac12} 
		- \psi_I^A|_{i,n_y^A+1,k-1}^{n+\frac12}
	}{\Delta z}
	\Bigg]
	+ \dx\dfrac{\dy}2 \dz \ U^A|_{i,n_y^A+1,k} \psi_I^A|_{i,n_y^A+1,k}^{n+\frac12}\,.
\end{multline}
The corresponding equation for the south face of Region~B is
\begin{multline}
	\label{eq:B-south}	\hbar \ \dx \dfrac{\dy}2 \dz \dfrac{\psi_R^B|_{i,1,k}^{n+1} - \psi_R^B|_{i,1,k}^{n}}{\dt} 
	= -\frac{\hbar^2}{2m} 
	\Bigg[
	\dfrac{\dy}2 \dz\dfrac{ 
		\psi_I^B|_{i+1,1,k}^{n+\frac12} 
		- \psi_I^B|_{i,1,k}^{n+\frac12}
	}{\Delta x}
	- 
	\dfrac{\dy}2 \dz\dfrac{ 
		\psi_I^B|_{i,1,k}^{n+\frac12} 
		- \psi_I^B|_{i-1,1,k}^{n+\frac12}
	}{\Delta x}	\\
	+ 
	\dx \dz\dfrac{ 
		\psi_I^B|_{i,2,k}^{n+\frac12} 
		- \psi_I^B|_{i,1,k}^{n+\frac12}
	}{\Delta y}	
	- 
	\dx \dz [\partial_y\psi_I^B]_{i,1,k}^{n+\frac12}
	+ 
	\dx\dfrac{\dy}2\dfrac{ 
		\psi_I^B|_{i,1,k+1}^{n+\frac12} 
		- \psi_I^B|_{i,1,k}^{n+\frac12} 
	}{\Delta z}\\
	-
	\dx\dfrac{\dy}2\dfrac{ 
		\psi_I^B|_{i,1,k}^{n+\frac12} 
		- \psi_I^B|_{i,1,k-1}^{n+\frac12}
	}{\Delta z}
	\Bigg]
	+ \dx\dfrac{\dy}2 \dz \ U^B|_{i,1,k} \psi_I^B|_{i,1,k}^{n+\frac12}\,.
\end{multline}
Adding \eqref{eq:A-north} and \eqref{eq:B-south} and using \eqref{eq:connection-ab-first}, \eqref{eq:connection-ab-second}, and \eqref{eq:connection-ab-last}, 
\begin{multline}
	\label{eq:AB}
	\hbar \ \dx \dy \dz \dfrac{\psi_R^B|_{i,1,k}^{n+1} - \psi_R^B|_{i,1,k}^{n}}{\dt} 
	= -\frac{\hbar^2}{2m} 
	\Bigg[
	\dy \dz\dfrac{ 
		\psi_I^B|_{i+1,1,k}^{n+\frac12} 
		- \psi_I^B|_{i,1,k}^{n+\frac12}
	}{\Delta x}
	- 
	\dy \dz\dfrac{ 
		\psi_I^B|_{i,1,k}^{n+\frac12} 
		- \psi_I^B|_{i-1,1,k}^{n+\frac12}
	}{\Delta x}	\\
	+ 
	\dx \dz\dfrac{ 
		\psi_I^B|_{i,2,k}^{n+\frac12} 
		- \psi_I^B|_{i,1,k}^{n+\frac12}
	}{\Delta y}
	- 
	\dx \dz\dfrac{ 
		\psi_I^B|_{i,1,k}^{n+\frac12} 
		- \psi_I^A|_{i,n_y^A,k}^{n+\frac12}
	}{\Delta y}
	+ 
	\dx\dy\dfrac{ 
		\psi_I^B|_{i,1,k+1}^{n+\frac12} 
		- \psi_I^B|_{i,1,k}^{n+\frac12} 
	}{\Delta z}\\
	-
	\dx\dy\dfrac{ 
		\psi_I^B|_{i,1,k}^{n+\frac12} 
		- \psi_I^B|_{i,1,k-1}^{n+\frac12}
	}{\Delta z}
	\Bigg]
	+ \dx\dy \dz \ \dfrac{U^A|_{i,n_y^A+1,k}+U^B|_{i,1,k}}{2} \psi_I^B|_{i,1,k}^{n+\frac12}\,,
\end{multline}
which has the exact same form as the FDTD equation for internal nodes~\eqref{eq:scalar-r-internal}, with the potential taken as the average of the two adjacent nodes\footnote{The derivation of \eqref{eq:AB} is analogous to the treatment of material interfaces in FDTD for electromagnetics~\cite{gedney-2011}. We also remark that the similarity between \eqref{eq:AB} at the interface and \eqref{eq:scalar-r-internal} for internal nodes in FDTD-Q makes existing results on FDTD-Q stability~\cite{soriano-analysis-2004,dai-stability-2005} applicable to this particular scenario and the specific boundary conditions \eqref{eq:connection-ab-first}--\eqref{eq:connection-ab-last} selected at the interface. However, as we discuss in Section~\ref{sec:joint-regs-stability-analysis} the proposed approach could be applied to developing other schemes and this example serves as an illustration.}. The same can be shown for the equations corresponding to \eqref{eq:cont-b}.

\subsubsection{Stability analysis}
\label{sec:joint-regs-stability-analysis}

Assume that the time step satisfies the generalized CFL limit \eqref{eq:cfl-gen} for each region
\begin{equation}
	\label{eq:minCFLgen-AB}
	\dt < \min\left(\dt_{\text{CFL, gen}}^A, \dt_{\text{CFL, gen}}^B\right)\,.
\end{equation}
From Theorem~\ref{thm:probability-conserv}, the probability in each region evolves according to
\begin{align}
	\dfrac{
		\mathcal P_{A}^{n+1} - \mathcal P_{A}^n
	}{\dt} &= -\mathcal I_{P,A}^{n+\frac12} \,,\label{eq:joint-regions-stability-proof-balance-A}
	\\
	\dfrac{
		\mathcal P_{B}^{n+1} - \mathcal P_{B}^n
	}{\dt} &= -\mathcal I_{P,B}^{n+\frac12}\label{eq:joint-regions-stability-proof-balance-B}\,.
\end{align}
Adding \eqref{eq:joint-regions-stability-proof-balance-A} and \eqref{eq:joint-regions-stability-proof-balance-B},
\begin{equation}
	\dfrac{
		\left(\mathcal P_{A}^{n+1} + \mathcal P_{B}^{n+1} \right) 
		- 
		\left(\mathcal P_{B}^n+\mathcal P_{B}^n\right)
	}{\dt} = 
	-\mathcal I_{P,A}|^{n+\frac12}
	-\mathcal I_{P,B}|^{n+\frac12} \,.\label{eq:joint-regions-stability-proof-balance-tot}
\end{equation}
As can be seen from \eqref{eq:current-expanded-0}, \eqref{eq:current-expanded-sides}, and \eqref{eq:current-expanded-w}, the conditions \eqref{eq:connection-ab-first}--\eqref{eq:connection-ab-last} equating the wavefunction and the hanging variables on the adjacent boundaries of the two regions ensure that the probability currents on the right hand side of \eqref{eq:joint-regions-stability-proof-balance-tot} cancel out. Hence, 
\begin{equation}\label{eq:joint-regions-tot-prob}
	\mathcal P_{A}^{n} + \mathcal P_{B}^{n}
	= 
	\mathcal P_{A}^{0} + \mathcal P_{B}^{0}\,.
\end{equation}
From Theorem~\ref{thm:probability-conserv}, under the CFL limit $\mathcal P_{A}^{n}$, $\mathcal P_{B}^{n}$ are both non-negative. Hence, from \eqref{eq:joint-regions-tot-prob}, $\mathcal P_{A}^{n}$, and $\mathcal P_{B}^{n}$ are each at most $\mathcal P_{A}^{0} + \mathcal P_{B}^{0}$. Repeating the reasoning in the proof of Lemma~\ref{lemma:stability}, 
\begin{subequations}
	\begin{align}
		||\PSIbig_A^n||_2^2\le \dfrac{\mathcal P_A^0 + \mathcal P_B^0}{\lambda_{\min}(\Pbig)}\,,\\
		||\PSIbig_B^n||_2^2\le \dfrac{\mathcal P_A^0 + \mathcal P_B^0}{\lambda_{\min}(\Pbig)}\,.
	\end{align}
\end{subequations}
This means that the system is stable, as the values of  the wavefunction samples cannot grow without bound. 
In Section~\ref{sec:ne-tunneling} we investigate the consequences of taking the time step beyond \eqref{eq:minCFLgen-AB}. In particular, we show that violating the generalized CFL limit in a region allows that region to provide infinite probability to the surrounding space and thus destabilize the simulation. Lastly, we remark that under condition \eqref{eq:minCFLgen-AB}, the coupled scheme can be shown to also conserve energy using similar arguments.

This example, although simple, shows how with the proposed theory one can create composite FDTD schemes obtained by coupling different models discretizing the Schrödinger equation. If each of the models satisfies the conservation of probability and the models are coupled in the probability-conserving manner, the resulting scheme will by construction satisfy the probability conservation. The resulting scheme will be stable under the most restrictive value of the generalized CFL limit \eqref{eq:minCFLgen-AB}, which is also known by construction. The choice of coupling between the models, such as the boundary conditions \eqref{eq:connection-ab-first}--\eqref{eq:connection-ab-last}, is essential for ensuring the probability conservation and stability. In general, a different coupling scheme is not guaranteed to be probability-conserving. 

The proposed approach could be applied to developing more advanced schemes in the future. 
For example, subgridding scenarios~\cite{salehi-2020-jce-polar,okoniewski-3d-subgridding-1997} could be analyzed as a connection of grids of different resolution~\cite{jnl-2018-fdtd-3d-dissipative}, where one needs to ensure that the grids exchange probability in a conserving manner to achieve the cancellation on the right hand side of \eqref{eq:joint-regions-stability-proof-balance-tot}. Another approach that could be used to analyze general schemes is the iteration matrix method~\cite{remis-2000-jcp,gedney-2011,decleer-2021-jcam}. The proposed conservation-based approach provides an intuitive physical interpretation that the root cause of instability is the generation of spurious energy or probability. Moreover, the proposed approach provides a very natural way to determine whether an FDTD-Q region or any other part of the setup is capable of introducing instability, prior to knowing anything about the overall setup. In general, such modular stability analysis is not trivial, since stability is a property of the entire system and coupling of equations corresponding to stable schemes does not automatically achieve stability. The use of concepts of conservation provides a systematic, modular, and constructive strategy to create stable composite FDTD schemes for the Schrödinger equation with a guarantee of probability and energy conservation.

\section{Numerical examples}
\label{sec:ne}

In order to investigate the validity of the results in Section~\ref{sec:probability-conserv} and Section~\ref{sec:energy-conserv}, the proposed method was implemented in Matlab and the time evolution of numerical probability and energy was investigated for different simulation scenarios.

\subsection{Infinite well}
\label{sec:ne-infinite-well}

First, we demonstrate the validity of the proposed theory for an isolated region. In particular, we consider an electron trapped in a potential well with the potential $U$ equal to zero inside the region and tending to infinity on the boundary of the region and outside the region. The infinite potential well results in zero Dirichlet boundary conditions~\cite{miller-2008}. The region has a side length of $a = 30$~nm in each dimension. 

The region was discretized into $n_x = n_y = n_z = 30$ primary cells. The CFL limit $\dt_{\text{CFL}}$ and the generalized CFL limit $\dt_{\text{CFL,gen}}$ were both equal to 2.879~fs, with the relative difference (1.370$\times$10$^{-15}$) comparable to machine precision. The time step $\dt$ was taken as 0.999~$\dt_{\text{CFL}}$.
The initial conditions $\psi_R|^0$ and $\psi_I|^{-0.5}$ were set by sampling the following particular solution of the Schrödinger equation~\cite{miller-2008}
\begin{equation}\label{eq:ne-infinite-well-analytical}
	\psi(x,y,z,t) = A \sin\left(k_x x\right) 
	\sin\left(k_y y\right) 
	\sin\left(k_z z\right) \exp\left(-\mathrm{i} \left(\dfrac{E_1}{\hbar} t+\frac{\pi}3\right)\right)
\end{equation}
at the primary nodes. In \eqref{eq:ne-infinite-well-analytical}, 
\begin{equation}
	\mathrm{i} = \sqrt{-1}\,,
\end{equation}
\begin{equation}
	k_x = k_y = k_z = \dfrac{\pi}{a} \,,
\end{equation}
\begin{equation}\label{eq:infinite-well-E1}
	E_1 = \dfrac{\hbar^2}{2m}(k_x^2 + k_y^2 + k_z^2)\,,
\end{equation}
and $A$ is the real positive normalization constant ensuring that $\mathcal P^0$ in \eqref{eq:probability} is equal to 1. This choice of normalization will be discussed shortly.

\begin{figure}
	\centering
	\includegraphics[scale=1]{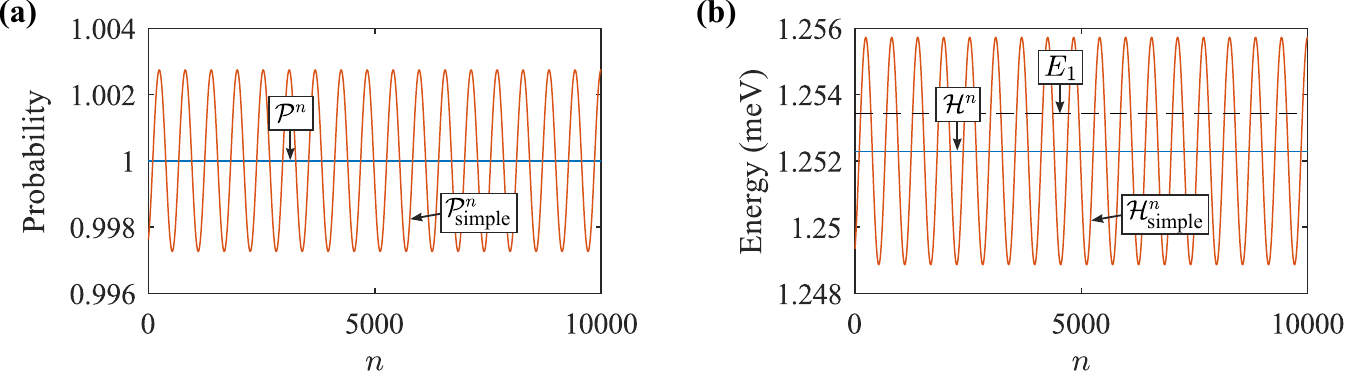}
	\caption{Infinite well example in Section~\ref{sec:ne-infinite-well} for ${n_x=n_y=n_z=30}$: (a) the total probability ${\mathcal P^n}$ computed with the proposed formula \eqref{eq:probability} and ${\mathcal P_{\text{simple}}^n}$ computed with the simpler expression \eqref{eq:Pnaive} (b) the total energy ${\mathcal H^n}$ in \eqref{eq:energy} and ${\mathcal H_{\text{simple}}^n}$ in \eqref{eq:H-naive}. The dashed line shows the analytic value of energy \eqref{eq:infinite-well-E1}.}
	\label{fig:ne-infinite-well}
\end{figure} 

\begin{table}
	\centering
	\caption{Error of energy expressions vs grid refinement in the example in Section~\ref{sec:ne-infinite-well} (${a=}$~30~nm, ${\dt}$ = 0.999${\dt}_{\text{CFL}}$, ${n_t\Delta t}$ = 28.76 ps).}
	\label{tab:ne-infinite-well}
	\begin{tabular}{|c|cc|cc|}
		\hline
		$n_x=n_y=n_z$ 
		& $\max_n|\mathcal P^n-1|$ & $\max_n|\mathcal P_{\text{simple}}^n-1|$
		& $\min_n|\mathcal H^n-E_1|/E_1$ & $\max_n|\mathcal H_{\text{simple}}^n-E_1|/E_1$
		\\\hline\hline
		10&	8.88$\times$10$^{-16}$&	2.51$\times$10$^{-2}$&	8.20$\times$10$^{-3}$ & 3.19$\times$10$^{-2}$\\\hline
		20&	5.55$\times$10$^{-16}$&	6.19$\times$10$^{-3}$&	2.05$\times$10$^{-3}$ & 8.15$\times$10$^{-3}$\\\hline
		30&	
		2.22$\times$10$^{-15}$&	2.74$\times$10$^{-3}$& 9.14$\times$10$^{-4}$ & 3.64$\times$10$^{-3}$\\\hline
		40&	1.55$\times$10$^{-15}$&	1.54$\times$10$^{-3}$& 5.14$\times$10$^{-4}$ & 2.05$\times$10$^{-3}$ \\\hline
		50&	3.44$\times$10$^{-15}$&	9.87$\times$10$^{-4}$& 3.29$\times$10$^{-4}$ & 1.31$\times$10$^{-3}$ \\\hline
	\end{tabular}
\end{table}

Since the region is isolated to the flow of power and probability current, the energy and probability contained inside the region are constant in the continuous domain. From the blue curves in Fig.~\ref{fig:ne-infinite-well}, it is evident that the proposed expressions for probability and energy respect this property in the discrete domain, in accordance with Theorem~\ref{thm:probability-conserv} and Theorem~\ref{thm:energy-conserv}. Indeed, the range of values of $\mathcal P^n$ and $\mathcal H^n$ was 3.997$\times$10$^{-15}$ and 3.228~aeV, which could be explained by finite machine precision. In contrast, the values of $\mathcal P_{\text{simple}}^n$ and $\mathcal H_{\text{simple}}^n$ shown in red in the figure exhibit some fluctuations. Moreover, the values of $\mathcal P_{\text{simple}}^n$ exceed 1 in some of the time intervals, which is inconsistent with physics.  

The continuous expression for energy \eqref{eq:hcont} can be shown to equal $E_1=$~1.2534~meV, which is plotted in Fig.~\ref{fig:ne-infinite-well}(b) with the dashed line. The value of the proposed expression, $\mathcal H^n = $~1.2523~meV, is in good agreement with $E_1$, corresponding to a relative error of 9.14$\times$10$^{-4}$.
Table~\ref{tab:ne-infinite-well} shows deviations of the values of different probability and energy expressions from the analytic solution. The values were obtained for different grid resolutions, while maintaining the infinite well geometry and keeping the constant ratio $\dt/\dt_{\text{CFL}}$.
As evident from the table, the discrepancy between $\mathcal H^n$ and $E_1$ can be reduced by refining the cell size and the corresponding time step. The error in $\mathcal P_{\text{simple}}^n$ and $\mathcal H_{\text{simple}}^n$ also reduces with improved spatial resolution. However, both $\mathcal P_{\text{simple}^n}^n$ and $\mathcal H_{\text{simple}}^n$ exhibit larger errors than the proposed expressions. 

Hence the proposed expressions of probability and energy provide a more accurate approximation of the analytical quantities and respect the principle of probability and energy conservation, in contrast to $\mathcal P_{\text{simple}}^n$ and $\mathcal H_{\text{simple}}^n$. The computational cost associated with evaluating $\mathcal P^n$ is increased compared to $\mathcal P_{\text{simple}}^n$ due to the off-diagonal blocks in $\Pbig$ in \eqref{eq:mtx-P}. However, the cost of the added computations is on par with evaluating the terms associated with the diagonal blocks in $\Pbig$, even if one evaluates the additional terms directly. Similarly, the overhead in computing the additional terms in $\mathcal H^n$ is comparable to the cost of computing the terms that are in common with $\mathcal H^n_{\text{simple}}$. Lastly, the proposed expression of the total probability is very convenient for computing the normalization constant $A$ in \eqref{eq:ne-infinite-well-analytical}, since the value of $\mathcal P^n$ stays constant in an isolated region. In contrast, the value $\mathcal P_{\text{simple}}^n$ varies from time step to time step. As evident from Fig.~\ref{fig:ne-infinite-well}(a), if the initial values of the wavefunction were normalized such that $\mathcal P_{\text{simple}}^0$ equaled 1, both $\mathcal P^n$ and $\mathcal P_{\text{simple}}^n$ would have been centered around a value that greater than 1.

\subsection{Reflection from a potential barrier}
\label{sec:ne-reflection}

We consider the scenario in Fig.~\ref{fig:ne-reflection-illustration}, where a Gaussian wavepacket impinges on a potential barrier. The expression for the incident pulse is~\cite{miller-2008}
\begin{equation}
	\psi_{\text{inc}}(x,t) = 
	\sum_{p} A_p e^{
		-\mathrm{i}(\omega_p t - k_p(x-x_0))}\,,
\end{equation}
where $p$ is an integer index, $x_0$ is the center of the wavepacket at $t = 0$, $k_p$ are evenly spaced real scalars, $\omega_p$ are the corresponding angular frequencies given by $\hbar k_p^2/(2m)$, and the coefficients $A_p$ are given by
\begin{equation}
	A_p =  \exp\left(-\dfrac{1}{4}\left(\dfrac{k_p-\bar{k}}{\sigma}\right)^2\right)\,,
\end{equation}
where $\bar k$ determines the center of the wavepacket in the $k$-space and $\sigma$ determines its width. The wavepacket impinges on a barrier with potential given by
\begin{equation}
	U(x) = \begin{cases}
		0 \,,& \quad x < a\\
		\dfrac{U_0}2 \,, & \quad x = a\\
		U_0 \,,& \quad x > a
	\end{cases}\,.
\end{equation}

\begin{figure}
	\centering
	\includegraphics[width=5.8in]{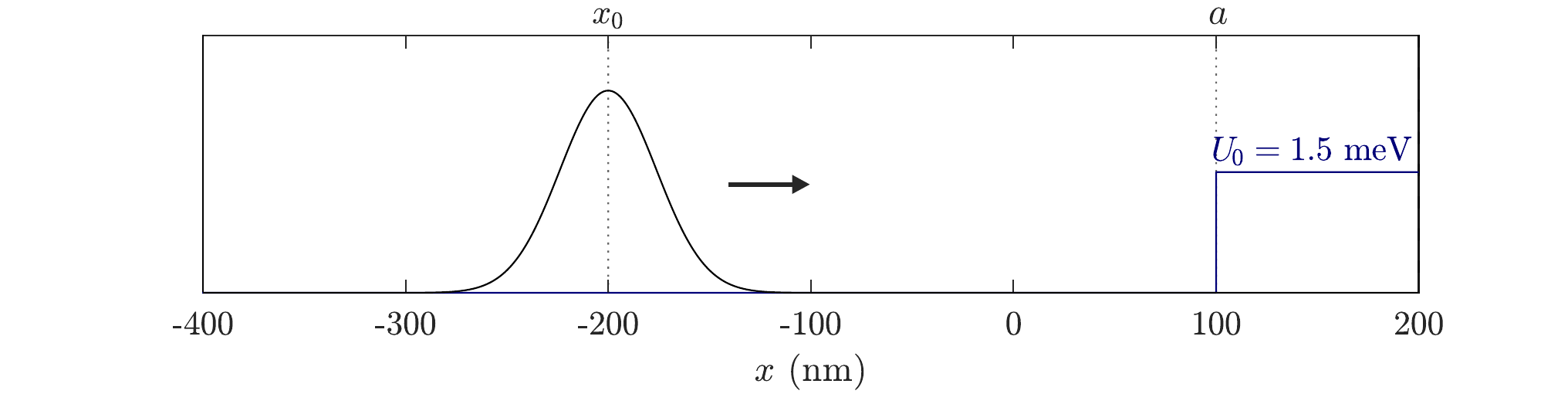}	
	\caption{Scenario considered in Section~\ref{sec:ne-reflection}, where a Gaussian wavepacket impinges on a potential barrier.}
	\label{fig:ne-reflection-illustration}
\end{figure}

The solution can be found in quantum mechanics textbooks~\cite{miller-2008} and is given by
\begin{equation}\label{eq:reflection-analyt-sol}
	\psi(x,t) = \begin{cases}
		\psi_{\text{inc}}(x,t) + \psi_{\text{ref}}(x,t) \,, &  x < a\\
		\psi_{\text{tran}}(x,t) & x > a
	\end{cases}\,,
\end{equation}
where $\psi_{\text{ref}}$ and $\psi_{\text{tran}}$ are reflected and transmitted waves given by
\begin{equation}
	\psi_{\text{ref}}(x,t) = \sum_{p} A_p R_p e^{-\mathrm{i}\left(\omega_p t - k_p(a-x_0 + a-x)\right)}\,,
\end{equation}
\begin{equation}
	\psi_{\text{tran}}(x,t) = \sum_p A_p T_p e^{-\mathrm{i}\left(\omega_p t - k_p(a-x_0) - K_p (x-a)\right)}\,,
\end{equation}
where
$K_p = \sqrt{2m(\hbar \omega_p - V_0)}/\hbar$, which can be real or imaginary, depending on the sign of $\hbar \omega_p - V_0$. When $K_p$ is real, the corresponding wave in $\psi_{\text{tran}}$ propagates forward in the $x>a$ region. When $K_p$ is imaginary, the wave decays with $x$ and hence cannot propagate. The reflection and transmission coefficients are given by
\begin{equation}
	R_p = \dfrac{k_p - K_p}{k_p + K_p}\,,
\end{equation}
\begin{equation}
	T_p = \dfrac{2 k_p}{k_p + K_p}\,.
\end{equation}
In this test, $m$ is the mass of an electron, $x_0 = -200$~nm, $\bar k = 2\pi/\bar{\lambda}$, where $\bar{\lambda}$ is 30~nm, and $\sigma=\bar k/10$.
The values of $k_p$ range between $k_{\min} = \bar k-10 \sigma$ and $k_{\max} = \bar k + 10 \sigma$, with a spacing of $\Delta k = \sigma/100$. The height of the potential barrier is $U_0 =$~1.5~meV and $a =$~100~nm. 

We select the space from $x=0$ to $x = $~200~nm to be modeled with the proposed method. The dimensions of the region were selected as 2~nm in the $y$ and $z$ dimensions. The region was discretized into cells of size $\dx = \dy = \dz=$~1~nm. The initial conditions in the region were set by sampling the analytical solution \eqref{eq:reflection-analyt-sol}. The update equations on the boundary were the modified FDTD-Q equations described in Section~\ref{sec:scalar-equations}, such as \eqref{eq:scalar-r-bottom}, \eqref{eq:scalar-r-bottom-east}, and \eqref{eq:scalar-r-bottom-south-east}, which involved the hanging variables. 
In order to obtain the values of the hanging variables, expressions were found for the normal component of the gradient of the analytical solution \eqref{eq:reflection-analyt-sol}. The expressions were evaluated at integer time points $n \dt$ for the real part and at $(n+0.5)\dt$ for the imaginary part. 
The values of the hanging variables were zero on all faces of the boundary except for the east and west, where their values were uniform in the $y$ and $z$ direction.
This, in essence, rendered the problem one-dimensional, which was the reason why choosing the $y$ and $z$ dimensions to be small (2~nm) was possible.
The generalized CFL limit in \eqref{eq:cfl-gen} was $\dt_{\text{CFL,gen}}$~=~2.869968~fs. The CFL limit $\dt_{\text{CFL}}$ in \eqref{eq:cfl} was slightly lower (2.869915~fs), which is in agreement with the theory. The simulation time step was set to 0.999$\dt_{\text{CFL}}$.

\begin{figure}[t!]
	\centering
	\includegraphics[width=6.1in]{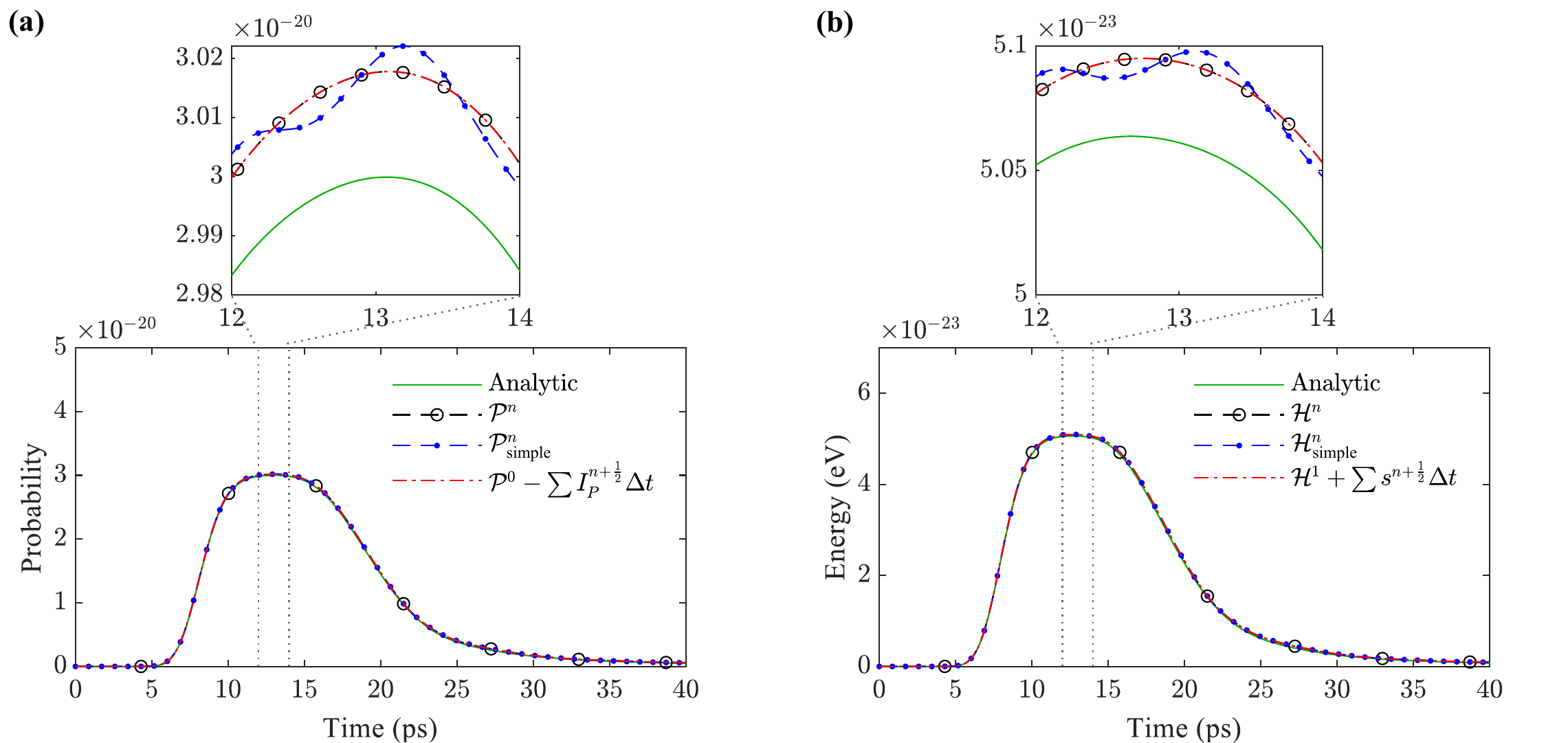}	
	\caption{Results of the test in Section~\ref{sec:ne-reflection}: (a) total probability in the region (b) total energy associated with the region.}
	\label{fig:ne-reflection-results}
\end{figure}
\begin{figure}
	\centering
	\includegraphics[width=6.7in]{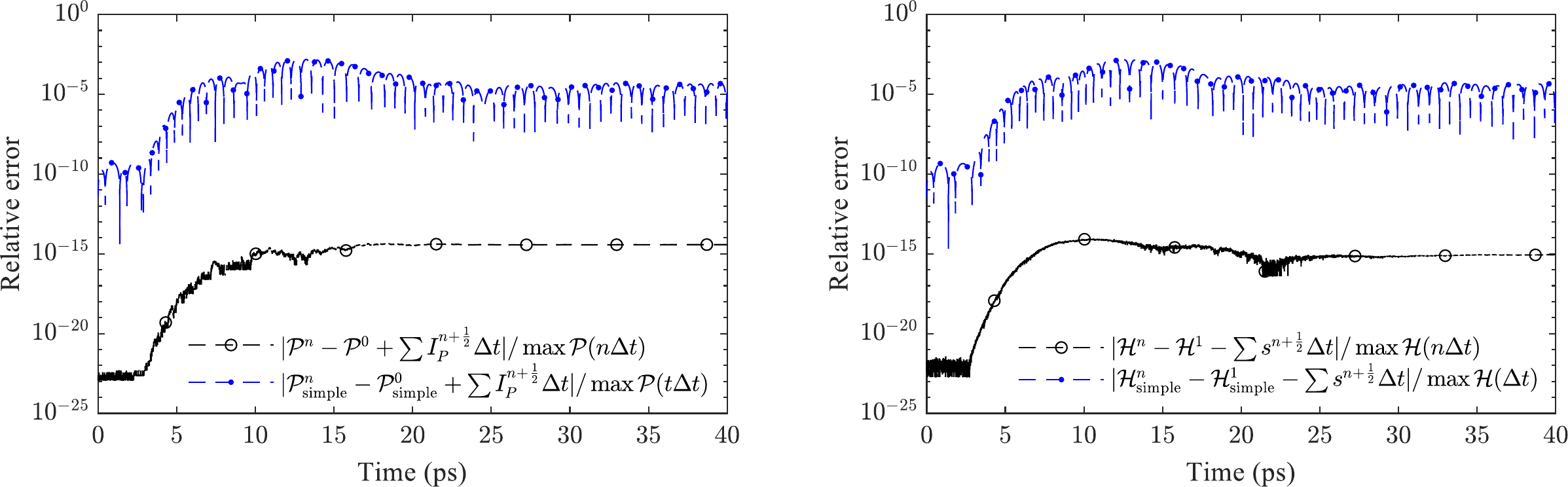}	
	\caption{Results of the test in Section~\ref{sec:ne-reflection}: (a) relative error in \eqref{eq:ne-reflection-prob-conserv} (b) relative error in \eqref{eq:ne-reflection-ener-conserv}.}
	\label{fig:ne-reflection-error-in-conserv}
\end{figure}

Fig.~\ref{fig:ne-reflection-results}(a) shows the probability of finding the particle in the region obtained via analytical computation, via the proposed expression $\mathcal P^n$ in \eqref{eq:probability}, and the simple expression $\mathcal P^n_{\text{simple}}$ in \eqref{eq:Pnaive}. The analytic values were obtained from \eqref{eq:probability-continuous}, applied to the analytical solution \eqref{eq:reflection-analyt-sol}. The integration in \eqref{eq:probability-continuous} was approximated via a Riemann sum using the midpoint rule with 1000 intervals between $x=$~0~nm and $x=$~200~nm. When using the analytical solution as a reference, the accuracy of $\mathcal P^n$ and $\mathcal P_{\text{simple}}^n$ was comparable. The maximum relative errors were 8.827$\times$10$^{-3}$ for $\mathcal P^n$, 8.978$\times$10$^{-3}$ for $\mathcal P_{\text{simple}}^n$. Fig.~\ref{fig:ne-reflection-results}(b) shows the energy stored in the region computed using $\mathcal H^n$ and $\mathcal H_{\text{simple}}^n$, as well as the analytic values approximated via the Riemann sum. The accuracy of $\mathcal H^n$ and $\mathcal H_{\text{simple}}^n$ was also comparable, with a maximum relative error associated with $\mathcal H^n$ of 1.188$\times$10$^{-2}$ and the error associated with $\mathcal H_{\text{simple}}^{n}$ of 1.197$\times$10$^{-2}$.

As predicted by Theorem~\ref{thm:probability-conserv}, the total probability $\mathcal P^n$ was non-negative, with the smallest value of 3.184$\times$10$^{-27}$. 
The smallest value of energy $\mathcal H^n$ was 4.812$\times$10$^{-30}$~eV, which was higher than the bound of $-$6.193$\times$10$^{-17}$~eV predicted by Theorem~\ref{thm:energy-conserv} when taking $\mathcal P_{\max}$ in \eqref{eq:H-lower-bound} to be the highest value of $\mathcal P^n$ over the course of the simulation. 

Based on Theorem~\ref{thm:probability-conserv}, the values of the proposed expression $\mathcal P^n$ must satisfy
\begin{equation}
	\label{eq:ne-reflection-prob-conserv}
	\mathcal P^n = \mathcal P^0 - \sum_{n'=0}^{n-1} \mathcal I_P^{n'+\frac12} \dt \,, 
	\quad \forall n = 1, \hdots, n_t \,,
\end{equation}
where the right hand side is the sum of the initial probability and the probability that has entered the region due to the flow of the probability current. The equality in \eqref{eq:ne-reflection-prob-conserv} is evident from Fig.~\ref{fig:ne-reflection-results}(a), showing both the left and right hand sides of \eqref{eq:ne-reflection-prob-conserv}.
The relative error between the two sides of \eqref{eq:ne-reflection-prob-conserv} is plotted in Fig.~\ref{fig:ne-reflection-error-in-conserv}(a). The largest relative error was 4.514$\times$10$^{-15}$, which is comparable to the machine precision and hence corroborates the prediction. The relative error in Fig.~\ref{fig:ne-reflection-error-in-conserv}(a) was normalized to the largest value of the theoretical probability over all time steps $\max \mathcal \{\mathcal P(n\Delta t)\}=$~3.000$\times$10$^{-20}$.
In contrast to the proposed expression for the total probability, when using $\mathcal P_{\text{simple}}^n$ in place of $\mathcal P^n$ in \eqref{eq:ne-reflection-prob-conserv}, the left and right hand sides of the relation are no longer equal, as can be seen from Fig.~\ref{fig:ne-reflection-error-in-conserv}(a). The maximum value of the relative error was 1.519$\times$10$^{-3}$, which is much larger than the error in \eqref{eq:ne-reflection-prob-conserv} associated with $\mathcal P^n$. Hence, at least when $\mathcal I_P^{n+0.5}$ is used as the representation of the probability current, $\mathcal P_{\text{simple}}^n$ does not possess the conservation properties exhibited by $\mathcal P^n$.
These errors in conservation when using the simpler expression $\mathcal P_{\text{simple}}^n$ also manifested themselves in the small fluctuations visible in the inset of Fig.~\ref{fig:ne-reflection-results}(a). 

Similarly, from Fig.~\ref{fig:ne-reflection-results}(b) and Fig.~\ref{fig:ne-reflection-error-in-conserv}(b) one can observe that the following relation is satisfied by the proposed expressions for the total energy and supplied power:
\begin{equation}
	\label{eq:ne-reflection-ener-conserv}
	\mathcal H^n = \mathcal H^1 + \sum_{n'=1}^{n-1} s^{n'+\frac12} \dt \,,
	\quad \forall n = 2, \hdots, n_t-1\,,
\end{equation}
with the largest relative error of 9.111$\times$10$^{-15}$, which is comparable to machine precision.
In contrast, replacing $\mathcal H^n$ with $\mathcal H_{\text{simple}}^n$ results in the relative error of 1.428$\times$10$^{-3}$ between the left and right hand sides of \eqref{eq:ne-reflection-ener-conserv}, which can no longer be explained by the finite machine precision. The normalization constant for the relative error in Fig.~\ref{fig:ne-reflection-error-in-conserv}(b) was $\max \{\mathcal H(n\Delta t)\}=$~5.064$\times$10$^{-23}$~eV. Small fluctuations that were exhibited by $\mathcal P_{\text{simple}}^n$ can also be observed for $\mathcal H_{\text{simple}}^n$ in the inset of Fig.~\ref{fig:ne-reflection-results}(b).
Hence, the results confirm that both $\mathcal P^n$ and $\mathcal H^n$ satisfy \eqref{eq:ne-reflection-prob-conserv} and \eqref{eq:ne-reflection-ener-conserv} and thus \eqref{eq:probability-balance} and \eqref{eq:energy-balance}, respectively. This property makes the proposed expressions preferable to $\mathcal P_{\text{simple}}^n$ and $\mathcal H_{\text{simple}}^n$.

\subsection{Proton tunneling}
\label{sec:ne-tunneling}

In this section we consider the scenario illustrated in Fig.~\ref{fig:ne-tunneling-setup}, which was taken from~\cite{cabonell-kostin-1973-ijqc-tunneling} and can be used as a simplified model for hydrogen transfer reactions. The problem consists of three regions: reactant (r), barrier (b), and product (p), and a proton that can transfer between the reactant and product regions via tunneling. The size of the reactant, barrier, and product regions is $l_x^{\text r}\times l_y \times l_z$, $l_x^{\text b}\times l_y \times l_z$, and $l_x^{\text p}\times l_y \times l_z$, respectively. The barrier region has a potential of $U_0$ and the other two regions have the potential equal to zero. Dirichlet zero boundary conditions are imposed on the external boundaries of the regions. On the interfaces, the continuity of the wavefunction and the normal component of its gradient is imposed. Using the separation of variables, the analytical solution in each region can be found as~\cite{cabonell-kostin-1973-ijqc-tunneling}
\begin{equation}\label{eq:ne-tunneling-analytic}
	\psi^{\text r, \text b, \text p}(x,y,z,t) = \sum_{m=1}^{N} M_m f_m^{\text r, \text b, \text p}(x) g_m(y) h_m(z) \exp\left(-\mathrm{i} \dfrac{E_m}{\hbar} t\right)\,, 
	\quad 0 \le x \le l_{x}^{\text r, \text b, \text p}\,,
\end{equation}
where $N$ is the number of eigenfunctions considered. The expression for $f_m^{\text r, \text b, \text p}(x)$ in each region is
\begin{subequations}
	\begin{align}
		&f_m^{\text r}(x) = A_m \sin\left(k_{xm}x\right) \,,\\ 
		&f_m^{\text b}(x) = 
		B_m \cosh\left(K_{xm}\left(x-\frac{l_x^{\text b}}{2}\right)\right)
		+C_m \sinh\left(K_{xm}\left(x-\frac{l_x^{\text b}}{2}\right)\right)\,,\\
		&f_m^{\text p}(x) = D_m \sin\left(k_{xm}\left(x-l_x^{\text p}\right)\right)\,,
	\end{align}
\end{subequations}
where the choice between $B_m = 0$ and $C_m=0$ determines the symmetry of the eigenfunction.
The coefficients $A_m$, $B_m$, $C_m$, and $D_m$ are chosen to ensure that the wavefunction is continuous across the region interfaces and that
\begin{equation}
	\int_{\text r} |f_m^{\text r}(x)|^2 dx + \int_{\text{b}} |f_m^{\text b}(x)|^2 dx + \int_{\text{p}} |f_m^{\text p}(x)|^2 dx = 1\,.
\end{equation}
The expressions for $g_m(y)$ and $h_m(z)$ are given by
\begin{equation}
	g_m(y) = \sqrt{\frac{2}{l_y}} \sin\left(k_{ym}y\right) \,,
\end{equation}
\begin{equation}
	h_m(z) = \sqrt{\frac{2}{l_z}} \sin\left(k_{zm}z\right)\,.
\end{equation}
The values of $k_{ym}$ and $k_{zm}$ are such that the zero Dirichlet boundary conditions are satisfied for $y=l_y$ and $z=l_z$. The values of $k_{xm}$ and $K_{xm}$ are related through 
\begin{equation}
	E_{xm} = \dfrac{\hbar^2 k_{xm}^2}{2m_P} =  V_0-\dfrac{\hbar^2 K_{xm}^2}{2m_P}\,,
\end{equation}
where $m_P$~=~1~dalton~$\approx$~1.661$\times$10$^{-27}$~kg. 
The value of $E_{x,m}$ can be obtained by imposing the continuity of the $x$ derivative of the wavefunction across the region interfaces.
The energy $E_m$ in \eqref{eq:ne-tunneling-analytic} corresponding to each eigenfunction is given by
\begin{equation}
	E_m = E_{xm} + E_{ym} + E_{zm}\,,
\end{equation}
where
\begin{align}
	&E_{ym} = \dfrac{\hbar^2 k_{ym}^2}{2m_P}\,,\\
	&E_{zm} = \dfrac{\hbar^2 k_{zm}^2}{2m_P}\,.
\end{align}

Following~\cite{cabonell-kostin-1973-ijqc-tunneling}, the coefficients $M_m$ in \eqref{eq:ne-tunneling-analytic} are assumed to have the following dependence on the temperature $T$:
\begin{equation}\label{eq:Mm}
	M_m = M M'_m\,,
\end{equation}
where
\begin{equation}
	M'_m = \exp\left(-\frac{E_m}{2Tk_{B}}+\mathrm{i}\delta_m\right)\,,
\end{equation}
where $k_{B}$ is the Boltzmann constant, and $\delta_m$ are phases that can be chosen arbitrarily. Scalar $M$ in \eqref{eq:Mm} is a normalization constant given by
\begin{equation}
	M = \dfrac{1}{\sqrt{\sum_{m=1}^{N} |M'_m|^2}}\,.
\end{equation}

The temperature and the number of eigenfunctions were taken from~\cite{cabonell-kostin-1973-ijqc-tunneling} as $T =$~298~K and $N=$~8, respectively. The dimensions of the regions, also from \cite{cabonell-kostin-1973-ijqc-tunneling}, are shown in Fig.~\ref{fig:ne-tunneling-setup}. Unlike \cite{cabonell-kostin-1973-ijqc-tunneling}, we only consider a particular solution with the phases $\delta_m$ in Table~\ref{tab:ne-tunneling-setup}, as opposed to an ensemble of tunneling systems with many sets of randomized phases $\delta_m$. In \cite{cabonell-kostin-1973-ijqc-tunneling}, the model was further extended to include the effect of thermal vibrations. 

\begin{figure}
	\centering
	\includegraphics[width=3.8in]{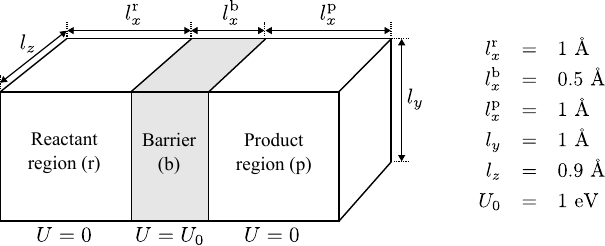}	
	\caption{The scenario from~\cite{cabonell-kostin-1973-ijqc-tunneling} considered in Section~\ref{sec:ne-tunneling}, modeling proton tunneling through a barrier with ${U = U_0}$.}
	\label{fig:ne-tunneling-setup}
\end{figure}

\begin{table}
	\centering
	\caption{Description of the solution modes considered (see \cite{cabonell-kostin-1973-ijqc-tunneling} for details).}\label{tab:ne-tunneling-setup}
	\begin{tabular}{|c|c|c|c|c|c|}\hline
		$m$ & $\delta_m$ & $E_{xm}$ (meV)& Symmetry of $f_m^{\text b}(x)$ &  $k_{ym}$ & $k_{zm}$\\\hline\hline
		1 & 0.15$\pi$	&18.858713602402805 & $B_m \ne 0$, $C_m = 0$ & $\pi/l_y$ & $\pi/l_z$\\
		2 & 0.95$\pi$ 	&18.858842481348997 & $B_m = 0$, $C_m \ne 0$ & $\pi/l_y$& $\pi/l_z$\\
		3 & 0.25$\pi$ 	&75.369931622707597 & $B_m \ne 0$, $C_m = 0$ & $\pi/l_y$& $\pi/l_z$\\
		4 & 1.1$\pi$ 	&75.370616971460279 & $B_m = 0$, $C_m \ne 0$ & $\pi/l_y$& $\pi/l_z$\\
		5 & 0 			&$E_{x1}$ & $B_m \ne 0$, $C_m = 0$ & $2\pi/l_y$& $\pi/l_z$\\
		6 & 1.3$\pi$ 	&$E_{x2}$ & $B_m = 0$, $C_m \ne 0$ & $2\pi/l_y$& $\pi/l_z$\\
		7 & 0 			&$E_{x1}$ & $B_m \ne 0$, $C_m = 0$ & $\pi/l_y$ & $2\pi/l_z$\\
		8 & 0.7$\pi$	&$E_{x2}$ & $B_m = 0$, $C_m \ne 0$ & $\pi/l_y$& $2\pi/l_z$\\\hline
	\end{tabular}
\end{table}

We use this scenario to examine the conservation of probability and energy when multiple FDTD-Q models are connected and to study the transfer of these quantities from one region's model to another.
Three FDTD-Q models were defined: one discretizing the reactant region, another discretizing the barrier, and the third discretizing the product region. The theoretical treatment of this scenario is described in Section~\ref{sec:stability-joint} for two regions and an extension to tree regions is analogous. The cell dimensions in each region were $\dx = \dy = \dz = (1/30)$~\AA. The three FDTD-Q models were coupled by equating the wavefunction values and hanging variables on the interfaces, as described in Section~\ref{sec:stability-joint}. 
The CFL limits and the generalized CFL limits in each region are shown in Table~\ref{tab:ne-tunneling-cfl}. Their values were the same, apart from round-off error. The simulation time step was taken as 0.999 of the CFL limit of the barrier region ($\dt \approx$~55.79~as), which automatically satisfied the CFL limit condition in the reactant and product regions. This time step ensures stability, according to~\cite{dai-stability-2005} or based on the discussion in Section~\ref{sec:stability-joint}. The initial conditions were set by sampling the analytical solution~\eqref{eq:ne-tunneling-analytic} and normalizing the result to achieve the total probability in the three regions $\mathcal P_{\text r}^0 + \mathcal P_{\text b}^0 + \mathcal P_{\text p}^0 = 1$.  

\begin{table}\centering
	\caption{CFL limit and generalized CFL limit in each region in the test of Section~\ref{sec:ne-tunneling}}
	\label{tab:ne-tunneling-cfl}
	\begin{tabular}{|c|c|c|c|}\hline
		Region 	   & Reactant & Barrier & Product \\\hline \hline
		$\dt_{\text{CFL,gen}}$ (as) & 58.318879645754564 & 55.844895610996367  & 58.318879645754564\\\hline 
		$\dt_{\text{CFL}}$ (as) & 58.318879645754819 & 55.844895610996282  &58.318879645754819\\\hline
	\end{tabular}
\end{table}

\begin{figure}
	\centering
	\includegraphics[scale=0.95]{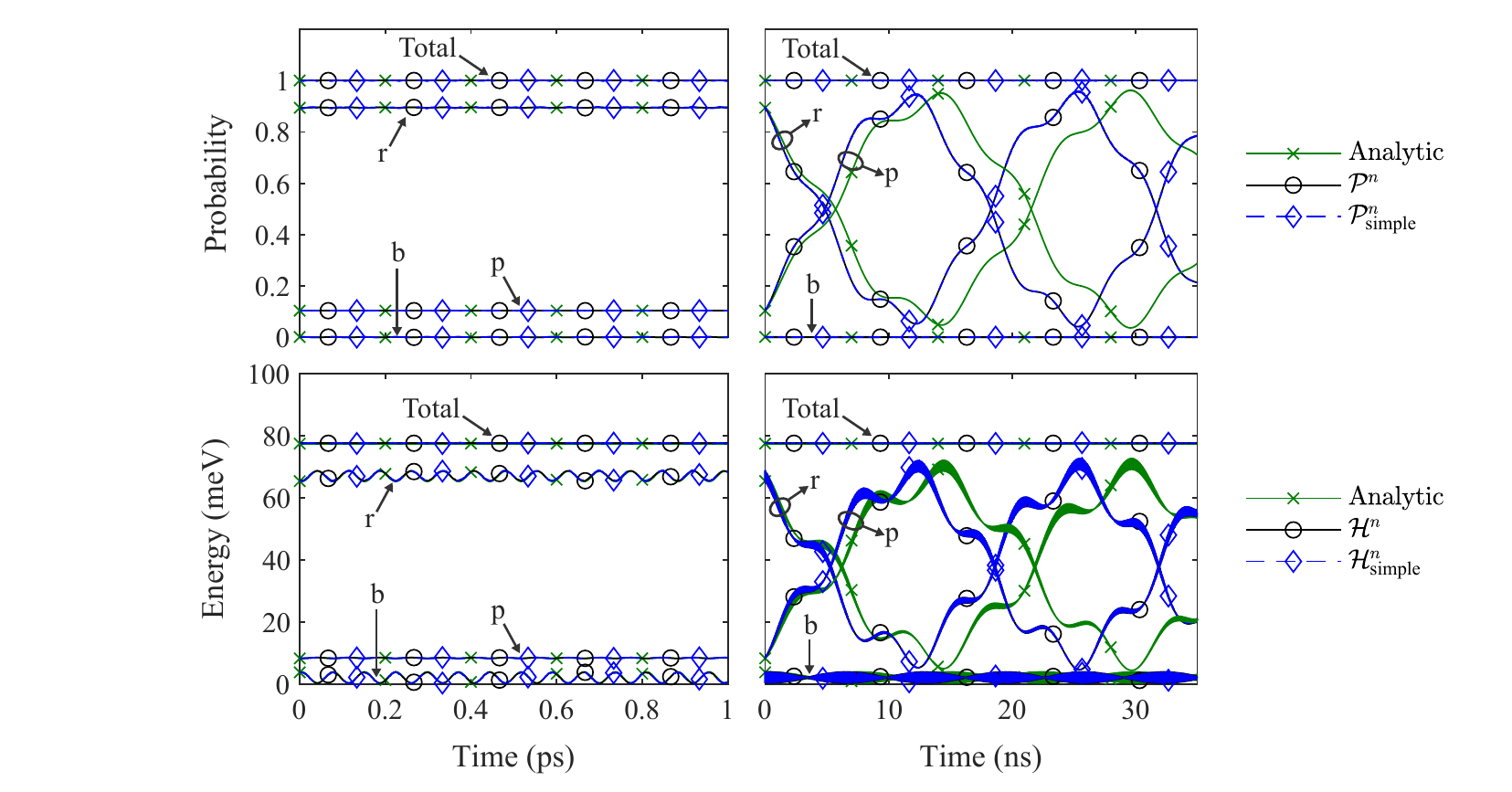}
	\caption{Energy and probability computed in the test of Section~\ref{sec:ne-tunneling}. Left panels: 1~ps simulation. Right panels: 35~ns simulation.}
	\label{fig:ne-tunneling-results}
\end{figure}

One simulation run was performed to study the first 1~ps of the particle's evolution and a longer simulation was done to observe the behavior of the system over 35~ns. In the longer simulation, probability and energy were computed once in ten time steps ($10\dt=$~0.5579~fs). This allowed reducing the memory usage while providing sufficient information, considering that the shortest period of a mode in \eqref{eq:ne-tunneling-analytic} was 29.26~fs. The results are shown in Fig.~\ref{fig:ne-tunneling-results}. During the first 1~ps, the simulated results match well with the analytical prediction. For the 35~ns simulation, the results deviate from the analytical solution, which is expected due to dispersion errors caused by the finite discretization of the Schrödinger equation. Nevertheless, the simulation was able to correctly model the overall trend of the solution. The accuracy of $\mathcal P_{\text{simple}}^n$ and $\mathcal H_{\text{simple}}^n$ for each region was comparable to that of $\mathcal P^n$ and $\mathcal H^n$.

\begin{figure}
	\centering
	\includegraphics[scale=0.7]{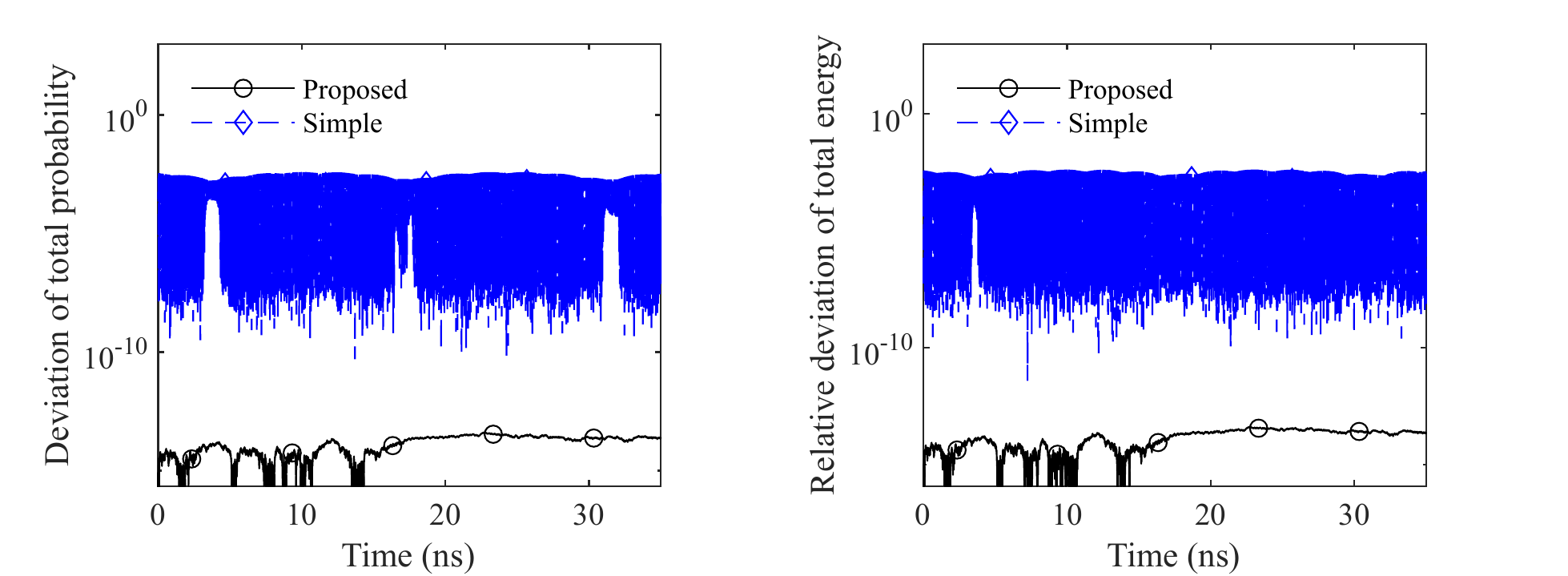}
	\caption{Relative deviation of the total probability and energy from their initial values (a) deviation of probability computed from \eqref{eq:ne-tunneling-rel-dev-from-init-prob-proposed} and \eqref{eq:ne-tunneling-rel-dev-from-init-prob-simple} (b) normalized deviation of energy computed from \eqref{eq:ne-tunneling-rel-dev-from-init-ener-proposed} and \eqref{eq:ne-tunneling-rel-dev-from-init-ener-simple}.}
	\label{fig:ne-tunneling-rel-dev-from-init}
\end{figure}

From Theorem~\ref{thm:energy-conserv}, using 1 as a bound on probability, the energy $\mathcal H^n$ cannot take on values below -1.608~keV for the reactant or product regions and below -54.10~keV for the barrier region. The values of energy in each region were positive, which is in agreement with these lower bounds. Probability in each region was also positive, which is consistent with Theorem~\ref{thm:probability-conserv}.

Based on the discussion in Section~\ref{sec:stability-joint}, the sum of the total probabilities over the three regions satisfies
\begin{equation}
	\label{eq:ne-tunneling-roc-prob}
	\dfrac{
		\left(\mathcal P_{\text r}^{n+1} + \mathcal P_{\text b}^{n+1} + \mathcal P_{\text p}^{n+1} \right) 
		- 
		\left(\mathcal P_{\text r}^{n} + \mathcal P_{\text b}^{n} + \mathcal P_{\text p}^{n}\right)
	}{\dt} = 
	-\mathcal I_{P,\text r}|^{n+\frac12}
	-\mathcal I_{P,\text b}|^{n+\frac12}
	-\mathcal I_{P,\text p}|^{n+\frac12}
	= 0
\end{equation}
and hence the total probability must stay constant throughout the simulation.
Similarly, one can show an analogous result for the sum of the total energies
\begin{equation}
	\label{eq:ne-tunneling-roc-ener}
	\dfrac{
		\left(\mathcal H_{\text r}^{n+1} + \mathcal H_{\text b}^{n+1} + \mathcal H_{\text p}^{n+1}\right) 
		- 
		\left(\mathcal H_{\text r}^{n} + \mathcal H_{\text b}^{n} + \mathcal H_{\text p}^{n}\right) 
	}{\dt} 
	= s_{\text r}^{n+\frac12} + s_{\text b}^{n+\frac12} + s_{\text p}^{n+\frac12} 
	= 0\,.
\end{equation}
Thus, the total energy must likewise stay constant.
In order to assess these predictions, we compute the relative deviation of the probability and energy as
\begin{equation}
	\label{eq:ne-tunneling-rel-dev-from-init-prob-proposed}
	\left| \left(\mathcal P_{\text r}^{n} + \mathcal P_{\text b}^{n} + \mathcal P_{\text p}^{n} \right) 
	- 
	\left(\mathcal P_{\text r}^{0} + \mathcal P_{\text b}^{0} + \mathcal P_{\text p}^{0}\right) \right|\,, \quad n = 0, 10, 20, \hdots \,,
\end{equation} 
\begin{equation}
	\label{eq:ne-tunneling-rel-dev-from-init-ener-proposed}
	\dfrac{1}{\mathcal H(n\dt)}\left| \left(\mathcal H_{\text r}^{n} + \mathcal H_{\text b}^{n} + \mathcal H_{\text p}^{n} \right) 
	- 
	\left(\mathcal H_{\text r}^{1} + \mathcal H_{\text b}^{1} + \mathcal H_{\text p}^{1}\right) \right|\,, \quad n = 1, 11, 21, \hdots \,,
\end{equation} 
where $\mathcal H(t)$ is the energy associated with the analytical solution and is equal to 77.5~meV. The values of \eqref{eq:ne-tunneling-rel-dev-from-init-prob-proposed} and \eqref{eq:ne-tunneling-rel-dev-from-init-ener-proposed} are shown in Fig.~\ref{fig:ne-tunneling-rel-dev-from-init}(a) and Fig.~\ref{fig:ne-tunneling-rel-dev-from-init}(b). The largest values of \eqref{eq:ne-tunneling-rel-dev-from-init-prob-proposed} and \eqref{eq:ne-tunneling-rel-dev-from-init-ener-proposed} were 4.13$\times$10$^{-14}$ and 4.16$\times$10$^{-14}$, respectively. As expected, these fluctuations are extremely small and can be attributed to the finite machine precision. 
For comparison, we also compute 
\begin{equation}
	\label{eq:ne-tunneling-rel-dev-from-init-prob-simple}
	\left| \left(
	\mathcal P_{\text r, \text{simple}}^{n} 
	+ \mathcal P_{\text b, \text{simple}}^{n} 
	+ \mathcal P_{\text p, \text{simple}}^{n} 
	\right) 
	- 
	\left(
	\mathcal P_{\text r, \text{simple}}^{0} 
	+ \mathcal P_{\text b, \text{simple}}^{0} 
	+ \mathcal P_{\text p, \text{simple}}^{0}
	\right) \right|\,, \quad n = 0, 10, 20, \hdots \,,
\end{equation} 
\begin{equation}
	\label{eq:ne-tunneling-rel-dev-from-init-ener-simple}
	\dfrac{1}{\mathcal H(n\dt)}\left| \left(
	\mathcal H_{\text r, \text{simple}}^{n} 
	+ \mathcal H_{\text b, \text{simple}}^{n} 
	+ \mathcal H_{\text p, \text{simple}}^{n} 
	\right) 
	- 
	\left(
	\mathcal H_{\text r, \text{simple}}^{1} 
	+ \mathcal H_{\text b, \text{simple}}^{1} 
	+ \mathcal H_{\text p, \text{simple}}^{1}
	\right) \right|\,, \quad n = 1, 11, 21, \hdots \,,
\end{equation} 
which are also shown in Fig.~\ref{fig:ne-tunneling-rel-dev-from-init}(a) and Fig.~\ref{fig:ne-tunneling-rel-dev-from-init}(b). 
The largest values of \eqref{eq:ne-tunneling-rel-dev-from-init-prob-simple} and \eqref{eq:ne-tunneling-rel-dev-from-init-ener-simple} are, respectively, 3.92$\times$10$^{-3}$ and 4.21$\times$10$^{-3}$, meaning that the simple expressions do not abide the conservation of probability and energy exactly. 



\begin{figure}
	\centering
	\includegraphics[scale=0.75]{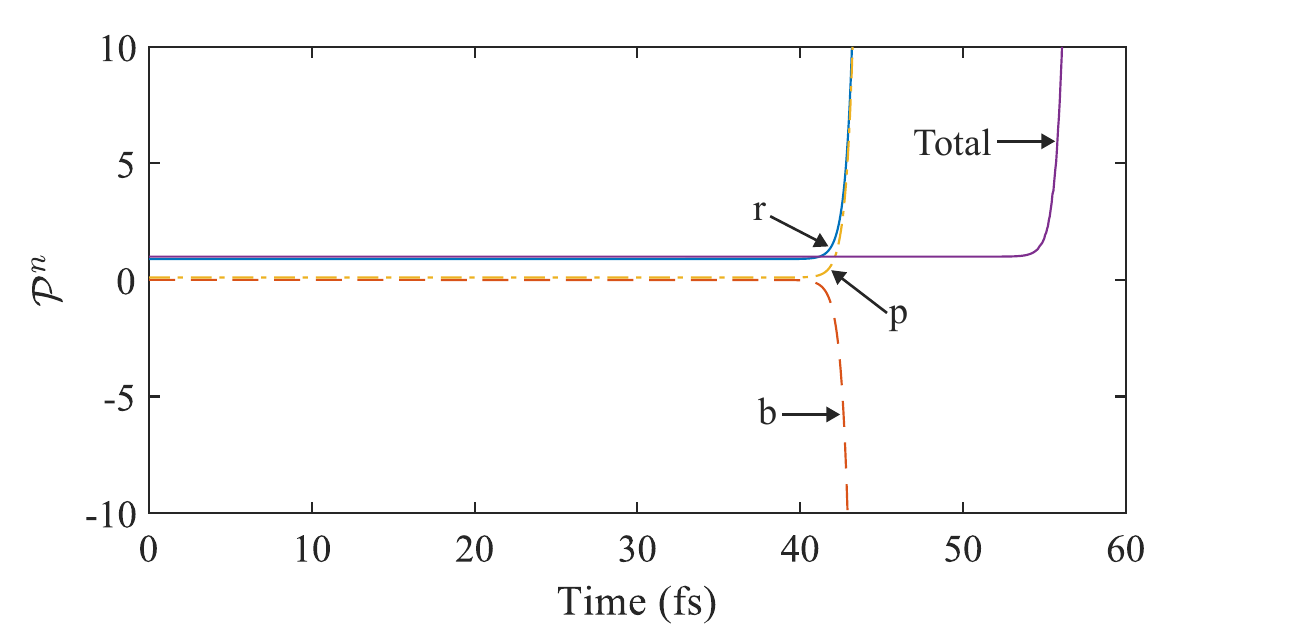}
	\caption{Results of the test in Section~\ref{sec:ne-tunneling} for the time step exceeding the CFL limit ${\dt_{\text{CFL}}^{\text b}}$ and the generalized CFL limit ${\dt_{\text{CFL, gen}}^{\text b}}$ in the barrier region.}
	\label{fig:ne-tunneling-results-unstable}
\end{figure}

In order to illustrate a possible mechanism for the breakdown of conservation properties beyond the CFL limit, we also run a simulation with a time step that violates the CFL limit condition in the barrier region ($\dt = 1.005 \dt_{\text{CFL}}^{\text b} = 0.9624 \dt_{\text{CFL}}^{\text r} = 0.9624 \dt_{\text{CFL}}^{\text p}$). The resulting probability is shown in Fig.~\ref{fig:ne-tunneling-results-unstable}. Beyond the CFL limit, the inequality \eqref{eq:bound-on-P} can be violated in the barrier region and the probability associated with that region is permitted to indefinitely grow negative. This allows the barrier region to provide an infinite amount of probability to the reactant and product regions, causing their probabilities to indefinitely grow positive. The total probability remains unchanged, until the exponential growth of the wavefunction values due to instability eventually causes the the total probability to deviate from 1, as a result of the finite machine precision.

\section{Conclusions}
\label{sec:conslusions}

In this work, we investigated the probability and energy conservation properties of the finite-difference time-domain scheme for solving the Schrödinger equation. Existing works on the conservation properties of numerical schemes for the Schrödinger equation consider a region that is terminated with periodic or zero Dirichlet boundary conditions, which do not allow any net exchange of probability and energy with the surrounding space. In contrast, we considered a general scenario where a region can either be terminated using boundary conditions or form a portion of a larger domain, where probability and energy may enter or leave the region through the boundary. 
We introduced modified equations on the boundary of the region in order to write a self-contained model that allows analyzing the properties of the region without making assumptions on the nature of the discretization beyond the boundary of the region.
We proposed expressions for the total probability and energy contained in a region discretized using FDTD-Q, as well as expressions for the probability current and supplied power. 
Using these expressions, we showed that the FDTD-Q method conserves probability under the CFL limit that has been traditionally used for ensuring stability. We provided an illustration of the mechanism in which a violation of the CFL condition can result in violation of the conservation of probability for an example involving three connected FDTD regions. Furthermore, we have shown that the CFL condition ensures that the energy is conserved, under an assumption that the region is coupled to other probability-conserving models.

The proposed expressions for computing probability and energy were compared to a more straightforward approach and several advantages were found. First, the proposed expressions respect the probability and energy conservation exactly. 
The exact conservation properties avoid spurious fluctuations exhibited by the values of the simpler expressions, which are especially evident when a system is isolated. Second, the proposed expressions for probability and energy tend to be slightly more accurate than the simple expressions. Lastly, in isolated systems, the value of the proposed probability is convenient for normalizing initial values, since it is guaranteed to stay constant.

The proposed theory sheds light on the energy and probability conservation in situations where the FDTD-Q model of the region can exchange these quantities with the surrounding space. This insight can serve as a basis for a stability analysis and enforcement framework in scenarios where the region is coupled to other models that can exchange energy and/or probability with the region. As a proof of concept, we considered the case of multiple regions with the same FDTD-Q discretization that are coupled to each other. We envision other possible applications, such as creation of stable subgridding schemes~\cite{salehi-2020-jce-polar,okoniewski-3d-subgridding-1997}, stable incorporation of reduced order models~\cite{kulas-macromodels-2004,jnl-tap-2018-fdtd-mor},
and stability analysis of advanced boundary conditions~\cite{sullivan-2005-jap-efncs-nanoscale,nagel-2009-aces,zhidong-2009-js-abc,decleer-2021-jcam}. Moreover, multi-physics simulations~\cite{lopata-2009-jcp-multiscale-ms,yao-2015-cma-leapfrog-ms-metamaterials,chen-2017-jcp-canonical,chen-2017-cpc-unified,xie-2021-universal-vec-scal} can also be considered as a type of scenario where an exchange occurs between the energy associated with the quantum particle and the energy stored in another form, such as the energy of electromagnetic field. 
Hence, the approach proposed in this paper could be very convenient if extended to such scenarios, allowing to prove stability by ensuring that each part of the system conserves the quantities of interest.

\appendix
\appendixpage

\section{Indexing convention and expressions for the matrices in Section~\ref{sec:cmf}}
\label{appendix:indexing-and-cmf-expressions}

\begin{table}\centering
	\caption{Indexing convention for vectors and matrices in Section~\ref{sec:cmf}.}\label{tab:indexing-convention}
	\begin{tabular}{|c|c|c|c|}\hline
		Sample description & Scalar index & Vector index & Example(s)\\\hline\hline
		Primary nodes & $(i,j,k)$ & $i + (j-1)(n_x+1) + (k-1)(n_x+1)(n_y+1)$ & $\PSI_R^n$, $\PSI_I^{n-0.5}$ in \eqref{eq:cmf-pre-a}\\\hline\hline
		\multicolumn{4}{|c|}{Primary edges} \\ \hline\hline
		$x$-directed& $(i+0.5,j,k)$ & $i+(j-1)n_x + (k-1)n_x (n_y+1)$ & $\partial_x \PSI_R^n$ in \eqref{eq:structure-gradPsiR}\\\hline
		$y$-directed & $(i,j+0.5,k)$ & $i+(j-1)(n_x+1) + (k-1)(n_x+1) n_y$ & $\partial_y \PSI_R^n$ in \eqref{eq:structure-gradPsiR}\\\hline
		$z$-directed & $(i,j,k+0.5)$ & $i+(j-1)(n_x+1) + (k-1)(n_x+1) (n_y+1)$ & $\partial_z \PSI_R^n$ in \eqref{eq:structure-gradPsiR}\\\hline\hline
		\multicolumn{4}{|c|}{Additional primary edges normal to the boundary} \\ \hline\hline
		$x$-directed, W & $(1,j,k)$ & $j + (k-1)(n_y+1)$ & $[\partial_x \PSI_I]_{\text{W}}^{n+0.5}$ in \eqref{eq:gradPsiIbot-struct}\\\hline
		$x$-directed, E & $(n_x+1,j,k)$ & $j + (k-1)(n_y+1)$ & $[\partial_x \PSI_I]_{\text{E}}^{n+0.5}$  in \eqref{eq:gradPsiIbot-struct}\\\hline
		$y$-directed, S & $(i,1,k)$ & $i + (k-1)(n_x+1)$ & $[\partial_y \PSI_I]_{\text{S}}^{n+0.5}$ in \eqref{eq:gradPsiIbot-struct}\\\hline
		$y$-directed, N & $(i,n_y+1,k)$ & $i + (k-1)(n_x+1)$ & $[\partial_y \PSI_I]_{\text{N}}^{n+0.5}$ in \eqref{eq:gradPsiIbot-struct}\\\hline
		$z$-directed, B & $(i,j,1)$ & $i + (j-1)(n_x+1)$ & $[\partial_z \PSI_I]_{\text{B}}^{n+0.5}$ in \eqref{eq:gradPsiIbot-struct}\\\hline
		$z$-directed, T & $(i,j,n_z+1)$ & $i + (j-1)(n_x+1)$ & $[\partial_z \PSI_I]_{\text{T}}^{n+0.5}$ in \eqref{eq:gradPsiIbot-struct}\\\hline
	\end{tabular}
\end{table}

This appendix provides expressions for the matrices in Section~\ref{sec:cmf}. These expressions assume a specific order of the samples in the vectors in Section~\ref{sec:cmf}. In particular, vectors of quantities sampled at primary nodes, such as $\PSI_R^{n}$ or $\PSI_I^{n-0.5}$ in \eqref{eq:cmf-pre-a} have the elements ordered based on the indexing convention in Table~\ref{tab:indexing-convention}. For example, a sample $\psi_R|_{i,j,k}^n$ would be placed in the element $i + (j-1)(n_x+1) + (k-1)(n_x+1)(n_y+1)$ in $\PSI_R^n$. With this convention, matrix $\volSec$ in \eqref{eq:cmf-pre-a} is defined as
\begin{equation}
	\volSec = \dx \dy \dz \ \tilde \I_{n_z+1} \otimes \tilde \I_{n_y+1} \otimes \tilde \I_{n_x+1}\,,
\end{equation}
where $\tilde \I_m$ is an $m \times m$ matrix given by
\begin{equation}
	\tilde \I_m = \operatorname{diag}\left(\frac12, 1, 1, \hdots, 1, \frac 12\right)\,.
\end{equation}
Diagonal matrix $\diagon{U}$ contains the samples $U|_{i,j,k}$ placed on the diagonal elements according to the indexing convention in Table~\ref{tab:indexing-convention}.

Vectors with samples on the primary edges, such as $\nabla \PSI_R^n$ in \eqref{eq:gradPsiR}, contain, in that order, the samples on the $x$-directed, $y$-directed, and $z$-directed edges. 
\begin{equation}\label{eq:structure-gradPsiR}
	\nabla \PSI_R^n = \begin{bmatrix}
		\partial_x \PSI_R^n\\
		\partial_y \PSI_R^n\\
		\partial_z \PSI_R^n
	\end{bmatrix}\,,
\end{equation}
where the ordering of samples in $\partial_x \PSI_R^n$, $\partial_y \PSI_R^n$, and $\partial_z \PSI_R^n$ is shown in Table~\ref{tab:indexing-convention}. With this, $\D$ is defined as
\begin{equation}
	\D = \begin{bmatrix}
		\D_x  & \D_y & \D_z
	\end{bmatrix}\,,
\end{equation}
where
\begin{equation}		
	\D_x = - \I_{n_z+1} \otimes \I_{n_y+1} \otimes \W_{n_x}^T\,,\quad
	\D_y = - \I_{n_z+1} \otimes \W_{n_y}^T \otimes \I_{n_x+1}\,,\quad
	\D_z = - \W_{n_z}^T \otimes \I_{n_y+1} \otimes \I_{n_x+1}\,,
\end{equation}
where $\W_m$ is an $m \times (m+1)$ matrix given by
\begin{equation}
	\W_m = 
	\begin{bmatrix}
		\mathbf 0_{m\times 1} & \I_m
	\end{bmatrix}
	-\begin{bmatrix}
		\I_m & \mathbf 0_{m\times 1}
	\end{bmatrix}.
\end{equation}
Matrices $\areaSec$ and $\lenPrim$ are given by
\begin{align}
	\areaSec &= \operatorname{diag}\left(
	\dy \dz\ \tilde \I_{n_z+1} \otimes \tilde \I_{n_y+1} \otimes \I_{n_x}, \
	\dx \dz\ \tilde \I_{n_z+1} \otimes \I_{n_y} \otimes \tilde \I_{n_x+1},\
	\dx \dy\ \I_{n_z} \otimes \tilde \I_{n_y+1} \otimes \tilde \I_{n_x+1}\right)\,,
	\\
	\lenPrim &= \operatorname{diag}\left(
	\dx\ \I_{n_z} \otimes \I_{n_y} \otimes \I_{n_z}, \ 
	\dy\ \I_{n_z} \otimes \I_{n_y} \otimes \I_{n_z}, \ 
	\dz\ \I_{n_z} \otimes \I_{n_y} \otimes \I_{n_z}
	\right)\,.
\end{align}

The hanging variables corresponding to the six faces of the boundary are ordered as follows: west, east, south, north, bottom, top. For example $[\nabla \PSI_I]_{\bot}^{n+0.5}$ in \eqref{eq:cmf-pre-a} has the following structure
\begin{equation}\label{eq:gradPsiIbot-struct}
[\nabla \PSI_I]_{\bot}^{n+\frac12} = 
\begin{bmatrix}
	[\partial_x \PSI_I]_{\text{W}}^{n+\frac12}\\
	[\partial_x \PSI_I]_{\text{E}}^{n+\frac12}\\
	[\partial_y \PSI_I]_{\text{S}}^{n+\frac12}\\
	[\partial_y \PSI_I]_{\text{N}}^{n+\frac12}\\
	[\partial_z \PSI_I]_{\text{B}}^{n+\frac12}\\
	[\partial_z \PSI_I]_{\text{T}}^{n+\frac12}
\end{bmatrix}\,,
\end{equation}
where the indexing convention of each of the vectors on the right hand side of \eqref{eq:gradPsiIbot-struct} is shown in Table~\ref{tab:indexing-convention}.
Based on this ordering of the hanging variables, $\Lbot$ is defined as
\begin{align}
	\Lbot &= \begin{bmatrix}
		\mathbf L_{\text W} & 
		\mathbf L_{\text E} & 
		\mathbf L_{\text S} & 
		\mathbf L_{\text N} & 
		\mathbf L_{\text B} & 
		\mathbf L_{\text T}
	\end{bmatrix}\,,
\end{align}
where
\begin{equation}
\begin{aligned}
	&\mathbf L_{\text W} = \I_{n_z+1} \otimes \I_{n_y+1} \otimes \mathbf e\{1,n_x+1\}\,, \quad
	&\mathbf L_{\text E} = \I_{n_z+1} \otimes \I_{n_y+1} \otimes \mathbf e\{n_x+1,n_x+1\}\,,\\
	&\mathbf L_{\text S} = \I_{n_z+1} \otimes \mathbf e\{1,n_y+1\} \otimes \I_{n_x+1} \,,\quad
	&\mathbf L_{\text N} = \I_{n_z+1} \otimes \mathbf e\{n_y+1,n_y+1\} \otimes \I_{n_x+1} \,,\\
	&\mathbf L_{\text B} = \mathbf e\{1,n_z+1\} \otimes \I_{n_y+1} \otimes \I_{n_x+1} \,,\quad
	&\mathbf L_{\text T} = \mathbf e\{n_z+1,n_z+1\} \otimes \I_{n_y+1} \otimes \I_{n_x+1}\,,
\end{aligned}
\end{equation}
where $\mathbf e\{p,m\}$ a vector of size $m \times 1$ with 1 in position $p$ and zeroes in all other positions. Similarly, matrix $\ndot$ is given by
\begin{equation}
	\ndot = \operatorname{diag}\left(
	\diagon{(\hat n \cdot)\text W},
	\diagon{(\hat n \cdot)\text E},
	\diagon{(\hat n \cdot)\text S},
	\diagon{(\hat n \cdot)\text N},
	\diagon{(\hat n \cdot)\text B},
	\diagon{(\hat n \cdot)\text T}
	\right)\,,
\end{equation}
where
\begin{equation}
\begin{aligned}
	&\diagon{(\hat n \cdot)\text W} = [\hat n_{\text W} \cdot \hat x]\ \I_{n_z+1} \otimes \I_{n_y+1}\,,
	&\diagon{(\hat n \cdot)\text E} = [\hat n_{\text E} \cdot \hat x]\ \I_{n_z+1} \otimes \I_{n_y+1}\,, \\
	&\diagon{(\hat n \cdot)\text S} = [\hat n_{\text S} \cdot \hat y]\ \I_{n_z+1} \otimes \I_{n_x+1}\,,
	&\diagon{(\hat n \cdot)\text N} = [\hat n_{\text N} \cdot \hat y]\ \I_{n_z+1} \otimes \I_{n_x+1} \,,\\
	&\diagon{(\hat n \cdot)\text B} = [\hat n_{\text B} \cdot \hat z]\ \I_{n_y+1} \otimes \I_{n_x+1}\,,
	&\diagon{(\hat n \cdot)\text T} = [\hat n_{\text T} \cdot \hat z]\ \I_{n_y+1} \otimes \I_{n_x+1} \,,
\end{aligned}
\end{equation}
where $[\hat n_{\text W} \cdot \hat x] = [\hat n_{\text S} \cdot \hat y] = [\hat n_{\text B} \cdot \hat z]= -1$ and $[\hat n_{\text E} \cdot \hat x] = [\hat n_{\text N} \cdot \hat y] = [\hat n_{\text T} \cdot \hat z] = 1$. Matrix $\areaSecBndry$ is given by
\begin{equation}
	\areaSecBndry = \operatorname{diag}\left(
	\diagon{S,b,\text W}'', 
	\diagon{S,b,\text E}'', 
	\diagon{S,b,\text S}'', 
	\diagon{S,b,\text N}'', 
	\diagon{S,b,\text B}'', 
	\diagon{S,b,\text T}''
	\right)\,,
\end{equation}
where
\begin{equation}
\begin{aligned}
	&\diagon{S,b,\text W}'' = \dy \dz \ \tilde{\I}_{n_z+1} \otimes \tilde{\I}_{n_y+1}\,,
	&\diagon{S,b,\text E}'' = \dy \dz \ \tilde{\I}_{n_z+1} \otimes \tilde{\I}_{n_y+1} \,,\\
	&\diagon{S,b,\text S}'' = \dx \dz \ \tilde{\I}_{n_z+1} \otimes \tilde{\I}_{n_x+1}\,,
	&\diagon{S,b,\text N}'' = \dx \dz \ \tilde{\I}_{n_z+1} \otimes \tilde{\I}_{n_x+1} \,,\\
	&\diagon{S,b,\text B}'' = \dx \dy \ \tilde{\I}_{n_y+1} \otimes \tilde{\I}_{n_x+1}\,,
	&\diagon{S,b,\text T}'' = \dx \dy \ \tilde{\I}_{n_y+1} \otimes \tilde{\I}_{n_x+1} \,.
\end{aligned}
\end{equation}
\section{Proof of Theorem~\ref{thm:cfl<=cflgen}}
\label{appendix:proof-CFL<=CFLgen}

The proof draws inspiration from the approach in \cite{edelvik2004general}, where positive definiteness was shown for a matrix analogous to $\Pbig$, but arising from FDTD for Maxwell's equations and associated with stored energy. In particular, in \cite{edelvik2004general}, the contribution of individual primary cells to a quadratic form analogous to \eqref{eq:probability} was considered in order to show that the conventional CFL limit for FDTD for Maxwell's equations guarantees positive definiteness of the matrix associated with the entire region.

\begin{lemma}
	\label{lemma:probability-sum-over-cells}
	Consider a region described by FDTD-Q equations \eqref{eq:cmf-a}--\eqref{eq:cmf-b} with $n_x \times n_y \times n_z$ primary cells. Let $\Pbig$ be the matrix corresponding to the entire region. Let $\PSIbig$ be given by 
	\begin{equation}
		\PSIbig = \begin{bmatrix} \PSI_R \\ \PSI_I \end{bmatrix}\,,
	\end{equation}
	where $\PSI_R, \PSI_I \in \mathbb{R}^{(n_x+1)(n_y+1)(n_z+1) \times 1}$ are arbitrary vectors with each element corresponding to a primary node in the region. Let $\Pbig_{ijk}$ be the matrix defined in the same way as $\Pbig$ but for a single-cell region formed by an individual primary cell with the bottom-south-west corner at the node $(i,j,k)$ and the top-north-east corner at the node $(i+1,j+1,k+1)$. Let $\PSIbig_{ijk}$ be a vector formed by selecting the elements of $\PSIbig$ that correspond to the nodes on the corners of that cell. Then, 
	\begin{equation}
		\label{eq:probability-sum-over-cells}
		\PSIbig^T \Pbig \PSIbig = \sum_{i=1}^{n_x}\sum_{j=1}^{n_y}\sum_{k=1}^{n_z} \PSIbig_{ijk}^T \Pbig_{ijk} \PSIbig_{ijk}\,.
	\end{equation}
\end{lemma}
The proof involves expanding both sides of \eqref{eq:probability-sum-over-cells} as a summation over primary nodes and primary edges and collecting terms on the right hand side associated with the same primary node or primary edge. Then, each term on the left hand side can be shown to have a corresponding term on the right hand side and vice versa.

\begin{lemma}
	\label{lemma:minCFLgenCells<=CFLgen}
	Let $\dt_{\text{CFL,gen}}$ and $\dt_{\text{CFL,gen}}^{(i,j,k)}$ be the generalized CFL limits corresponding, respectively, to the entire region and to a single-cell region formed by the primary cell with the bottom-south-west corner at the node $(i,j,k)$. Then,
	\begin{equation}
		\label{eq:minCFLgenCells<=CFLgen}
		\min_{i,j,k}\dt_{\text{CFL,gen}}^{(i,j,k)} \le \dt_{\text{CFL,gen}}\,.
	\end{equation}
\end{lemma}
\begin{proof}
	Consider matrices $\Pbig$ and $\Pbig_{ijk}$ and vectors $\PSIbig$ and $\PSIbig_{ijk}$ as described in Lemma~\ref{lemma:probability-sum-over-cells}.
	Consider a time step $\dt$ such that 
	\begin{equation}
		\dt < \min_{i,j,k} \dt_{\text{CFL,gen}}^{(i,j,k)}  \,.
	\end{equation}
	With this time step, by Lemma~\ref{lemma:Ppd<=>dt<CFLgen} applied to the single-cell regions, 
	\begin{equation}
		\Pbig_{ijk} \succ 0 , \quad 
		\forall i \in {1, 2, \hdots, n_x}, \
		\forall j \in {1, 2, \hdots, n_y}, \
		\forall k \in {1, 2, \hdots, n_z}\,,
	\end{equation}
	and for each $i$, $j$, and $k$, we have
	\begin{equation}
		\begin{cases}
			\PSIbig_{ijk}^T \Pbig_{ijk} \PSIbig_{ijk} > 0, & \forall \PSIbig_{ijk} \in \mathbb{R}^{16\times 1}~\text{where}~\PSIbig_{ijk}\ne \mathbf 0\\
			\PSIbig_{ijk}^T \Pbig_{ijk} \PSIbig_{ijk} = 0, & \text{if}~\PSIbig_{ijk} = \mathbf 0
		\end{cases}\,.
	\end{equation}
	By Lemma~\ref{lemma:probability-sum-over-cells}, 
	\begin{equation}
		\PSIbig^T \Pbig \PSIbig > 0 \,,
		\quad 
		\forall \PSIbig \in \mathbb{R}^{2(n_x+1)(n_y+1)(n_z+1)\times 1}~\text{with}~\PSIbig\ne \mathbf 0\,,
	\end{equation}
	which means that $\Pbig$ is positive definite. Hence, by Lemma~\ref{lemma:Ppd<=>dt<CFLgen}, $\dt < \dt_{\text{CFL,gen}}$.
	
	In conclusion, the following implication holds
	\begin{equation}
		\dt < \min_{i,j,k} \dt_{\text{CFL,gen}}^{(i,j,k)} 
		\quad \implies \quad
		\dt < \dt_{\text{CFL,gen}}\,.
	\end{equation}
	This is only possible if \eqref{eq:minCFLgenCells<=CFLgen} is true. 		
\end{proof}

\begin{lemma}
	\label{lemma:CFL<=CFLgen-single-cell}
	Let $\dt_{\text{CFL}}^{(i,j,k)}$ and $\dt_{\text{CFL,gen}}^{(i,j,k)}$ be the CFL limit and the generalized CFL limit for a single-cell region composed of a primary cell with the bottom-south-west corner at the node $(i,j,k)$.
	Then, 
	\begin{equation}
		\label{eq:CFLsingleCell<=CFLgenSingleCell}
		\dt_{\text{CFL}}^{(i,j,k)} \le \dt_{\text{CFL,gen}}^{(i,j,k)}\,.
	\end{equation}
\end{lemma}
\begin{proof}
	With reference to Appendix~\ref{appendix:indexing-and-cmf-expressions}, 
	\begin{equation}
		\diagon{V,ijk}'' = \dfrac{\dx \dy \dz}{8} \I_8\,,
	\end{equation}
	\begin{equation}
		\diagon{S,ijk}'' = \operatorname{diag}\left(
		\dfrac{\dy \dz}{4} \I_4,
		\dfrac{\dx \dz}{4} \I_4,
		\dfrac{\dx \dy}{4} \I_4
		\right)\,,
	\end{equation}
	\begin{equation}
		\diagon{l,ijk}' = \operatorname{diag}\left(
		\dx\ \I_4,
		\dy\ \I_4,
		\dz\ \I_4
		\right)\,,
	\end{equation}
	\begin{equation}
		\D_{ijk} = -\begin{bmatrix}
			\I_2 \otimes \I_2 \otimes \W_1^T
			&
			\I_2 \otimes \W_1^T \otimes \I_2
			&
			\W_1^T \otimes \I_2 \otimes \I_2
		\end{bmatrix}\,,
	\end{equation}
	where subscripts $ijk$ indicate that a matrix corresponds to the single-cell region. Define matrix $\mathbf{\Sigma}_{ijk}$ as
	\begin{multline}
		\mathbf{\Sigma}_{ijk} = \dfrac{1}{\hbar} (\diagon{V,ijk}'')^{-\frac12} \HH_{ijk} (\diagon{V,ijk}'')^{-\frac12}
		= 
		\dfrac{\hbar}{2m} (\diagon{V,ijk}'')^{-\frac12} \D_{ijk} \areaSec (\diagon{l,ijk}')^{-1} \D_{ijk}^T (\volSec)^{-\frac12}
		+
		\dfrac{1}{\hbar} \diagon{U,ijk}
		\\
		=\dfrac{\hbar}{m} 
		\left[
		\dfrac{1}{(\dx)^2}(\I_2 \otimes \I_2 \otimes \W_1^T \W_1)
		+ \dfrac{1}{(\dy)^2}(\I_2 \otimes \W_1^T \W_1 \otimes \I_2) 
		+ \dfrac{1}{(\dz)^2}(\W_1^T \W_1 \otimes \I_2 \otimes \I_2)
		\right]
		+ \dfrac{1}{\hbar} \diagon{U,ijk} \,.
	\end{multline}
	Let $\mathbf V$ be the following matrix:
	\begin{equation}
		\mathbf V = 
		\begin{bmatrix}
			\mathbf v_0 & \mathbf v_x & \mathbf v_y & \mathbf v_z & \mathbf v_{yz} & \mathbf v_{xz} & \mathbf v_{xy} & \mathbf v_{xyz}
		\end{bmatrix}\,,
	\end{equation}
	where
	\begin{equation}
		\begin{aligned}
			&\mathbf v_{0} = (\sqrt 8)^{-1} |\W_1^T| \otimes |\W_1^T| \otimes |\W_1^T|\,,
			&\mathbf v_{x} = (\sqrt 8)^{-1} |\W_1^T| \otimes |\W_1^T| \otimes \W_1^T\,,\\
			&\mathbf v_{y} = (\sqrt 8)^{-1} |\W_1^T| \otimes \W_1^T \otimes |\W_1^T|\,,
			&\mathbf v_{z} = (\sqrt 8)^{-1} \W_1^T \otimes |\W_1^T| \otimes |\W_1^T|\,,\\
			&\mathbf v_{yz} = (\sqrt 8)^{-1} \W_1^T \otimes \W_1^T \otimes |\W_1^T|\,,
			&\mathbf v_{xz} = (\sqrt 8)^{-1} \W_1^T \otimes |\W_1^T| \otimes \W_1^T\,,\\
			&\mathbf v_{xy} = (\sqrt 8)^{-1} |\W_1^T| \otimes \W_1^T \otimes \W_1^T\,,
			&\mathbf v_{xyz} = (\sqrt 8)^{-1} \W_1^T \otimes \W_1^T \otimes \W_1^T\,.
		\end{aligned}
	\end{equation}	
	Vectors $|\W_1^T|$ and $\W_1^T$ are eigenvectors of $\W_1^T \W_1$ with eigenvalues 0 and 2, respectively. Using this fact, 
	\begin{multline}
		\mathbf{\Sigma}_{ijk} \mathbf V =
		\dfrac{2\hbar}{m}\mathbf V\operatorname{diag}\bigg(
		0, 
		\frac1{(\dx)^2}, \frac1{(\dy)^2}, \frac{1}{(\dz)^2}, 
		\frac1{(\dy)^2}\!+\!\frac{1}{(\dz)^2}, \frac1{(\dx)^2}\!+\!\frac{1}{(\dz)^2}, \frac1{(\dx)^2}\!+\!\frac{1}{(\dy)^2}, \\
		\frac1{(\dx)^2}\!+\!\frac1{(\dy)^2}\!+\!\frac{1}{(\dz)^2}
		\bigg) 
		+ \frac1{\hbar} \diagon{U,ijk} \mathbf V\,.
	\end{multline}
	Matrix $\mathbf V$ can be easily shown to satisfy $\mathbf V^T \mathbf V = \I = \mathbf V \mathbf V^T$.
	Thus, 
	\begin{multline}
		\mathbf{\Sigma}_{ijk} =
		\dfrac{2\hbar}{m}\mathbf V\operatorname{diag}\bigg(
		0, 
		\frac1{(\dx)^2}, \frac1{(\dy)^2}, \frac{1}{(\dz)^2}, 
		\frac1{(\dy)^2}\!+\!\frac{1}{(\dz)^2}, \frac1{(\dx)^2}\!+\!\frac{1}{(\dz)^2}, \frac1{(\dx)^2}\!+\!\frac{1}{(\dy)^2}, \\
		\frac1{(\dx)^2}\!+\!\frac1{(\dy)^2}\!+\!\frac{1}{(\dz)^2}
		\bigg)  \mathbf V^T
		+ \frac1{\hbar} \diagon{U,ijk}\,.
	\end{multline}
	Hence $\mathbf{\Sigma}_{ijk}$ is a sum of two symmetric matrices. The first matrix has the 2-norm equal to $(2\hbar/m)((\dx)^{-2}+(\dy)^{-2}+(\dz)^{-2})$. As a result~\cite{golub-matrix-comput-4ed},
	\begin{equation}
		\rho(\mathbf{\Sigma}_{ijk}) = ||\mathbf{\Sigma}_{ijk}||_2 \le 
		\dfrac{2\hbar}{m}\left(
		\frac1{(\dx)^2}+\frac1{(\dy)^2}+\frac{1}{(\dz)^2}
		\right)
		+ \frac1{\hbar} \left|\left|\diagon{U,ijk}\right|\right|_2
	\end{equation}
	and
	\begin{equation}
		\label{eq:proofDtCflSingleCell<=dtCflGenSingleCellQED}
		\dfrac{2}{\dfrac{2\hbar}{m}\left(
			\dfrac1{(\dx)^2}+\dfrac1{(\dy)^2}+\dfrac{1}{(\dz)^2}
			\right)
			+ \dfrac1{\hbar} \left|\left| \diagon{U,ijk}\right|\right|_2} \le \dfrac{2}{\rho\left(\mathbf{\Sigma}_{ijk}\right)}\,,
	\end{equation}
	where
	\begin{equation}
		\left|\left|\diagon{U,ijk}\right|\right|_2 = \max\left(
		\left| U|_{i,j,k} \right|, \left| U|_{i+1,j,k} \right|, 
		\left| U|_{i,j+1,k} \right|, \left| U|_{i+1,j+1,k} \right|,
		\left| U|_{i,j,k+1} \right|, \left| U|_{i+1,j,k+1} \right|, 
		\left| U|_{i,j+1,k+1} \right|, \left| U|_{i+1,j+1,k+1} \right|
		\right)\,.
	\end{equation}
	Inequality \eqref{eq:proofDtCflSingleCell<=dtCflGenSingleCellQED} proves \eqref{eq:CFLsingleCell<=CFLgenSingleCell} by the definition of the two time steps in \eqref{eq:CFLsingleCell<=CFLgenSingleCell}.
\end{proof}

\begin{table}\centering 
	\caption{Relative difference of the CFL limit \eqref{eq:cfl} and the generalized CFL limit \eqref{eq:cfl-gen} in the scenarios considered in Appendix~\ref{appendix:proof-CFL<=CFLgen}. }\label{tab:cfl-vs-cflgen}
	\begin{tabular}{|c|c|c|}\hline
		Potential &	$n_x = n_y = n_z = 1$ & $n_x = n_y = n_z = 10$\\\hline\hline
		$U|_{i,j,k} = \text 0$					&~2.49$\times$10$^{-16}$ 	&0 	\\\hline
		\hline
		$U|_{i,j,k} = \text{0.1 eV}$	&$-$1.84$\times$10$^{-16}$ 	&~1.47$\times$10$^{-15}$ 	\\\hline
		$U|_{i,j,k} = \text{1 eV}	$	&$-$1.81$\times$10$^{-16}$ 	&$-$1.81$\times$10$^{-16}$ 	\\\hline
		$U|_{i,j,k} = \text{10 eV}$		& ~1.91$\times$10$^{-16}$ 	&~9.56$\times$10$^{-16}$ 	\\\hline
		\hline
		$U|_{i,j,k}= - \text{0.1 eV}$	&$-$6.51$\times$10$^{-1}$ 	&$-$6.51$\times$10$^{-1}$ \\\hline
		$U|_{i,j,k}= - \text{1 eV}$		&$-$1.72$\times$10$^{-1}$ 	&$-$1.72$\times$10$^{-1}$ 	\\\hline
		$U|_{i,j,k}= - \text{10 eV}$	&$-$2.03$\times$10$^{-2}$ 	&$-$2.03$\times$10$^{-2}$ 	\\\hline
		\hline
		$0 < U|_{i,j,k} < \text{0.1 eV}$ (randomized)	&$-$5.13$\times$10$^{-2}$ 	&$-$9.18$\times$10$^{-2}$ 	\\\hline
		$0 < U|_{i,j,k} < \text{1 eV}$ (randomized)		&$-$8.05$\times$10$^{-2}$ 	&$-$4.96$\times$10$^{-2}$ 	\\\hline
		$0 < U|_{i,j,k} < \text{10 eV}$ (randomized)	&$-$1.15$\times$10$^{-2}$ 	&$-$1.01$\times$10$^{-2}$ 	\\\hline
		\hline
		$- \text{0.1 eV} < U|_{i,j,k} < 0$ (randomized)&$-$3.69$\times$10$^{-1}$ 	&$-$4.24$\times$10$^{-1}$ 	\\\hline
		$- \text{1 eV} < U|_{i,j,k} < 0$ (randomized) &$-$2.55$\times$10$^{-1}$ 	&$-$2.28$\times$10$^{-1}$ 	\\\hline
		$- \text{10 eV} < U|_{i,j,k} < 0$ (randomized)	&$-$3.52$\times$10$^{-2}$ 	&$-$3.04$\times$10$^{-2}$ 	\\\hline		
	\end{tabular}
\end{table}

\begin{corollary}
	\label{corollary:CFL<=minCFLgenCells}
	\begin{equation}
		\dt_{\text{CFL}} \le \min_{i,j,k} \dt_{\text{CFL,gen}}^{(i,j,k)}\,,
	\end{equation}
	where $\dt_{\text{CFL}}$ is the CFL limit for the entire region given by \eqref{eq:cfl}.
\end{corollary}
\begin{proof}
	\begin{equation}
		\label{eq:proof-step-CFL<=minCFLgenCells}
		\dt_{\text{CFL}} \le \dt_{\text{CFL}}^{(i,j,k)} \le \dt_{\text{CFL,gen}}^{(i,j,k)} \quad
		\forall i \in {1, 2, \hdots, n_x}, \
		\forall j \in {1, 2, \hdots, n_y}, \
		\forall k \in {1, 2, \hdots, n_z}\,.
	\end{equation}
	The statement of the corollary follows directly from \eqref{eq:proof-step-CFL<=minCFLgenCells}.
\end{proof}
\begin{proof}[Proof of Theorem~\ref{thm:cfl<=cflgen}.]
	Combining the statements of Lemma~\ref{lemma:minCFLgenCells<=CFLgen} and Corollary~\ref{corollary:CFL<=minCFLgenCells},
	\begin{equation}
		\dt_{\text{CFL}} \le \min_{i,j,k}\dt_{\text{CFL,gen}}^{(i,j,k)} \le \dt_{\text{CFL,gen}}\,,
	\end{equation}
	which proves the theorem.
\end{proof}

Table~\ref{tab:cfl-vs-cflgen} shows the relative difference between the two time steps in \eqref{eq:cfl} and \eqref{eq:cfl-gen}, computed as $(\dt_{\text{CFL}}-\dt_{\text{CFL,gen}})/\dt_{\text{CFL,gen}}$.
The cell dimensions were taken as $\dx = $~1~nm, $\dy$~=~2~nm, and $\dz = $~3~nm. The mass of the particle was that of an electron. As expected, for the regions consisting of a single cell with constant nonnegative potential $U$ at each of the eight nodes, the two time steps $\dt_{\text{CFL}}$ and $\dt_{\text{CFL, gen}}$ were the same, with small discrepancy due to machine precision.
Interestingly, the two time steps were also equal for the multi-cell regions in Table~\ref{tab:cfl-vs-cflgen} with constant nonnegative potential. However, when the potential either had a spacial variation or was negative, the two time steps deviated, although the difference tended to be small. In all regions where the CFL limit \eqref{eq:cfl} and the generalized CFL limit~\eqref{eq:cfl-gen} deviated, the CFL limit~\eqref{eq:cfl} had a lower value, confirming the statement of Theorem~\ref{thm:cfl<=cflgen}.
\section*{Acknowledgments}

The work was in part funded by the Natural Sciences and Engineering Research Council of Canada (NSERC) Discovery Grant Program [funding reference number RGPIN-2019-05060], in part by the Canada Research Chairs Program [funding reference number 950-232062], and in part by the School of Graduate Studies and The Edward S. Rogers Sr. Department of Electrical and Computer Engineering at the University of Toronto.
The authors also would like to thank Alison Okumura for performing preliminary investigations of the total probability and stability properties of FDTD-Q in one dimension.

\bibliographystyle{IEEEtran}
\bibliography{bibliography.bib}

\end{document}